\documentclass[a4paper,11pt]{article}
\pdfoutput=1 

\usepackage{jheppub} 

\usepackage[T1]{fontenc} 

\usepackage{float, array, xspace, amscd, amsthm}
\usepackage{fancyhdr}
\usepackage{longtable}
\newcommand{\nn}{\nonumber}
\newcommand{\dd}{{\rm d}}
\newcommand{\w}{\wedge}

\newtheorem{lemma}{Lemma}
\newtheorem{proposition}{Proposition}

\newcommand{\IR}{\mathbb{R}}

\newcommand{\be}{\begin{equation}}
\newcommand{\ee}{\end{equation}}
\def\bea#1\eea{\begin{align}#1\end{align}}

\usepackage{color}

\preprint{UUITP-11/14}

\title{
Exploring $SU(3)$ Structure Moduli Spaces with Integrable $G_2$ Structures}

\author[a]{Xenia de la Ossa,}
\author[b]{Magdalena Larfors,}
\author[a,c,d,e]{Eirik E.~Svanes}

\affiliation[a]{Mathematical Institute, Oxford University\\Andrew Wiles Building, Woodstock Road\\Oxford OX2 6GG, UK }
\affiliation[b]{Department of Physics and Astronomy,Uppsala University\\ SE-751 20 Uppsala, Sweden}
\affiliation[c]{Sorbonne Universit\'es, UPMC Univ. Paris 06, UMR 7589, LPTHE, F-75005, Paris, France}
\affiliation[d]{CNRS, UMR 7589, LPTHE, F-75005, Paris, France}
\affiliation[e]{Sorbonne Universit\'es, Institut Lagrange de Paris, 98 bis Bd Arago, 75014 Paris, France\\}

\emailAdd{delaossa@maths.ox.ac.uk, magdalena.larfors@physics.uu.se, esvanes@lpthe.jussieu.fr
}

\abstract{We study the moduli space of $SU(3)$ structure manifolds $X$ that form the internal compact spaces in four-dimensional $N=\textstyle{\frac{1}{2}}$ domain wall solutions of heterotic supergravity with flux. Together with the direction perpendicular to the four-dimensional domain wall, $X$ forms a non-compact 7-manifold $Y$ with torsionful $G_2$ structure. We use this $G_2$ embedding to explore how $X(t)$ varies along paths $C(t)$ in the $SU(3)$ structure moduli space. Our analysis includes the Bianchi identities which strongly constrain the flow. We show that requiring that the $SU(3)$ structure torsion is preserved along the path leads to constraints on the $G_2$ torsion and the embedding of $X$ in $Y$. Furthermore, we study flows along which the torsion classes of $X$ go from zero to non-zero values. In particular, we present evidence that the flow of half-flat $SU(3)$ structures may contain Calabi--Yau loci, in the presence of non-vanishing $H$-flux. }
\begin{document}

\maketitle
\flushbottom

\newpage
\section{Introduction}

Compactifications of string theory provide one of the most fruitful grounds for the study of the theory's formal and phenomenological aspects. String compactifications can also be used as a tool to study properties of compact manifolds, and has been instrumental in the study of Calabi--Yau manifolds. In particular, the demand that supersymmetry is preserved leads to severe constraints on the metric of the compactified space. In heterotic string theory, the conditions for $N=1$ supersymmetric, maximally symmetric, four-dimensional vacua have long been known: with vanishing flux $H$, the internal 6-manifold must be Calabi--Yau \cite{Candelas:1985en}, whereas a non-zero $H$-flux requires the internal geometry to be complex non-K\"ahler \cite{Strominger:1986uh,Hull:1986kz} (see \cite{Lust:1986ix,Dasgupta:1999ss,Ivanov:2000fg,Becker:2002sx,Becker:2003sh,Becker:2003yv,Gauntlett:2002sc,Cardoso:2002hd,Gauntlett:2003cy,Gran:2005wf,Becker:2006xp,Gran:2007kh,Ivanov:2009rh} for further discussions). These two types of 6-manifolds both allow a globally defined, nowhere vanishing spinor $\eta$, that can be used to decompose the ten-dimensional supercharge $\epsilon$ of heterotic supergravity into internal and external components, $
\epsilon = \eta \otimes \rho$. If the external spinor component $\rho$ is covariantly constant, it provides a four-dimensional supercharge and guarantees that the four-dimensional vacuum has $N=1$ supersymmetry.
 
All six-dimensional manifolds with one nowhere vanishing spinor have $SU(3)$ structure, and so their geometry is completely specified by a real two-form $\omega$ and a complex three-form $\Psi$ that need not be closed \cite{Hitchin:2000jd,chiossi}. The non-closure of  $\omega$ and $\Psi$ determines the intrinsic torsion of the geometry. In this language, Calabi--Yau manifolds correspond to torsion-free $SU(3)$ structures, and the spaces that are solutions to the Strominger system \cite{Strominger:1986uh,Hull:1986kz} have torsion components transforming in a particular irreducible $SU(3)$ representation \cite{Cardoso:2002hd,Gauntlett:2003cy}.\footnote{$SU(3)$ structure manifolds are also relevant for (supersymmetric) type II compactifications with flux, and reviews of this topic can be found in \cite{Grana:2005jc,Blumenhagen:2006ci,Koerber:2010bx}. A summary of the torsion constraints that also includes non-supersymmetric vacua can be found in section 2 of \cite{Larfors:2013zva}.}

If the restrictions on the torsion of the internal $SU(3)$ structure manifold are relaxed, the resulting vacuum will break supersymmetry. More general $SU(3)$ structure manifolds can thus provide interesting non-supersymmetric vacua of heterotic string theory, with the benefit that the $SU(3)$ structure guarantees an $N=1$ four-dimensional effective field theory description of the low-energy dynamics. A simple class of such vacua, that will be studied in this paper, are half-BPS domain wall solutions in four dimensions, that preserve $N=\textstyle{\frac{1}{2}}$ supersymmetry. As has been shown recently \cite{Gray:2012md}, the $N=\textstyle{\frac{1}{2}}$ supersymmetry constraints put very mild restrictions on the intrinsic torsion; for the most general $H$-flux preserving the symmetry of the ansatz, almost all torsion components can be balanced by the appropriate flux (for studies with restricted flux, see \cite{Lukas:2010mf,Klaput:2011mz,Klaput:2012vv,Klaput:2013nla} for a ten-dimensional perspective, and  \cite{Gurrieri:2004dt,Micu:2004tz,deCarlos:2005kh,Gurrieri:2007jg,Micu:2009ci} for four-dimensional studies). This heterotic set-up can thus be used to study the properties of many different types of $SU(3)$ structure manifolds.

\begin{figure}[htb]
\centering
\includegraphics[width=0.8\textwidth]{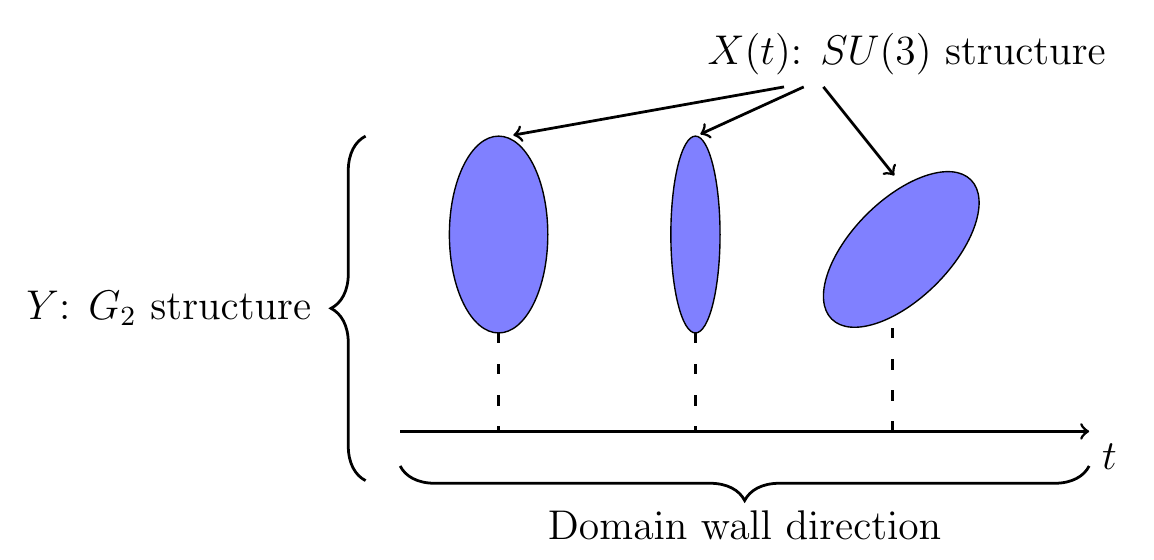}
\caption{The heterotic $N=\textstyle{\frac{1}{2}}$ domain wall solutions can be viewed as warped product of a compact $SU(3)$ structure 6-fold $X(t)$ with a four-dimensional domain wall spacetime, or as a the product of a non-compact $G_2$ manifold $Y$ with a maximally symmetric three-dimensional spacetime. The latter is suppressed in this picture.}
\label{fig:XY}
\end{figure}

The scope of this paper is to use heterotic $N=\textstyle{\frac{1}{2}}$ domain wall solutions to explore the moduli space of different $SU(3)$ structure manifolds. We restrict our study to the zeroth order $\alpha'$ approximation of heterotic string theory,\footnote{The heterotic gauge fields appear at first order in $\alpha'$, and are neglected in our study, as are the $\mathcal{O}(\alpha')$ corrections to the Bianchi identity of $H$. Note however that domain wall solutions avoid the usual no-go theorems for $H$-flux that appear at zeroth order in $N = 1$ solutions \cite{Ivanov:2000fg,Gauntlett:2003cy}, since the compactification is on a seven-dimensional {\it non-compact} manifold.} so that the bosonic part of the ten-dimensional supergravity action reduces to
\be
S = \frac{1}{2 \alpha'} \int \dd^{10} \, x \, e^{-2\phi} \sqrt{|G|} \left(R + 4(\partial \phi)^2 - \frac{1}{12} H^2  \right) \; ,
\ee
where $G$ is the ten-dimensional metric, $R$ the corresponding Ricci scalar, $\phi$ the dilaton, and $H = \dd B$ the flux of the Kalb--Ramond field $B$, satisfying the Bianchi identity $\dd H = 0$. We look for solutions where spacetime decomposes into a warped product of a compact $SU(3)$ structure manifold $X$ and a four-dimensional non-compact spacetime. The latter decomposes, for $N=\textstyle{\frac{1}{2}}$ domain wall solutions, into a maximally symmetric three-dimensional space along the domain wall world-volume and a direction perpendicular to the wall. Alternatively, as depicted in figure  \ref{fig:XY}, we may combine the direction perpendicular to the domain wall with the $SU(3)$ structure manifold to a non-compact seven-dimensional manifold $Y$ with $G_2$ structure \cite{Lukas:2010mf,Gray:2012md}. The spacetime metric thus has the form
\be
\dd s^2_{10} = \underbrace{\dd s^2_{3}}_\text{max. sym.} + \underbrace{N_t^2 (t, x) \dd^2 t + \underbrace{g_{mn}(t,x)\dd x^m \dd x^n}_\text{$SU(3)$ structure}}_\text{$G_2$ structure} \; ,
\ee
where the function $N_t$ and $SU(3)$ structure metric depend both on $t$ and the coordinates on $X$. For the $H$-flux, we allow all components that preserve the symmetries of the metric; $f \epsilon_{\alpha \beta \gamma}$ along the maximally symmetric three-dimensional spacetime, and $\widehat{H}$ along the seven-dimensional $G_2$ manifold.

The $N=\textstyle{\frac{1}{2}}$ supersymmetry constraints and the Bianchi identity for $H$ can, as we show in section \ref{sec:domainwall}, be reformulated as constraints on the intrinsic torsion of the $G_2$ structure. For the most general flux compatible with the spacetime symmetries, these conditions are met by a certain type of integrable $G_2$ structures. Consequently, the physical problem of finding $N=\textstyle{\frac{1}{2}}$ supersymmetric domain wall solutions in heterotic string theory translates into the mathematical problem of determining how the $SU(3)$ structure, dilaton and flux flow to form an integrable $G_2$ structure. The solutions are thus generalisations of Hitchin flow \cite{Hitchin:2000jd}, where a half-flat $SU(3)$ structure manifold varies over a perpendicular direction to form a seven-manifold with $G_2$ holonomy.  Indeed, the Hitchin flow is reproduced by the domain wall solutions when the flux is set to zero, and the dilaton is taken to be constant. Other types of $G_2$ flows have been discussed by Chiossi and Salamon \cite{chiossi}.
 
\begin{figure}[tb]
\centering
\includegraphics[width=0.6\textwidth]{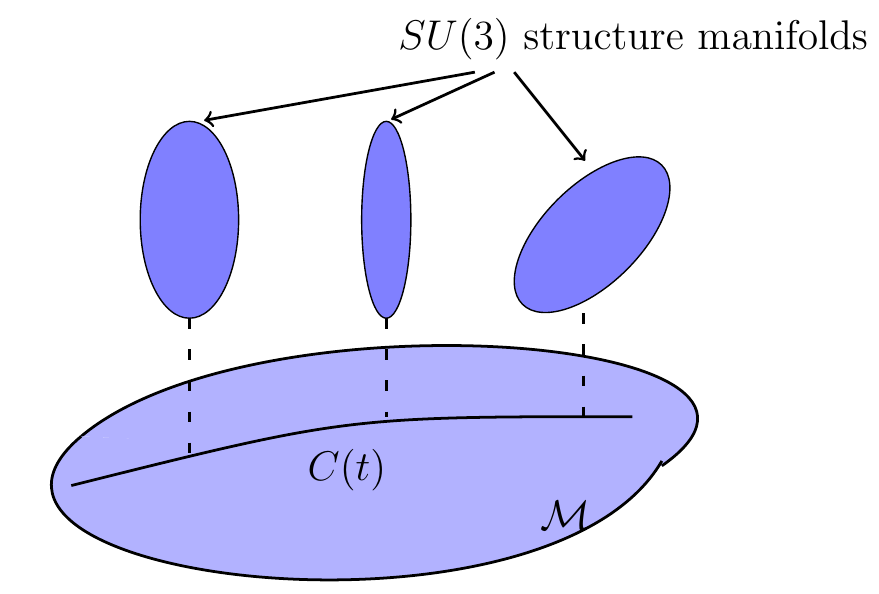}
\caption{The embedding of an $SU(3)$ structure  into an integrable $G_2$ structure induces a flow along a curve $C$ in the moduli space $\mathcal{M}$ of $SU(3)$ structure metrics, that is parameterised by the coordinate $t$.}
\label{fig:dw}
\end{figure}

As the $SU(3)$ structure flow along the domain wall direction, it will trace out a curve in the moduli space of $SU(3)$ structure manifolds, see figure \ref{fig:dw}. Consequently, through these constructions we uncover information about the parameter space of generic $SU(3)$ structure manifolds, which is highly non-trivial to study. For Calabi--Yau metrics, the parameters decompose into moduli corresponding to variations of the closed forms $\omega$ and $\Psi$: the former describe variations of the K\"ahler structure, and the latter variations of the complex structure, and the dimension of the two spaces is determined by the Hodge numbers $h^{(1,1)}$ and $h^{(2,1)}$ of the Calabi--Yau three-fold, respectively \cite{Candelas:1990pi}. On generic $SU(3)$ structure manifolds, neither $\omega$ nor $\Psi$ are closed, and the relevant parameters are only partially known (see \cite{Tseng:2009gr,Tseng:2010kt,Tseng:2011gv} for some recent progress). Indeed, care is needed when studying the moduli, since even very basic questions, such as the dimensionality of the parameter space, become subtle. As an example, the moduli space of deformations of $\omega$ in the Strominger system appears infinite-dimensional \cite{Becker:2006xp}, unless deformations of the $\mathcal{O}(\alpha')$-corrected Bianchi identity for $H$ are simultaneously taken into account \cite{Anderson:2014xha,delaOssa:2014cia}. 

In the present paper, the $G_2$ embedding constrains the allowed deformations, and thus simplifies the variational analysis. While this restriction means that questions regarding the number of parameters cannot be addressed, the setting is rich enough to tackle non-trivial questions such as the connections between the moduli spaces of Calabi--Yau and non-Calabi--Yau manifolds. With the most general $H$-flux, the heterotic constraints are solved by a wide range of $SU(3)$ structure manifolds, and different domain wall solutions may connect $SU(3)$ structures of different type. Through these constructions, we can thus study whether certain properties, such as an integrable complex structure, are preserved or violated by the flow, and whether torsion classes or flux components can be switched on and off by the flow. As a result, we will gain insight about the interconnections between the parameter spaces of different $SU(3)$ structure manifolds.

The rest of this paper is organised as follows. We start, in section \ref{sec:domainwall}, by presenting the integrable $G_2$ structures that are relevant for $N=\textstyle{\frac{1}{2}}$ domain wall solutions. We translate the supersymmetry conditions and Bianchi identities, which imply the equations of motion for the supergravity fields, into constraints on the $G_2$ torsion classes. In section \ref{sec:embedding}, we analyse the same equations from the perspective of the $SU(3)$ structure manifolds. We derive the constraints on the variations of the $SU(3)$ structure forms, and find that the flow of $\omega$ is completely determined by the torsion and flux, whereas some freedom remain in the variation of $\Psi$. In the following sections, we give examples of different types of $SU(3)$ structure flows that illustrate the intricacy of the moduli space of such manifolds. In section \ref{sec:CY}, we study the flow of Calabi--Yau manifolds with flux: we derive the constraints on the flow to preserve the Calabi--Yau properties, and perform a first-order analysis of the flow without these constraints that show that non-zero torsion is induced. In sections \ref{sec:NK} and \ref{sec:HF}, we study the flow of nearly K\"ahler and half-flat $SU(3)$ structures. We determine the constraints required to preserve the torsion along the flow, and study whether loci with vanishing torsion are possible. In section \ref{sec:conclusion} we summarise our results and discuss possible extensions of our analysis. Our conventions are summarised in appendix \ref{ap:conv}. In a separate paper \cite{OKS}, a complementary study of the embedding of the Strominger system into $G_2$ and Spin(7) manifolds will be presented.

\section{\texorpdfstring{$N=\textstyle{\frac{1}{2}}$}\ \  domain wall solution and \texorpdfstring{$G_2$}\ \ structures}
\label{sec:domainwall}

In this paper, we are interested in four-dimensional domain wall solutions of heterotic string theory that preserve $N=\textstyle{\frac{1}{2}}$ supersymmetry. Such configurations arise from heterotic compactifications on six-dimensional manifolds $X$ with $SU(3)$ structure. Alternatively, as shown in Figure \ref{fig:XY}, they can be viewed as three-dimensional maximally symmetric heterotic solutions that result from ``compactification'' on a non-compact seven-dimensional $G_2$ structure manifold $Y$, that is foliated by the $SU(3)$ six-manifolds. In this section, we describe how $N=\textstyle{\frac{1}{2}}$ supersymmetry determines the $G_2$ structure of $Y$. 

Any heterotic vacuum solution must satisfy the equations of motion and Bianchi identities of the low-energy supergravity description of the  theory. For supersymmetric solutions, the vanishing of the supersymmetry variations of fermionic fields lead to additional constraints. As described in detail in \cite{Lukas:2010mf,Gray:2012md}, to lowest order in the $\alpha'$ expansion, these constraints require the existence of a three-form $\varphi$ on $Y$, that must satisfy the  following constraints
\begin{align}
\dd_7\varphi &= 2\, \dd_7\phi\wedge\varphi -  *_7 \widehat H - f\, \psi~,\label{eq:susy1}\\
\dd_7\psi &= 2\, \dd_7\phi\wedge\psi~,\label{eq:susy2}\\
 *_7\dd_7\phi &= - \textstyle{\frac{1}{2}}\, \widehat H\wedge \varphi~,\label{eq:susy3}\\[1pt]
\textstyle{\frac{1}{2}}\,  *_7 f &= \widehat H\wedge\psi~,\label{eq:susy4}
\end{align}
where the three-form $\widehat H$ and the function $f$ are the components of the ten-dimensional flux $H = \dd B$, which lie along $Y$ and the three-dimensional, maximally symmetric domain wall world-volume, respectively.\footnote{The flux component $f$ determines the cosmological constant of the three-dimensional spacetime; a non-zero $f$ gives a anti de Sitter spacetime, while a vanishing $f$ gives Minkowsi spacetime.} $\phi$ is the ten-dimensional dilaton, and $\psi$ is the seven-dimensional Hodge dual of $\varphi$. $\dd_7$ and $*_7$ denote the exterior derivative and Hodge dual on $Y$, respectively (see appendix \ref{ap:conv} for further conventions).

 This system is equivalent to a $G_2$ structure \cite{FerGray82,chiossi} determined by $\varphi$
\begin{align}
\dd_7\varphi &= \tau_0\, \psi + 3\, \tau_1\wedge\varphi 
+  *_7 \tau_3~,\label{eq:one}\\
\dd_7\psi &= 4\, \tau_1\wedge\psi +  *_7\tau_2~.\label{eq:two}
\end{align}
with intrinsic torsion specified by  
\begin{align}
\tau_0 &= - \textstyle{\frac{15}{14}}\, f~, \label{eq:tildetau0}\\
\tau_1 &= \textstyle{\frac{1}{2}}\,\dd_7\phi~,\label{eq:tildetau1}\\
\tau_2 &= 0~, \label{eq:tilde2}\\
\tau_3 & = - H_\perp - \textstyle{\frac{1}{2}}\,\dd_7\phi\lrcorner\psi~,\quad
*_7\tau_3 = -  *_7 H_\perp 
+ \textstyle{\frac{1}{2}}\,\dd_7\phi\wedge\varphi~,\label{eq:tildetau3}
\end{align}
The $G_2$ torsion classes $\tau_p$, $p=0,...3$ are $p$-forms that transform in irreducible representations of $G_2$. The torsion class $\tau_3$ must satisfy the primitivity constraints
 \[ \varphi\wedge *_7\tau_3 = 0~,\qquad
 \tau_3\wedge\varphi = 0~, \]
and the $G_2$ stucture is called integrable when $\tau_2=0$, as is the case here. In the above equations, we have decomposed  $\widehat H$ into components that are parallel and orthogonal to $\varphi$: 
\begin{equation*}
 \widehat H = \textstyle{\frac{1}{14}}\, f\, \varphi + H_\perp~,
\qquad H_\perp\wedge\psi = 0~, \quad  *_7 H_\perp\wedge\varphi = 0~.
\end{equation*}
By further restricting the flux, the torsion class $\tau_0$ can be set to zero so that the $G_2$ structure is integrable conformally balanced \cite{Firedrich:2003}; this case is studied in \cite{Gauntlett:2002sc,Gran:2005wf,Lukas:2010mf}.

It is straight-forward to prove the equivalence between  \eqref{eq:susy1}-\eqref{eq:susy4} and \eqref{eq:one}-\eqref{eq:tildetau3}. The equations for $\tau_2$ and $\tau_1$ are obvious, whereas the remaining equations require some work. Comparing equations \eqref{eq:one} and \eqref{eq:susy1} we obtain
\begin{equation}
\tau_0\, \psi +  *_7\tau_3 = - f\, \psi -  *_7 \widehat H 
 + \textstyle{\frac{1}{2}}\,\dd_7\phi\wedge\varphi~.\label{eq:three}
 \end{equation}
Using the first primitivity constraint on $\tau_3$ in \eqref{eq:three} gives the relation
\[ (\tau_0 + f) \psi\wedge\varphi = -  *_7 \widehat H\wedge\varphi ~.\]
Decomposing $\widehat H$ as $\widehat H = h_{\varphi}\, \varphi + H_\perp$, with $H_\perp$ as above, we have $ *_7 \widehat H = h_{\varphi}\, \psi +  *_7 H_\perp$, and thus
\[ \tau_0 = - f - h_{\varphi}~.\]
Using this in equation \eqref{eq:three},  and taking the Hodge-dual, we find
 \[ \tau_3  
  = - H_\perp -  \textstyle{\frac{1}{2}}\, \dd_7\phi \lrcorner\psi~.\]
  The second primitivity constraint on $\tau_3$ now gives
 \[ 0 = - H_\perp\wedge\varphi 
 - \textstyle{\frac{1}{2}}\, (\dd_7\phi\lrcorner\psi)\wedge \varphi~.\]
It turns out that this equation is equivalent to \eqref{eq:susy3}.  In fact,
\[ (\dd_7\phi \lrcorner\psi)\wedge \varphi 
= -  *_7 (\dd_7\phi\wedge  *_7 \psi)\wedge \varphi 
= -  *_7 (\dd_7\phi\wedge \varphi)\wedge \varphi
= -  *_7 (\varphi\lrcorner(\dd_7\phi\wedge\varphi))
= 4\,  *_7 \dd_7\phi~.\]
The last equation \eqref{eq:susy4} gives
\[  \textstyle{\frac{1}{2}}\, f =  *_7 (\widehat H\wedge\psi) = 
h_{\varphi}\,  *_7 (\varphi\wedge\psi) =  7 h_{\varphi}~,\]
and therefore
\[ h_{\varphi} = \textstyle{\frac{1}{14}}\, f~, \quad
\mbox{and} \quad
\tau_0 = - \textstyle{\frac{15}{14}}\, f~.\]
This concludes the proof of the equivalence between \eqref{eq:susy1}-\eqref{eq:susy4} and \eqref{eq:one}-\eqref{eq:tildetau3}.
For future reference, we record the inverse relations for $\phi$, $f$ and $\widehat H$ in terms of the torsion classes
\be
\begin{split}
\label{eq:torinverse}
f &= - \textstyle{\frac{14}{15}}\, \tau_0~,\\
\dd_7\phi &= 2\,\tau_1~,\\
H_\perp &= - \tau_3 - \tau_1\lrcorner\psi~.
\end{split} 
\ee

In addition to the supersymmetry equations, heterotic vacuum solutions  have to satisfy the equations of motion and the Bianchi identities. At lowest order in $\alpha'$, the latter is 
\[ \dd_7\widehat H = 0~, \qquad f = {\rm constant}~,\]
Inserting \eqref{eq:torinverse} in these relations, we find further constraints on the $G_2$ torsion:
\be
\begin{split}
\label{eq:torbi}
\tau_0 &= {\rm constant}~,\\
0 &= \dd_7\left(\tau_3 + \tau_1\lrcorner\psi 
+ \textstyle{\frac{1}{15}}\, \tau_0\,\varphi\right)~.
\end{split} 
\ee

Finally, we turn to the equations of motion. At lowest order in $\alpha'$, the former comprise Einstein's equations, the equation of motion for the dilaton $\phi$ and the equations of motion for the flux $\widehat H$. To first order in the $\alpha'$ expansion, it has been shown \cite{Gauntlett:2002sc,Ivanov:2009rh} that both Einstein's equations and the equation of motion for the dilaton are implied by the supersymmetry equations, Bianchi identity and flux equations of motion, and so provide no further constraint on the $G_2$ structure. An extension of this result, that includes the flux equation of motion in the equations implied by supersymmetry and the Bianchi identities, can be found in \cite{Martelli:2010jx}. Let us provide an alternative proof for this last point. The seven-dimensional flux equation of motion are
\bea \label{eq:fluxeom7}
0 &= \dd_7 \left( *_7 e^{-2 \phi} \hat{H} \right) \Leftrightarrow \\ \nn
0 &= -\frac{1}{8} \tau_0  \left( -2\dd_7 \phi \w \psi + \dd_7 \psi \right)
-\frac{1}{8} \dd_7 \tau_0  \w  \psi  +
\left( -2\dd_7 \phi \w *_7 H_{\perp} + \dd_7 *_7 H_{\perp} \right)
 \; .
\eea
This equation is also implied by the Killing spinor equations and Bianchi identities, as we now show. Clearly, the supersymmetry constraints and the Bianchi identities set the first three terms on the right hand side of the second equation in \eqref{eq:fluxeom7} to zero. What remains is
\be
0 = \left( -2\dd_7 \phi \w *_7 H_{\perp} + \dd_7 *_7 H_{\perp} \right)= \dd_7 *_7 \tau_3 + \tau_0 \tau_1 \w \psi - 3 \tau_1 \w *_7 \tau_3 \; .
\ee 
This should be compared to
\be
0 = \dd_7^2 \varphi =\left( \dd_7 *_7 \tau_3  + \tau_0 \tau_1 \w \psi - 3 \tau_1 \w *_7 \tau_3 \right) + \left( d_7 \tau_0 \w \psi + 3 \dd_7 \tau_1 \w \varphi \right) \; ,
\ee
where the last bracket vanish as a consequence of the supersymmetry conditions and Bianchi identities. Thus, to zeroth order in $\alpha'$, all equations of motion follow from the Killing spinor equations and Bianchi identities, and the $G_2$ structure is completely specified by \eqref{eq:one}-\eqref{eq:tildetau3} in conjunction with the differential constraints \eqref{eq:torbi}.\\

\section{\texorpdfstring{$SU(3)$}\ \ structures embedded into integrable \texorpdfstring{$G_2$}\ \  structures}
\label{sec:embedding}

In this section, we derive the constraints that the Killing spinor equations and Bianchi identities of $N=\textstyle{\frac{1}{2}}$ domain wall solutions put on the $SU(3)$ structure of $X(t)$. Mathematically, these constraints determine how the $SU(3)$ structure embeds into integrable $G_2$ structures. As in the previous section, we start with the supersymmetry constraints, and then proceed to the Bianchi identities.

\subsection{The embedding}

Consider a manifold $X$ with an $SU(3)$ structure \cite{Gray:1980fk,chiossi,Cardoso:2002hd,Gauntlett:2003cy} determined by the complex $(3,0)$-form $\Psi$ and the real $(1,1)$ form $\omega$ which satisfy the compatibility equations
\begin{equation}
\omega\wedge\Psi = 0~,\qquad\qquad
{\rm dvol}_6 = \frac{1}{6}\, \omega\wedge\omega\wedge\omega =
\frac{i\, \Psi\wedge\bar\Psi}{||\Psi||^2} ~,\label{eq:su3compat}
\end{equation}
and the torsion structure equations 
\begin{align}
\label{eq:su3struct1}
\dd \omega &= - \frac{12}{||\Psi||^2}\ {\rm Im}(W_0\, \overline\Psi) + W_1^\omega\wedge\omega + W_3~,\\
\label{eq:su3struct2}
\dd \Psi &= W_0\ \omega\wedge\omega + W_2\wedge\omega + \overline W_1^\Psi\wedge\Psi~,
\end{align}
where $W_0$ is a complex function, $W_2$ is a primitive $(1,1)$-form, $W_3$ is a real primitive 3-form of type $(1,2)+(2,1)$, and $W_1^\omega$ and $W_1^\Psi$ are 1-forms.  The torsion components $W_1^\omega$ and $W_1^\Psi$ are the Lee-forms of $\omega$ and $\Psi$ respectively. 

The $SU(3)$ structure forms determine both the almost complex structure $J$ (determined by the real part of $\Psi$), and the metric on the manifold \cite{Hitchin01stableforms}, see appendix \ref{ap:conv}. The torsion classes $W_0, W_2$ are related to the Nijenhuis tensor of the almost complex structure, and vanish if and only if the latter is integrable.

We now embed the $SU(3)$ structure $(X,\omega,\Psi)$ into an integrable $G_2$ structure $(Y,\varphi)$, by choosing a 1-form $N= N_t\, \dd t$ and a complex valued function $\alpha$ such that
\begin{equation}
\varphi = N\wedge\omega + {\rm Re} (\alpha\Psi)~,\label{eq:su3intog21}
\end{equation}
The Hodge dual $\psi$ of $\varphi$ is
\begin{equation}
\psi = *_7\varphi =  - N\wedge {\rm Im}(\alpha\Psi ) + \frac{1}{2}\,\omega\wedge\omega~.
\label{eq:su3intog22}
\end{equation}
where we have used
\[*_6\,\omega = \frac{1}{2}\, \omega\wedge\omega~,\qquad *_6\,\Psi  = - i\, \Psi ~,\]
and the metric $g_{\varphi}$ on $Y$ induced by the $G_2$ structure $\varphi$ which is
\begin{equation}
\dd^2 s_{\varphi} = N_t^2\, \dd^2 t\, + \dd^2 s_{X}~,\qquad N_t\,N^t = 1~,\qquad
|\alpha|^2\, ||\Psi ||^2 = 8~.\label{eq:metric}
\end{equation}
Note that we are not assuming that the metric on $X$ is independent of $t$. 
The $SU(3)$ structure $(X,\omega, \Psi)$ varies with $t$ and so does the metric on $X$.
As $\Psi$ and the metric on $X$ are covariantly constant (with respect to the connection with torsion),  $||\Psi||^2$ is a constant on $X$, $\dd ||\Psi||^2 =0$, and therefore so is $|\alpha|$ 
\[ \dd|\alpha| = 0~.\]
It is important however to keep in mind that all these quantities may depend on $t$.  Note also that different choices of $\alpha$ correspond to same almost complex structure $J$ on $X$, however they give different embeddings of the $SU(3)$ structure into the $G_2$ structure.

Recall that an integrable $G_2$ structure satisfies $\tau_2=0$.  Thus the torsion structure equations for $(Y,\varphi)$ are
\begin{align}
\label{eq:intg2struct1}
\dd_7\varphi &= \tau_0\, \psi + 3\tau_1\wedge\varphi + *_7\tau_3~,\\
\label{eq:intg2struct2}
\dd_7\psi &= 4 \tau_1\wedge\psi.
\end{align}
The torsion class $\tau_3$ belongs to $\Lambda_{27}^3$, that is 
\begin{equation}
\varphi\wedge*_7\tau_3= 0~,\qquad\qquad\varphi\lrcorner(*_7\tau_3)= 0~.
\label{eq:consttau3}
\end{equation}
The torsion class $\tau_0$ is therefore
\begin{equation}
7\, \tau_0 = *_7(\varphi\wedge\dd_7\varphi)~.
\label{eq:tau0}
\end{equation}
 For these $G_2$ structures the 1-form $\tau_1$ is not in general closed.  Since we are interested in the $N=1/2$ domain wall solutions of Section \ref{sec:domainwall},
we restrict the $G_2$ structure so that $\tau_1$ is exact
\[ \tau_1 = \frac{1}{2}\, \dd_7 \phi~.\] 

To embed the $SU(3)$ structure into an integrable $G_2$ structure,
we need to decompose the $G_2$ torsion classes as follows
\begin{align}
\tau_1 &= \tau_{1\, t}\,\dd t + \tau_1^X~,
\label{eq:embedtau1}\\
\tau_3 &= \dd t\wedge\tau_{3\, t} + \tau_3^X~,\qquad
*_7\tau_3 = - N\wedge*\tau_3^X + N_t^{-1}\, *\tau_{3\, t}~,\label{eq:embedtau3}
\end{align}
where
\[ \tau_{1\, t} = \frac{1}{2}\, \partial_t\phi~,\qquad\qquad \tau^X = \frac{1}{2}\, \dd\phi~.\]

The constraints on $\tau_3$, equations \eqref{eq:consttau3}, can be decomposed using equation \eqref{eq:embedtau3}.
The first constraint gives 
\begin{equation}
\omega\wedge*\tau_{3\, t} = - N_t\, {\rm Re}(\alpha\Psi)\wedge *\tau_3^X~,\label{eq:consttau31}
\end{equation}
and the second 
\begin{equation}
{\rm Re}(\alpha\Psi)\lrcorner*\tau_3^X = 0~,\qquad 
N_t\, \omega\lrcorner*\tau_3^X = {\rm Re}(\alpha\Psi)\lrcorner*\tau_{3\,t}~.
\label{eq:consttau32}
\end{equation}

The following identities will be useful in our computations.
\begin{lemma}\label{idsforembed}
Let 
\[ \rho = *\omega = \frac{1}{2}\, \omega\wedge\omega~.\]
Then
\begin{align*}
\dd{\rm Re}(\alpha\Psi) &= 
2 {\rm Re}(\alpha\, W_0)\, \rho + {\rm Re}(\alpha\, W_2)\wedge\omega
+ {\rm Re}((\overline W_1^\Psi + \dd\log\alpha)\wedge\alpha\Psi)\\
&=
2 {\rm Re}(\alpha\, W_0)\, \rho + {\rm Re}(\alpha\, W_2)\wedge\omega
+ {\rm Re}W_1^{\,\Psi}\wedge {\rm Re}(\alpha\Psi)
- ({\rm Im}\overline W_1^{\,\Psi} + \dd a)\wedge {\rm Im}(\alpha\Psi)~,\\
\dd{\rm Im}(\alpha\Psi) &= 
 2 {\rm Im}(\alpha\, W_0)\, \rho + {\rm Im}(\alpha\, W_2)\wedge\omega
+ {\rm Im}((\overline W_1^\Psi + \dd\log\alpha)\wedge\alpha\Psi)\\
& = 2 {\rm Im}(\alpha\, W_0)\, \rho + {\rm Im}(\alpha\, W_2)\wedge\omega
+ {\rm Re}W_1^{\,\Psi} \wedge {\rm Im}(\alpha\Psi)
+ ({\rm Im}\overline W_1^{\,\Psi} + \dd a)\wedge {\rm Re}(\alpha\Psi)~,\\[3pt]
\dd(N_t\, \omega) &= N_t \left(
(\dd\log N_t + W_1^\omega)\wedge\omega 
- \textstyle{\frac{3}{2}}\, {\rm Im}((\alpha W_0)\, (\bar\alpha\overline\Psi))
+ W_3\right)~.
\end{align*}
where $a$ is the argument of $\alpha$. 
\end{lemma}
\begin{proof}
These identities are easily obtained by a calculation using equations \eqref{eq:su3struct1} and \eqref{eq:su3struct2}.  In the third equation we have also used identity \eqref{eq:metric}. 
\end{proof}

To embed the $SU(3)$ structure into an integrable $G_2$ structure, we use equations \eqref{eq:su3intog21} and \eqref{eq:su3intog22}, into the equations \eqref{eq:intg2struct1} and  \eqref{eq:intg2struct2}.  

\begin{proposition}\label{EMBED}
The embedding given by equations \eqref{eq:su3intog21} and \eqref{eq:su3intog22} of the $SU(3)$ structure $(X, \omega,\Psi)$ with torsion classes as in \eqref{eq:su3struct1} and \eqref{eq:su3struct2} into an integrable $G_2$ structure $(Y,\varphi)$ with torsion classes \eqref{eq:intg2struct1} and\eqref{eq:intg2struct2} gives the the following relations between the torsion classes and the flow of $\omega$ and $\Psi$ with respect to the coordinate $t$:
\begin{align}
W_1^\omega &= 2 \tau_1^X\label{eq:W1tau1}~,\\[5pt]
\partial_t\rho &= 4\tau_{1\,t}\, \rho
+ 2 N_t\, W_1^\omega\wedge{\rm Im}(\alpha\Psi)
 - \dd (N_t\,{\rm Im}(\alpha\Psi))
 \label{eq:flowforrho}\\
&= 2(2\, \tau_{1\,t} - N_t\, {\rm Im}(\alpha W_0))\, \rho 
- N_t\, {\rm Im}(\alpha W_2)\wedge\omega 
\nn\\&\qquad\qquad\qquad\qquad\qquad 
- N_t\, {\rm Im}\left( (\dd\log(\alpha\, N_t) - 2 W_1^\omega+ \overline W_1^\Psi )
\wedge\alpha\Psi\right)~,
\nn\\[3pt]
\partial_t\omega &=
(2\, \tau_{1\, t} - 3\, N_t\, {\rm Im}(\alpha W_0))\, \omega
+ 2\, N_t\, W_1^\omega\lrcorner{\rm Re}(\alpha\Psi)
- {\dd^c}^\dagger(N_t\, {\rm Im}(\alpha\Psi))\label{eq:flowforomega}\\
&= (2\, \tau_{1\, t} + \, N_t\, {\rm Im}(\alpha W_0))\, \omega
- 2\,N_t\, {\rm Im}(\alpha W_2)
+ 2\, N_t\, W_1^\omega\lrcorner{\rm Re}(\alpha\Psi)
+ \dd^\dagger(N_t\, {\rm Re}(\alpha\Psi))
\nn\\
&= (2\, \tau_{1\,t} - N_t\, {\rm Im}(\alpha W_0))\, \omega 
- N_t\, {\rm Im}(\alpha W_2) \nn\\
&\qquad\qquad\qquad\qquad\qquad 
- N_t\, {\rm Re}\left((\dd\log(\alpha\, N_t) - 2 W_1^\omega + \overline W_1^\Psi)
\lrcorner(\alpha\Psi)\right)~,\nn\\[3pt]
7\, N_t\, \tau_0 &= 12\, N_t\, {\rm Re}(\alpha\, W_0)
- {\rm Im}(\alpha\Psi)\lrcorner\partial_t{\rm Re}(\alpha\Psi)~,\label{eq:flowforPsi1}\\[5pt]
N_t^{-1}\, \tau_{3\, t} &= \dd^\dagger{\rm Im}(\alpha\Psi) - \tau_0\, \omega 
+ \textstyle{\frac{3}{2}}\, W_1^\omega\lrcorner{\rm Im}(\alpha\Psi)\label{eq:tau3t}\\
&= (2\, {\rm Re}(\alpha W_0) - \tau_0)\, \omega - {\rm Re}(\alpha W_2)
-{\rm Im}\left((\dd\log\alpha-\textstyle{\frac{3}{2}}\, W_1^\omega + \overline W_1^\Psi)
\lrcorner(\alpha\Psi)\right) ~,
\nn\\[5pt]
\partial_t{\rm Re}(\alpha\Psi) &= \dd (N_t\,\omega) 
-  \textstyle{\frac{3}{2}}\, N_t\, W_1^\omega\wedge\omega
+  3\, \tau_{1\,t}\, {\rm Re}(\alpha\Psi) - N_t\, \tau_0\, {\rm Im}(\alpha\Psi) - N_t *\tau_3^X
\label{eq:flowforPsi2}\\[3pt]
&= N_t\, \left(\dd\log N_t - \textstyle{\frac{1}{2}}\, W_1^\omega\right)\wedge\omega 
+\textstyle{\frac{3}{2}}\, ( 2\, \tau_{1\, t}
 - N_t\, {\rm Im}(\alpha W_0))\, {\rm Re}(\alpha\Psi)\nn\\[1pt]
 &\quad
 - N_t\, \left(\tau_0 - \textstyle{\frac{3}{2}}\, {\rm Re}(\alpha W_0)\right)\, {\rm Im}(\alpha\Psi)
 - N_t\, *\tau_3^X + N_t\, W_3~.\nn
\end{align}
\end{proposition}
\begin{proof}
We begin our analysis of the embedding with equation \eqref{eq:intg2struct2} (which enforces the condition that $\tau_2=0$). Using equations \eqref{eq:su3intog21} and \eqref{eq:su3intog22} and Lemma \ref{idsforembed}, we obtain
\begin{align}
\dd_7\psi  - 4 \tau_1\wedge\psi &= 
 \dd t\wedge\left(\, \partial_t\omega\wedge\omega - 2\tau_{1\,t}\, \omega\wedge\omega
+ \dd (N_t\,{\rm Im}(\alpha\Psi)) - 4 N_t\, \tau_1^X\wedge{\rm Im}(\alpha\Psi)\,\right)\nn \\
& \quad + \dd\omega\wedge\omega - 2 \tau_1^X\wedge\omega\wedge\omega\nn~,
\end{align}
which gives equations \eqref{eq:W1tau1} and \eqref{eq:flowforrho}
\begin{align*}
\dd\omega\wedge\omega &= W_1^\omega\wedge\omega\wedge\omega= 2\, \tau_1^X\wedge\omega\wedge\omega\qquad\iff\qquad W_1^\omega = 2 \tau_1^X~,\\
\partial_t\omega\wedge\omega &= 2\tau_{1\,t}\, \omega\wedge\omega
 - \dd (N_t\,{\rm Im}(\alpha\Psi)) + 4 N_t\, \tau_1^X\wedge{\rm Im}(\alpha\Psi)~.
\end{align*}
The last of these equations gives the first identity in \eqref{eq:flowforrho} and we obtain the second after using Lemma \ref{idsforembed}. Equations \eqref{eq:flowforomega} can be obtained from equation \eqref{eq:flowforrho} using the formula
\[ \partial_t\omega = \omega\lrcorner\partial_t\rho 
- \frac{1}{2}\, (\rho\lrcorner\partial_t\rho)\, \omega~,\]
and recalling that for any $k$-form $\beta$
\[\dd^c\beta = J^{-1}\dd (J\beta)~,\qquad
{\dd^c}^\dagger\beta =  - *\dd^c(*\beta)~.\]

Now we turn to equation \eqref{eq:intg2struct1}. Equation \eqref{eq:flowforPsi1} can be easily obtained by computing $\tau_0$ using equations \eqref{eq:tau0} and \eqref{eq:intg2struct1}:
\begin{align*}
7\, \tau_0\, \dd{\rm vol}_\varphi &= \varphi\wedge\dd_7\varphi = (N\wedge\omega + {\rm Re}(\alpha\Psi))\wedge
(\dd t\wedge (\partial_t{\rm Re}(\alpha\Psi) - \dd(N_t\omega)) + \dd{\rm Re}(\alpha\Psi))\\
&= \dd t\wedge \left( N_t\, \omega\wedge \dd{\rm Re}(\alpha\Psi) - 
{\rm Re}(\alpha\Psi)\wedge (\partial_t{\rm Re}(\alpha\Psi) - \dd(N_t\omega))\right)\\
&= \dd t\wedge ( 12\, N_t \, {\rm Re}(\alpha W_0)\, \dd {\rm vol}_X 
- {\rm Re}(\alpha\Psi)\wedge \partial_t{\rm Re}(\alpha\Psi)  )~.
\end{align*}
Taking the Hodge dual, we obtain equation~\eqref{eq:flowforPsi1} where we have used
\[ *({\rm Re}(\alpha\Psi)\wedge \partial_t{\rm Re}(\alpha\Psi))
= {\rm Im}(\alpha\Psi)\lrcorner\partial_t{\rm Re}(\alpha\Psi)~.\]

Using equations \eqref{eq:su3intog21} and \eqref{eq:su3intog22} into equation \eqref{eq:intg2struct1}
we obtain
\begin{align*}
\dd_7\varphi  - \left(\tau_0\, \psi + 3\tau_1\wedge\varphi + *\tau_3\right)   &=   
 \dd t\wedge
 ( -\dd (N_t\,\omega) + \partial_t{\rm Re}(\alpha\Psi)
+ N_t\, \tau_0\, {\rm Im}(\alpha\Psi) \\
 & \qquad\qquad - 3\, \tau_{1\,t}\, {\rm Re}(\alpha\Psi) 
 +3\, N_t\, \tau_1^X\wedge\omega + N_t *\tau_3^X)) \\
 &\quad  + \dd{\rm Re}(\alpha\Psi) - \tau_0 \, \rho
- 3\, \tau_1^X\wedge{\rm Re}(\alpha\Psi) 
- N_t^{-1}\, *\tau_{3\, t}
~,\end{align*}
from which we find two relations
\begin{align*}
\dd{\rm Re}(\alpha\Psi) &= \tau_0 \, \rho + 3\, \tau_1^X\wedge{\rm Re}(\alpha\Psi)  + 
N_t^{-1}\,*\tau_{3\, t}~,\\
\partial_t{\rm Re}(\alpha\Psi) &= \dd (N_t\,\omega) - 3\, N_t\, \tau_1^X\wedge\omega
+  3\, \tau_{1\,t}\, {\rm Re}(\alpha\Psi) - N_t\, \tau_0\, {\rm Im}(\alpha\Psi) - N_t *\tau_3^X~.
\end{align*} 
The first one gives a relation between the torsion classes of the $G_2$ structure and the $SU(3)$ structure
which we can write, for example, as an equation for $*\tau_{3t}$.  Using Lemma \ref{idsforembed} we find
\begin{align}
\label{eq:dualtau3t}
N_t^{-1}\,*\tau_{3\, t} &= \dd{\rm Re}(\alpha\Psi) 
- \tau_0 \, \rho - 3\, \tau_1^X\wedge{\rm Re}(\alpha\Psi)\\ \nn
&= (2{\rm Re}(\alpha W_0)  - \tau_0)\, \rho
+ {\rm Re}(\alpha W_2)\wedge\omega
+{\rm Re}\left((\dd\log\alpha-\textstyle{\frac{3}{2}}\, W_1^\omega + \overline W_1^\Psi)
\wedge\alpha\Psi\right)
~.
\end{align}
Taking the Hodge dual, we find equation \eqref{eq:tau3t}. This relation is a flow equation for ${\rm Re}(\alpha\Psi)$ which, using Lemma \ref{idsforembed}, gives equation \eqref{eq:flowforPsi2}.  
\end{proof}

It will be useful for later to have an expression for $\tau_3^X$ which satisfies the contraints in equations
\eqref{eq:consttau31} and \eqref{eq:consttau32}.
\begin{proposition}
\begin{equation}
*\tau_3^X = 
\frac{1}{2}\, \left(2\,J\dd a + 3\, W_1^\omega - 4\, {\rm Re} W_1^\Psi\right)\wedge\omega
+ \frac{3}{4}\, (\tau_0 - 2\, {\rm Re}(\alpha W_0))\, {\rm Im}(\alpha\Psi) + \gamma~,
\label{eq:relfordualtau3X}\end{equation}
where $\gamma$ is a primitive 3-form of type $(2,1)+(1,2)$.  
\end{proposition}
\begin{proof}
We begin by writing the Lefshetz decomposition of $*\tau_3^X$
\[ *\tau_3^X = \beta\wedge\omega + \tilde\gamma~, \qquad \omega\lrcorner\tilde\gamma = 0~.\]
Also, the Hodge decomposition of $\tilde\gamma$ can be written as
\[ \tilde\gamma = \kappa_1\, {\rm Re}(\alpha\Psi) + \kappa_2\, {\rm Im}(\alpha\Psi) + \gamma~,\]
where $\gamma$  is of type $(2,1) + (1,2)$.
Using equation \eqref{eq:dualtau3t} into the second equation in \eqref{eq:consttau32} we obtain
\[ \omega\lrcorner *\tau_3^X = 2\, \beta = 
3\, W_1^\omega - 4\, {\rm Re} W_1^\Psi + 2\,J\dd a~,
\]
which gives the first term in equation \eqref{eq:relfordualtau3X}. The first relation in equation 
\eqref{eq:consttau32} gives
\[ 0 = {\rm Re}(\alpha\Psi)\lrcorner *\tau_3^X = 4\, \kappa_1~,\]
whereas the Hodge dual of equation \eqref{eq:consttau31} gives
\[ {\rm Im}(\alpha\Psi)\lrcorner *\tau_3^X = 4\, \kappa_2 = - N_t^{-1}\, \omega\lrcorner *\tau_{3\, t} 
= 3(\tau_0 - 2\, {\rm Re}(\alpha W_0))~.\]
Putting all this together we obtain equation \eqref{eq:relfordualtau3X}.  By computing the Hodge dual of \eqref{eq:relfordualtau3X}, we obtain an equation for $\tau_3^X$.
\end{proof}

Using the expression for $\tau_3$ in this proposition, we can rewrite equation \eqref{eq:flowforPsi2} for the flow of $\Psi$. Substituting equation \eqref{eq:relfordualtau3X} into equation \eqref{eq:flowforPsi2} we find
\begin{align}
\partial_t{\rm Re}(\alpha\Psi) &= \dd(N_t\,\omega)
+ N_t\, (- 3 W_1^\omega + 2 {\rm Re}W_1^\Psi - J\dd a)\wedge\omega\nn\\
 & \qquad +  3\, \tau_{1\, t}\, {\rm Re}(\alpha\Psi)
 - \textstyle{\frac{1}{4}}\, N_t\, (7\tau_0 - 6\, {\rm Re}(\alpha W_0))\, {\rm Im}(\alpha\Psi)
  - N_t \, \gamma\nn\\
&= N_t\, (\dd\log N_t - 2 W_1^\omega + 2 {\rm Re}W_1^\Psi - J\dd a)\wedge\omega\nn\\
 & \qquad 
 + \textstyle{\frac{3}{2}}\, (2\, \tau_{1\, t} - N_t\, {\rm Im}(\alpha W_0))\, {\rm Re}(\alpha\Psi)
 - \textstyle{\frac{1}{4}}\, N_t\, (7\tau_0 - 12\, {\rm Re}(\alpha W_0))\, {\rm Im}(\alpha\Psi)\nn\\
 &\qquad + N_t (W_3 - \gamma)~.\label{eq:flowforPsi3}
\end{align}
Note the consistency between this equation and equation \eqref{eq:flowforPsi1}.  

\subsection{Flow equations and moduli}

The manifold $X$ is an almost-hermitian manifold with an $SU(3)$ structure.  Its almost-complex structure $J$  is completely determined by $\Psi$, and therefore, the flow of $\Psi$ corresponds to variations of the almost complex structure of $X$ as $t$ varies.  On the other hand, the flow of $\omega$ has simultaneous variations of both the hermitian structure and those of the almost complex structure.

The variations of $\omega$ can be written as
\begin{equation}\partial_t\omega = \lambda_t\, \omega + h_t~,
\label{eq:variomega}\end{equation}
where $\lambda_t$ is a function on $Y$ and $h_t$ is a primitive real 2-form on $X$. Comparing the flow equation for $\omega$ \eqref{eq:flowforomega} with \eqref{eq:variomega} we have
\begin{align*}
\lambda_t &=2\, \tau_{1\, t} - N_t\, {\rm Im}(\alpha W_0)
= \partial_t\phi - N_t\, {\rm Im}(\alpha W_0)~,\\
h_t^{(1,1)} &= - N_t\, {\rm Im}({\alpha W_2})~,\\
h_t^{(0,2)} &= - \textstyle{\frac{1}{2}}\, N_t(\dd\log N_t + \bar\eta)\lrcorner (\bar\alpha\overline\Psi)~,
\end{align*}
where we have defined the $(0,1)$ form
\begin{equation}
\eta = \overline W_1^{\, \Psi} - 2\, (W_1^\omega)^{(0,1)} + \bar\partial\log\alpha 
= \overline W_1^{\, \Psi} - 2\, (W_1^\omega)^{(0,1)} + i\, \bar\partial a~,
\label{eq:eta}
\end{equation}
for later convenience.
It is very interesting to note that $h_t^{(1,1)}= 0$ vanishes for all $SU(3)$ structures for which ${\rm Im}(\alpha W_2)$ vanishes.  In these cases  the $SU(3)$ structure deforms with $t$ such that the hermitian structure is fixed. 

The variations of $\Psi$  with respect to $t$ are given by \cite{MR0112154, 0128.16902}
\begin{equation}\partial_t\,\Psi = K_t\, \Psi + \chi_t~,
\label{eq:variPsi}\end{equation}
where $K_t$ is a complex valued function on $Y$ and $\chi_t$ is a $(2,1)$-form on $X$.  As $t\in\IR$, we also have
\[\partial_t\,\overline\Psi = \overline K_t\, \overline\Psi +\overline\chi_t~.\]
The variation of the volume form compatibility condition \eqref{eq:su3compat} of $SU(3)$ structure gives
\begin{equation*}
\partial_t \log||\Psi||^2 = 2\, {\rm Re} K_t  - 3 \lambda_t~.
\end{equation*}
${\rm Re} K_t$ is constant on $X$  up to diffeomorphisms, and therefore $\lambda_t$ is constant on $X$. Both functions can vary with $t$
\[ \dd{\rm Re} K_t = 0~,\qquad\dd\lambda_t=0~.\]
Therefore
\[ \partial_t {\rm Re}(\alpha\Psi) = {\rm Re}(\alpha\chi_t)
+ \textstyle{\frac{3}{2}}\, \lambda_t\, {\rm Re}(\alpha\Psi)
- (\partial_t a + {\rm Im} K_t )\, {\rm Im}(\alpha\Psi)~,\]
where we have used equation \eqref{eq:metric}.
Comparing with the flow equation \eqref{eq:flowforPsi3}  
we find
\begin{align}
{\rm Re}(\alpha\chi_t) 
&= N_t\,\left(
(\dd\log N_t  + 2\, {\rm Re}\,\eta)\wedge\omega 
+ W_3 - \gamma\right)~,\\ \label{eq:KtI}
N_t\, (7\,\tau_0 - 12\, {\rm Re}(\alpha W_0)) &= 4\, (\partial_t a + {\rm Im} K_t )~.
\end{align}
We can obtain $\chi_t$ by taking the $(2,1)$ component of the equation for ${\rm Re}(\alpha\chi_t)$:
\begin{equation}
\alpha\,\chi_t = 2\, N_t\, \left((P\dd\log N_t + \bar\eta)\wedge\omega 
+ (W_3 - \gamma)^{(2,1)}\right)~.\label{eq:chit}
\end{equation}

As $W_1^\omega = \dd\phi$, the $SU(3)$ structure equation for $\rho$
\[ \dd\rho = 2\, W_1^\omega\wedge\rho~,\]
becomes
\[ \dd(e^{-2\phi} \, \rho) = 0~.\]
This needs to be compatible with the  flow equation \eqref{eq:flowforrho} for $\partial_t\rho$ which can now be written as
\[ \partial_t (e^{-2\phi}\, \rho) = - \dd (e^{-2\phi}\, N_t\, {\rm Im}(\alpha\Psi))~.\]
Therefore, the 4-form $e^{-2\phi} \, \rho \in H^4(X)$, flows into another $\dd$-closed 4-form in the {\it same} cohomology class as $e^{-2\phi}\,\rho$.
Consider now the equations for $\omega$.  The flow equation \eqref{eq:flowforomega} can be written as
\[\partial_t \omega= \lambda_t\, \omega - N_t\, {\rm Im}(\alpha W_2) 
- N_t\, {\rm Re}\left((\bar\partial\log N_t + \eta)\lrcorner(\alpha\Psi)\right)~.\]
This expression must be a solution to the variation of the $SU(3)$ integrability equation \eqref{eq:su3struct1}
\begin{equation}
\dd\partial_t\omega = - \textstyle{\frac{3}{2}}\, |\alpha|^2\, 
{\rm Im}\left((\partial_t W_0)\overline\Psi + W_0\, \partial_t\overline\Psi\right)
+ \partial_t \dd\phi\wedge\omega + \dd\phi\wedge\partial_t\omega
+ \partial_t W_3~.\label{eq:eqvaromegat}
\end{equation}
Plugging into this equation the expression for $\partial_t\omega$ above, gives a non-trial equation for the variation of the torsion class $W_3$
\begin{align}
\partial_t W_3 - \lambda_t W_3 \, &= 
 \textstyle{\frac{3}{2}}\, |\alpha|^2\,
{\rm Im}\left((\partial_t W_0 + (\overline K_t - \lambda_t)\, W_0)\overline\Psi + W_0\, \overline\chi_t\right)
\nn\\[5pt]
&\qquad + e^\phi\, \dd\big(e^{-\phi}\, h_t\big) - \,  \partial_t\dd\phi\wedge\omega ~.\label{eq:flowW3}
\end{align}

Consider now the integrability equation for $\Psi$, equation \eqref{eq:su3struct2}.  Varying this equation with respect to $t$, we find a differential equation for the variation $\chi_t$ of $\Psi$:
\begin{align}
\dd\chi_t - \overline W_1^\Psi\wedge\chi_t &= 
(\partial_t\overline W_1^\Psi - \dd K_t)\wedge\Psi
+ 2 \, (\partial_t W_0 - K_t\, W_0)\, \rho + (\partial_t W_2 - K_t\, W_2)\wedge\omega\nn\\
&\qquad + 2\, W_0\, \partial_t\rho + W_2\wedge\partial_t\omega~.\label{eq:eqchit}
\end{align}
This equation together with the exterior derivative of equation \eqref{eq:chit} 
give a differential equation for the primitive part of $\chi_t$.

\subsection{Bianchi identities}

The supersymmetric solutions we have discussed need to satisfy the Bianchi identities.  We consider these identities as a further constraint on the integrable $G_2$ structure $(Y,\varphi)$.  Recall
\[\widehat H = - \left ( \tau_3 + \tau_1\lrcorner\psi + \textstyle{\frac{1}{15}}\, \tau_0\, \varphi\right)~.\]
The Bianchi identities state that $\widehat H$ is closed
\[ \dd_7 \widehat H = 0~,\]
and that $\tau_0$ is a constant on $Y$. 
We express these constraints as constraints on the $SU(3)$ structure $(X,\Psi,\omega)$ embedded into the $G_2$ structure.

Let
\[ \widehat H = \dd t \wedge S_t + S^X~.\]
Then
\begin{align*}
 S^X &= (2 {\rm Im}\eta + JW_1^\omega)\wedge\omega 
- \textstyle{\frac{1}{2}}\, \left(\textstyle{\frac{49}{30}} \, \tau_0
- 3 {\rm Re}(\alpha\, W_0)\right) {\rm Re}(\alpha\Psi)
+ N_t^{-1}\, \tau_{1\, t}\, {\rm Im}(\alpha\Psi) + J\gamma~,\\[5pt]
N_t^{-1}\, S_t &= 2\left(\textstyle{\frac{7}{15}}\, \tau_0 - {\rm Re}(\alpha W_0)\right)\omega
+ {\rm Re} (\alpha W_2)
+ {\rm Im} (\eta\lrcorner(\alpha\Psi))~,
\end{align*}
where
\[ \eta = \overline W_1^{\, \Psi} - 2 (W_1^\omega)^{(0,1)} + \bar\partial\log\alpha
= \overline W_1^{\, \Psi} - 2 (W_1^\omega)^{(0,1)} + i\bar\partial a~,\]
and we have used
\[ \tau_1\lrcorner\psi = \frac{1}{2}\, N\wedge(W_1^\omega\lrcorner{\rm Im}(\alpha\Psi))
- N_t^{-1}\, \tau_{1\, t}\, {\rm Im}(\alpha\Psi) - \frac{1}{2}\, JW_1^\omega\wedge\omega~.\]

Now the Bianchi identity gives
\begin{equation*}
0=\dd_7\widehat H = \dd t\wedge (\partial_t S^X - \dd S_t) + \dd S^X ~,
\end{equation*}
from which we obtain two equations
\begin{align} \label{eq:bi}
\dd S^X &= 0\\ \nn
\partial_t S^X&= \dd S_t~.
\end{align}
This means that variations $\partial_t S^X$ of a solution $S^X$ in a cohomology class in $H^3(X)$ remain in the same cohomology class.  

\subsection{Summary of constraints}

Since this is a long section, let us briefly recall the constraints that supersymmetry and the Bianchi identities put on the the $SU(3)$ structure.
\begin{itemize}
\item The Lie form $W_1^{\omega}$ is exact: $W_1^{\omega}=\dd \phi$.
\item The embedding of the $SU(3)$ structure into an integrable $G_2$ structure is specified through \eqref{eq:su3intog21} by a real function $N_t$ and a complex function $\alpha$. 
\item The remaining torsion classes $W_0$, $W_2$ and $W_3$ can take on different values, and determine, together with $N_t, \alpha$ and $\phi$, the NS flux components $\widehat H$ and $f$, as well as the flow of $\omega$ and $\Psi$.
\item The Bianchi identities \eqref{eq:bi} further constrain the $G_2$ structure.
\end{itemize}
In addition, the flow of the $SU(3)$ structure is determined by
\begin{equation}
\label{eq:generalflow}
\begin{split}
\partial_t \Psi &= K_t\, \Psi + \chi_t~,\\
{\rm Im} K_t &= \textstyle{\frac{N_t}{4}}\, (7\, \tau_0 - 12\, {\rm Re}W_0)~, \qquad \dd{\rm Re}K_t = 0~,\\
\chi_t &= 2\, N_t\, \left((\dd\log N_t + 2 {\rm Re}\eta)\wedge\omega + W_3 - \gamma\right)^{(2,1)}~,
\quad \omega\lrcorner\gamma = 0~,\quad \Psi\wedge\gamma = 0~,\\
\eta &= \overline W_1{}^{\Psi} - 2 \, (W_1^\omega)^{(0,1)}~,\\[10pt]
\partial_t \omega &= \lambda_t\,\omega + h_t~,\\
\lambda_t &= 2\,\tau_{1\, t} - N_t\, {\rm Im} W_0~, 
\quad \tau_{1\, t} = \textstyle{\frac{1}{2}}\,\partial_t\phi~,\quad \dd\lambda_t = 0~,\\ 
h_t &= - N_t\, {\rm Im}W_2 - N_t\, {\rm Re}\left((\partial\log N_t + \bar\eta)\lrcorner\bar\Psi\right)~.
\end{split}
\end{equation}

\vspace{0.5cm}


\section{Example: Flow of Calabi--Yau manifold with flux}
\label{sec:CY}

In this, and the two following sections, we will determine the flow of $SU(3)$ structure manifolds with restricted torsion. The aim is to study how the constraints from the embedding $G_2$ structure determines a path in the moduli space of the $SU(3)$ structure manifolds. We will show that flows can be $SU(3)$ structure-preserving, so that a fixed set of torsion classes are non-zero along the flow. Flows where the $SU(3)$ torsion classes change are also allowed, and will be analysed from different perspectives.

To simplify our calculations,  we henceforth work in the gauge $\alpha=1$. This means in particular that we have absorbed the phase $a$ in $\Psi$. As a consequence, the phases of $W_0, W_2$ change, and $W_1^{\Psi}$ is modified by $i \partial a$. 
Moreover, in the gauge $\alpha = 1$ we have $||\Psi||^2 = 8$ for all values of $t$.  In this case we have the relation
\begin{equation}
{\rm Re}K_t = \frac{3}{2}\, \lambda_t~.\label{eq:ReK}
\end{equation}

The simplest $SU(3)$ structure is when all torsion classes are zero so that the manifold is Calabi--Yau. We now determine the necessary and sufficient constraints on the flow so that the three-fold stays Calabi--Yau to all orders in $t$. In the next subsection, we will compute the torsion classes generated at linear order in $t$, once these constraints fail to hold. A summary and discussion of the result from this perturbative analysis is found in the last subsection.

\subsection{Calabi--Yau to all orders}
\begin{proposition}\label{prop:cyflow}
The $N=1/2$ domain wall flow preserves the Calabi--Yau conditions if and only if $\dd N_t = 0$ and $\gamma$ is harmonic.
\end{proposition}
\begin{proof} 
Suppose that all $SU(3)$ torsion classes vanish. Then the general analysis gives
\begin{align}
\partial_t \Psi &= K_t\, \Psi + \chi_t~,\nn\\
{\rm Im} K_t &= \textstyle{\frac{N_t}{4}}\, 7\, \tau_0 ~, \qquad \dd K_t = 0~,\label{eq:K&N}\\
\chi_t &= 2\, \left(\dd N_t\wedge\omega  - N_t\, \gamma\right)^{(2,1)}~,
\quad \omega\lrcorner\gamma = 0~,\quad \Psi\wedge\gamma = 0~,\label{eq:chiCY}\\[10pt]
\partial_t \omega &= \lambda_t\,\omega + h_t~,\nn\\
\lambda_t &= 2\,\tau_{1\, t} ~, 
\quad \tau_{1\, t} = \textstyle{\frac{1}{2}}\,\partial_t\phi~,\quad \dd\lambda_t = 0~,\nn\\ 
h_t &= - \dd N_t\lrcorner{\rm Re}\Psi = \dd^\dagger(N_t\, \Psi)~.\label{eq:htCY}
\end{align}
Note that  $K_t$  is $\dd$-constant whenever the form $\Psi$ is holomorphic. Then, taking the exterior derivative of the first equation in \eqref{eq:K&N}
we find that 
\[ \tau_0 \,\dd\, N_t = 0~.\]
Therefore, either $\tau_0=0$ or $N_t$ is a constant.

Consider the flow equation for $\omega$. To preserve the K\"ahler  condition, we have that
\[ \dd\partial_t\omega = 0\quad\iff\quad \dd h_t= 0~.\]
Using the expression for $h_t$ in equation \eqref{eq:htCY}, we find
\[ \dd h_t = 0\quad\iff\quad \dd\dd^\dagger(N_t\, \Psi)= 0\quad\iff\quad h_t = \dd^\dagger(N_t\, \Psi) = 0 ~.\]
Therefore
\[ \dd N_t = 0~,\]
and $N_t$ should be  $\dd$-constant, which proves the first condition in Proposition \ref{prop:cyflow}.  We can absorb $N_t$, which is a function of $t$ only, into the definition of $t$, and therefore we can set
\[ \partial_t N_t = 0~.\]
 Consider now the flow for $\Psi$. As $N_t$ is constant, equation \eqref{eq:chiCY} gives
\[ {\rm Re}\chi_t = - \gamma~,\qquad {\rm Im}\chi_t = J\gamma~,\]
and the flow equation implies 
\[ \dd\partial_t\Psi = 0 \quad\iff\quad \dd\chi_t = 0 \quad\iff \dd\gamma = 0~.\]

The first Bianchi identity
\[ \dd S^X = 0~,\]
where
\[ S^X = 
-  \textstyle{\frac{49}{60}} \, \tau_0\, {\rm Re}\Psi
+ \textstyle{\frac{1}{2}}\,N_t^{-1}\, \lambda_t\, {\rm Im}\Psi + J\gamma~,\]
is satisfied only if 
\[ \dd (J\gamma) = 0~,\]
where have used the fact that $\lambda_t$ is a constant.  This constraint on the primitive form
$\gamma$ is equivalent to $\gamma$ being co-closed
\[ \dd^\dagger \gamma = 0~.\]
\end{proof} 

We obtain further constraints from the second Bianchi identity, which in our case is
\[ \partial_t S^X = \dd S_t~,\qquad S_t =  N_t\, \textstyle{\frac{14}{15}}\, \tau_0 \,\omega~.\]
Noting that $S_t$ is closed, and using the expression for $S^X$ above this gives an expression for the variation of
$J\gamma$
\begin{align*}
 \partial_t(J\gamma) &= - \textstyle{\frac{49}{60}}\, \tau_0 \,\gamma
- \textstyle{\frac{1}{2}}\, \lambda_t N_t^{-1}\,J\gamma\\[3pt]
&\quad + \textstyle{\frac{7}{20}}\, \tau_0 \,\lambda_t\, {\rm Re}\Psi
- \textstyle{\frac{1}{2}}\, \left(\partial_t(\lambda_t N_t^{-1}) + 
\textstyle{\frac{7^3}{120}}\,\tau_0^2\, N_t + \textstyle{\frac{3}{2}}\,\lambda_t^2\, N_t^{-1}
\right)\, {\rm Im}\Psi \\[5pt]
&= - \textstyle{\frac{49}{60}}\, \tau_0 \,\gamma
- \textstyle{\frac{1}{2}}\, \lambda_t N_t^{-1}\,J\gamma
+ \textstyle{\frac{7}{20}}\, \tau_0 \,\lambda_t\, {\rm Re}\Psi
- \textstyle{\frac{1}{2}}\, N_t^{-1}\,\left(\partial_t\lambda_t  + 
\textstyle{\frac{7^3}{120}}\,\tau_0^2\, N_t^2 + \textstyle{\frac{3}{2}}\,\lambda_t^2
\right)\, {\rm Im}\Psi~. 
\end{align*}
We remark that this variation is harmonic.  

In conclusion, we see that the Calabi--Yau flow requires that $N_t$ is a constant and that the primitive form $\gamma$ is harmonic.
Moreover
\begin{align*}
\partial_t\omega &= \lambda_t\, \omega~,\qquad \dd\lambda_t = 0~,\\[3pt]
\partial_t\Psi &= K_t\, \Psi + \chi_t, \qquad \dd K_t = 0~,\\
&{\rm Re}\chi_t = - \gamma~,\qquad {\rm Im}\chi_t = J\gamma~,\qquad\omega\lrcorner\gamma = 0~.
\end{align*}

\subsection{Flow from Calabi--Yau:  first order analysis}
As shown in the last section, a non-constant $N_t$ and/or a non-harmonic $\gamma$, implies that the $G_2$ embedding will not preserve the Calabi--Yau conditions, and $SU(3)$ torsion will be generated by the flow. Here, we determine the torsion classes to first order in $t$.

From the general analysis, the flow equations are given by \eqref{eq:generalflow}. Recall that in the gauge $\alpha = 1$ we have $||\Psi||^2 = 8$ for all values of $t$, so ${\rm Re}K_t = \frac{3}{2}\, \lambda_t$.  The variations of the torsion classes must be such that they preserve their original properties:
\begin{align*}
(\partial_t W_2)^{(0,2)} &= \Delta_t{}^m \wedge W_2{}_{mn}\, \dd x^n~,&\quad\hbox{to preserve type $(1,1)$}\\
\omega\lrcorner(\partial_t W_2)^{(1,1)} &= - N_t\, W_2\lrcorner{\rm Im}W_2~,
&\quad\hbox{to preserve primitivity}\\[15pt]
(\partial_t W_3)^{(0,3)} &= \textstyle{\frac{1}{2}}\, \Delta_t^m\wedge (W_3^{(1,2)})_{mnp}\, \dd x^n\wedge\dd x^p~,
&\quad\hbox{to preserve type $(2,1)+(1,2)$}\\
\omega\lrcorner(\partial_t W_3) &= J(h_t) \lrcorner W_3~,&\quad\hbox{to preserve primitivity}\; 
\end{align*} 
where $\Delta_t{}^m$  is given in \eqref{eq:DeltafromPsi} and corresponds to first order variations of the almost complex structure. Note that any primitive three form, such as $\gamma$, satisfies equations identical to those for $W_3$.  

Let $\beta$ be any form in our equations above.  We will consider a Taylor series expansion 
\[\beta(t) = \sum_{i = 0}^{\infty}\, \frac{1}{i!}\, \delta_i\beta\, t^i~,\]
where we have set
\[ \beta_0 = \delta_0\beta~.\]
\vskip10pt

At $t=0$ the equations above give
\begin{align}
\delta_1\Psi &= K_0\, \Psi_0 + \chi_0~,\qquad
{\rm Im} K_0 = \textstyle{\frac{7}{4}}\,N_0\, \tau_0 ~,\\
\chi_0 &= 2\, \left(\dd N_0 \wedge\omega_0 - N_0\,\gamma_0\right)^{(2,1)}~,
\quad \omega_0\lrcorner\gamma_0 = 0~,\quad \Psi_0\wedge\gamma_0 = 0~,\label{eq:chi0}\\
\Delta_0^m &=  \textstyle\frac{1}{8}\, \overline\Psi_0^{mpq}\, \big(2\, (\partial_p N_0)\, \omega_{0\, qn}
- N_0\, \gamma_{0\, pqn}^{(2,1)}\big)\, \dd x^n~,\label{eq:Delta0}\\
\eta_0 &= 0~,\\[10pt]
\delta_1 \omega &= \lambda_0\,\omega + h_0~,\qquad
\lambda_0 = \tau_{1\, t}|_0 = \delta_1\phi~,\\ 
h_0 &=  - {\rm Re}\left(\partial N_0\lrcorner\bar\Psi_0\right)~.\label{eq:h0}
\end{align}
Note also that  $\delta_1 W_2$, $\delta_1 W_3$ are primitive with respect to $\omega_0$,  that $\delta_1 W_2$ is type $(1,1)$  and $\delta_1 W_3$ is type $(2,1) + (1,2)$ with respect to $\Psi_0$. We have included $\Delta_0$ which may be needed later.

Varying the integrability equations for $\Psi$ (see equation \eqref{eq:eqchit}) and evaluating at  $t=0$ we find a differential equation for $\chi_0$
\begin{equation}
\dd\chi_0= (\delta_1 \overline W_1{}^{\Psi} - \dd K_0)\wedge\Psi_0 + \delta_1 W_0\, \omega_0\wedge\omega_0 + \delta_1 W_2\wedge\omega_0~,
\label{eq:dchi0}
\end{equation}
The expression for $\chi_0$ in \eqref{eq:chi0} must satisfy equation \eqref{eq:dchi0}.  We find
\begin{equation}
\begin{split}
-2 \, \dd(N_0\, \gamma_0)^{(2,1)} &= \delta_1 W_0\, \omega_0\wedge\omega_0 
+ (\delta_1 W_2 - 2 \bar\partial\partial N_0)\wedge\omega_0 
+ (\delta_1 \overline W_1^{\Psi} - \dd K_0)\wedge\Psi_0\\
&= (\delta_1 W_0\, \omega_0 - 2 \bar\partial\partial N_0 + \delta_1 W_2)\wedge\omega_0 
+ \left(\delta_1 \overline W_1^{\Psi} - \textstyle{\frac{7}{4}}\, i\, \tau_0\,\dd N_0\right)\wedge\Psi_0
~.\label{eq:intgammaone}
\end{split}
\end{equation}
Taking the wedge product of this equation with $\omega_0$ and recalling that $\gamma$ and $\delta_1 W_2$ are primitive, we obtain, after taking the Hodge dual (with respect to $\omega_0$), an equation for $\delta_1 W_0$:
\begin{equation*}
\frac{3}{2}\, \delta_1 W_0 = \omega_0\lrcorner(\bar\partial\partial N_0) 
= - \frac{i}{2}\, \omega_0\lrcorner\dd(J(\dd N_0))~.
\end{equation*}
Using the identity \eqref{idone} for $\alpha = J(\dd N_0)$ we have
\begin{equation}
\frac{3}{2}\, \delta_1 W_0 = - \frac{i}{2}\, \dd^{\dagger_0}\dd N_0 ~.\label{eq:DelW0}
\end{equation}
Putting this expression back into equation \eqref{eq:intgammaone}, we find
\begin{equation}
\begin{split}
-2 \, \dd(N_0\, \gamma_0)^{(2,1)} &= 
\left( i\, \left(- \textstyle{\frac{1}{3}}\,  (\dd^{\dagger_0}\dd\, N_0)\, \omega_0 +  \dd(J(\dd N_0))\right) 
+ \delta_1 W_2\right)\wedge\omega_0\\ 
&\qquad\qquad + \left(\delta_1 \overline W_1^{\Psi} - \textstyle{\frac{7}{4}}\, i\, \tau_0\,\dd N_0\right)\wedge\Psi_0
~.\label{eq:intgammatwo}
\end{split}
\end{equation}
This constraint  can be separated by type giving
\begin{align}
\partial(N_0\, \gamma_0)^{(2,1)} &= \textstyle{\frac{1}{2}}\left(\textstyle{\frac{7}{4}}\, i\, \tau_0\,\dd N_0
- \delta_1 \overline W_1^{\Psi} \right)\wedge\Psi_0~,\label{eq:intgamma31}\\[5pt]
\bar\partial(N_0\, \gamma_0)^{(2,1)} &= 
\textstyle{\frac{1}{2}}
\left( i\, \left(\textstyle{\frac{1}{3}}\,  (\dd^{\dagger_0}\dd\, N_0)\, \omega_0 -  \dd(J(\dd N_0))\right) - \delta_1 W_2\right)\wedge\omega_0
~,\label{eq:intgamma22}
\end{align}
which can be used to write  expressions for $\delta_1 \overline W_1^\Psi$ and $\delta_1 W_2$
\begin{align}
(\delta_1 \overline W_1^{\Psi})^{(0,1)} &=
\textstyle{\frac{1}{4}}\,\overline\Psi_0\lrcorner\partial(N_0\, \gamma_0)^{(2,1)}
+ \textstyle{\frac{7}{4}}\, i\, \tau_0\,\bar\partial N_0~,\label{eq:DelW1Psi}\\
\delta_1 W_2 &= -2\, \omega_0\lrcorner \bar\partial(N_0\, \gamma_0)^{(2,1)} 
+  i\, \left(\textstyle{\frac{1}{3}}\,  (\dd^{\dagger_0}\dd\, N_0)\, \omega_0 -  \dd(J(\dd N_0))\right)~.\label{eq:DelW2} 
\end{align}
Note that $W_1^\Psi$ at any $t$ is a $(1,0)$-form with respect to $\Psi(t)$.  Hence 
\[\partial_t W_1^\Psi = \Delta_t^m \, W_1^\Psi{}_m~\]
is a $(0,1)$ form.  To first order in $t$, this means that 
\[(\delta_1 \overline W_1^\Psi)^{(1,0)} = 0~.\]

Consider now the integrability equation for $\omega$ (see equation \eqref{eq:eqvaromegat}).  Evaluating at $t=0$ we find a differential equation for $h_0$
\begin{equation*}
\dd\delta_1\omega = \dd h_0 = - \frac{3}{2}\, {\rm Im}(\delta_1 W_0\, \overline\Psi_0) + \delta_1 W_1^{\omega}\wedge\omega_0 
+ \delta_1 W_3~.
\end{equation*}
The expression for $h_0$ in \eqref{eq:h0} must satisfy this equation.  Hence
\begin{equation}
- \dd{\rm Re}(\bar\partial N_0 \lrcorner \Psi_0) = 
- \frac{3}{2}\, {\rm Im}(\delta_1 W_0\, \overline\Psi_0) + \delta_1 W_1^{\omega}\wedge\omega_0 
+ \delta_1 W_3~.\label{eq:DelW3pre}
\end{equation}
Taking the wedge product of this equation with $\omega_0$ and using the fact that $\delta_1 W_3$ is primitive we find
\begin{equation}
\delta_1 W_1^\omega = 0~.\label{eq:DelW1o}
\end{equation}
Equation \eqref{eq:DelW3pre} then becomes
\begin{equation*}
 \dd\dd^{\dagger_0}\,{\rm Re}(N_0 \Psi_0) = 
- \frac{3}{2}\, {\rm Im}(\delta_1 W_0\, \overline\Psi_0)  + \delta_1 W_3~.
\end{equation*}
where in the left hand side we have used equation \eqref{eq:idtwo}. The $(2,1) + (1,2)$ part of this equation gives an expression for $\delta_1 W_3$
\begin{equation}
2\, \delta_1 W_3 = \bar\partial\partial^{\dagger_0}(N_0\, \Psi) + \partial\bar\partial^{\dagger_0}(N_0\, \overline\Psi)
~,\label{eq:DelW3}
\end{equation}
and the $(0,3)$ part gives another expression for $\delta_1 W_0$
\begin{equation}
\frac{3}{2}\, \delta_1 W_0 = -  \frac{i}{8}\, \Psi_0\lrcorner(\dd\dd^{\dagger_0} (N_0\overline\Psi_0 ))~.\label{eq:W0bis}
\end{equation}

To show that equation \eqref{eq:W0bis} is equivalent to \eqref{eq:DelW0} we use the identity \eqref{eq:idthree} with
\[\beta = \dd^{\dagger_0}(N_0\overline\Psi_0) = - \dd N_0\lrcorner\overline\Psi_0~,\] 
we find that \eqref{eq:W0bis} becomes
\begin{equation*}
\frac{3}{2}\, \delta_1 W_0 =-  \frac{i}{8}\, \dd^{\dagger_0}( (\dd N_0\lrcorner\overline\Psi)\lrcorner\Psi)
= -i\, \dd^{\dagger_0}(\partial N_0)~,
\end{equation*}
which is equivalent to equation \eqref{eq:DelW0}.

In summary we have the following equations for the torsion classes
\begin{align*}
\delta_1 W_0 &=  - \textstyle{\frac{i}{3}}\, \dd^{\dagger_0}\dd\, N_0~,\\
\delta_1 W_1^\omega &= 0~,\\
\delta_1 \overline W_1^{\Psi}&=
\textstyle{\frac{1}{4}}\,\overline\Psi_0\lrcorner\partial(N_0\, \gamma_0^{(2,1)})
+ \textstyle{\frac{7}{4}}\, i\, \tau_0\,\bar\partial N_0~,\\
\delta_1 W_2 &= -2\, \omega_0\lrcorner \bar\partial(N_0\, \gamma_0^{(2,1)}) 
+  i\, \left(\textstyle{\frac{1}{3}}\,  (\dd^{\dagger_0}\dd\, N_0)\, \omega_0 -  \dd(J(\dd N_0))\right)~,\\
\delta_1 W_3 &= \textstyle{\frac{1}{2}}\,\left(\bar\partial\partial^{\dagger_0}(N_0\, \Psi) 
+ \partial\bar\partial^{\dagger_0}(N_0\, \overline\Psi)\right)~.
\end{align*}

The Bianchi identities at $t=0$ are
\begin{align}
\dd S_0^X &= 0~,\label{eq:BIone}\\
 \delta_1 S^X &= \dd S_0~,\label{eq:BItwo}
\end{align}
where
\begin{align*}
S^X_0 &= - \frac{49}{60}\, \tau_0\, {\rm Re}\Psi_0 + \frac{1}{2}\, \lambda_0\, N_0^{-1}\, {\rm Im}\Psi_0 + J(\gamma_0)~,\\[5pt]
S_0 &= \frac{14}{15}\, \tau_0\, N_0\, \omega_0~,\\[5pt]
\delta_1 S^X &= 2\ {\rm Im}\delta_1 \overline W_1^{\Psi}\wedge\omega_0 + \delta_1(N_t^{-1}\, \tau_{1\, t})\, {\rm Im}\Psi_0
- \frac{49}{60}\, \tau_0\, {\rm Re}\delta_1\Psi + \frac{1}{2}\, \lambda_0\, N_0^{-1}\, \delta_1{\rm Im}\Psi + \delta_1(J\gamma)~.
\end{align*}

The first Bianchi identity \eqref{eq:BIone} gives
\[ \dd(J\gamma_0) =  \frac{1}{2}\, \lambda_0\, N_0^{-2}\, \dd N_0\wedge{\rm Im}\Psi_0~.\]
Separating this equation by type, we find
\begin{align}
&(3,1){\rm -part:}&\qquad\partial\gamma_0^{(2,1)} &= \bar\partial \left(\textstyle{\frac{1}{4}}\, \lambda_0\, N_0^{-1}\, \Psi_0\right)
= - \textstyle{\frac{1}{4}}\, \lambda_0\, N_0^{-2}\, \bar\partial N_0\wedge\Psi_0~,\label{eq:pargamma21}\\[5pt]
&(2,2){\rm -part:}&\qquad\bar\partial\gamma_0^{(2,1)} &- \partial \gamma_0^{(1,2)} = 0~.
\end{align}

The second Bianchi identity \eqref{eq:BItwo} is more involved and requires some preliminary computations.  
This constraint is an equation for $\delta_1(J\gamma)$, which, by primitivity, satisfies the constraints for $\delta_1 W_3$ presented at the first page of this section. Applying these to $J\gamma$ we find:
 \begin{align}
(\delta_1(J\gamma))^{(0,3)} &= - \textstyle{\frac{i}{2}}\, \Delta_0^m\wedge (\gamma_0^{(1,2)})_{mnp}\, \dd x^n\wedge\dd x^p\\
&\qquad\qquad\Longrightarrow~ 
\Psi_0\lrcorner\delta_1(J\gamma) = - \textstyle{\frac{i}{2}}\, \Delta_{0\, m}{}^q\, \Psi_0^{mnp}\, (\gamma_0^{(1,2)})_{npq}
~,\label{eq:DeltaJgamma03}\\[10pt]
\omega_0\lrcorner\delta_1(J\gamma) &= - h_0 \lrcorner (J\gamma)_0 
~\Longrightarrow~
\omega_0\lrcorner\delta_1(J\gamma) = 4\, N_0^{-1}\, {\rm Im}(\Delta_0^m\, \partial_m N_0)~,\label{eq:DeltaJgammaO}
\end{align}

The Bianchi identity \eqref{eq:BItwo} can be written as an equation for the change in $\gamma$. Using previous results we have
\begin{align*}
\delta_1(J\gamma) &= 
\big(\textstyle{\frac{7}{4}}\, \tau_0\, \dd N_0  +  \textstyle{\frac{1}{2}}\, \lambda_0\, N_0^{-1}\,J(\dd N_0)
- 2\ {\rm Im}(\delta_1 \overline W_1^{\Psi})\big)\wedge \omega_0
 - \textstyle{\frac{49}{60}}\, \tau_0\, N_0\, \gamma_0 - \textstyle{\frac{1}{2}}\, \lambda_0\, (J\gamma)_0\\
&\quad + \textstyle{\frac{7}{20}}\, \tau_0\,\lambda_0 \, {\rm Re}\Psi_0
- \big( \textstyle{\frac{7^3}{240}}\, \tau_0^2\, N_0 
+ \textstyle{\frac{3}{4}}\, \lambda_0^2\, N_0^{-1} + \delta_1(N_t^{-1}\, \tau_{1\, t})\big)\, {\rm Im}\Psi_0~.
\end{align*}
Finally we compute ${\rm Im}(\delta_1\overline W_1^\Psi)$.  Using equations \eqref{eq:Delta0} and \eqref{eq:pargamma21} in the first term in $\delta_1 \overline W_1^\Psi$ we have
\begin{align*} 
\overline\Psi_0\lrcorner\partial(N_0\, \gamma_0)^{(2,1)} &=
\overline\Psi_0\lrcorner(\partial N_0\wedge\gamma_0^{(2,1)}
  + N_0\, \partial\gamma_0^{(2,1)})~, \\
\overline\Psi_0\lrcorner(\partial N_0\wedge\gamma_0^{(2,1)})&= 
   4\, (\partial_m N_0) \, \overline\Psi_0^{mnp}\, \gamma_{0\,npq}\, \dd x^q
   = - 4\, N_0^{-1}\, (\partial_m N_0) \, \Delta_0^m \, , \\
   \overline\Psi_0\lrcorner( N_0\, \partial\gamma_0^{(2,1)}) &=
   2\, \lambda_0\, N_0^{-1}\, \bar\partial N_0~.
  \end{align*}
Therefore
\begin{equation} 
2\, \delta_1 \overline W_1^{\Psi}=
 - 2\, N_0^{-1}\, (\partial_m N_0) \, \Delta_0^m
+  \left( \lambda_0\, N_0^{-1} +\frac{7}{2}\, i\, \tau_0\right)\,\bar\partial N_0~,
\label{eq:newDeltaW1Psi}
\end{equation}
and
\begin{equation} 
2\, {\rm Im} \delta_1 \overline W_1^{\Psi}=
 - 2\, N_0^{-1}\, {\rm Im}\big((\partial_m N_0) \, \Delta_0^m)
+ \frac{1}{2}\,\lambda_0\, N_0^{-1} J(\dd N_0) + \frac{7}{4}\, \tau_0\, \dd N_0~.
\label{eq:ImDeltaW1Psi}
\end{equation}
Inserting equation \eqref{eq:ImDeltaW1Psi} into our expression above for the variation of $J\gamma$ we find
\begin{equation}
\begin{split}
\delta_1(J\gamma) &= 
2\, N_0^{-1}\, {\rm Im}(\Delta_0^m\, \partial_m N_0)\wedge \omega_0
 - \textstyle{\frac{49}{60}}\, \tau_0\, N_0\, \gamma_0 - \textstyle{\frac{1}{2}}\, \lambda_0\, (J\gamma)_0\\
&\quad + \textstyle{\frac{7}{20}}\, \tau_0\,\lambda_0\, {\rm Re}\Psi_0
- \big( \textstyle{\frac{7^3}{240}}\, \tau_0^2\, N_0 
+ \textstyle{\frac{3}{4}}\, \lambda_0^2\, N_0^{-1} + \delta_1(N_t^{-1}\, \tau_{1\, t})\big)\, {\rm Im}\Psi_0
~.\label{eq:DeltaJgamma}
\end{split}
\end{equation}

This result must be consistent with equations \eqref{eq:DeltaJgamma03} and \eqref{eq:DeltaJgammaO}.  Consistency with \eqref{eq:DeltaJgammaO} is obvious.  The consistency with \eqref{eq:DeltaJgamma03} is however nontrivial and rather nice as it gives an expression for trace of the first order metric $\Delta_{0\, m}{}^n\, \bar\Delta_{0\, n}{}^m$  on the moduli space.  Contracting \eqref{eq:DeltaJgamma} with $\Psi_0$ and comparing with \eqref{eq:DeltaJgamma03} we find
\begin{align*}
\Psi_0\lrcorner\delta_1(J\gamma) &= 
- \textstyle{\frac{i}{2}}\, \Delta_{0\, m}{}^q\, \Psi_0^{mnp}\, (\gamma_0^{(1,2)})_{npq}\\
=& \textstyle{\frac{7}{5}}\, \tau_0\,\lambda_0
- i\, \big( \textstyle{\frac{7^3}{60}}\, \tau_0^2\, N_0 
+ 3\, \lambda_0^2\, N_0^{-1} + 4\,\delta_1(N_t^{-1}\, \tau_{1\, t})\big)~.
\end{align*}
Using equation \eqref{eq:Delta0} we have
\begin{align*}
- N_0\, \delta_1(N_t^{-1}\, \tau_{1\, t})  &= \Delta_{0\, m}{}^n\, \bar\Delta_{0\, n}{}^m 
- \textstyle{\frac{1}{4}}\, \Delta_{0\, m}{}^n\, \Psi_0^{mpq}\, (\partial_p N_0)\,\omega_{0\, qn}\\
&\qquad + \textstyle{\frac{7}{20}}\, i \tau_0\, N_0\, \lambda_0
+ \textstyle{\frac{7^3}{240}}\, \tau_0^2\, N_0^2
+ \textstyle{\frac{3}{4}}\, \lambda_0^2~.
\end{align*}

\subsection{$SU(3)$ structure at first order}
\label{sec:o1sum}

Our analysis shows that, when embedded in an integrable $G_2$ structure, a Calabi--Yau threefold $X_{0}$ may flow to an $SU(3)$ structure manifold $X_{\delta t}$ with non-vanishing torsion, where the latter is determined by the $G_2$ torsion classes. In summary, we find
\begin{align*}
\delta_1 W_0 &=  - \textstyle{\frac{i}{3}}\, \dd^{\dagger_0}\dd\, N_0~,\\
\delta_1 W_1^\omega &= 0~,\\
\delta_1 \overline W_1^{\Psi}&=
 -  N_0^{-1}\, (\partial_m N_0) \, \Delta_0^m
+ \textstyle{\frac{1}{2}}\, \left( \lambda_0\, N_0^{-1} +\textstyle{\frac{7}{2}}\, i\, \tau_0\right)\,\bar\partial N_0~,\\
\delta_1 W_2 &= -2\, \omega_0\lrcorner \bar\partial(N_0\, \gamma_0)^{(2,1)} 
+  i\, \left(\textstyle{\frac{1}{3}}\,  (\dd^{\dagger_0}\dd\, N_0)\, \omega_0 -  \dd(J(\dd N_0))\right)~,\\
\delta_1 W_3 &= \textstyle{\frac{1}{2}}\,\left(\bar\partial\partial^{\dagger_0}(N_0\, \Psi) 
+ \partial\bar\partial^{\dagger_0}(N_0\, \overline\Psi)\right)~,\\\\
\partial\gamma_0^{(2,1)} &= \bar\partial \left(\textstyle{\frac{1}{4}}\, \lambda_0\, N_0^{-1}\, \Psi_0\right)
= - \textstyle{\frac{1}{4}}\, \lambda_0\, N_0^{-2}\, \bar\partial N_0\wedge\Psi_0~,\label{eq:pargamma21}\\[5pt]
\bar\partial\gamma_0^{(2,1)} &= \partial \gamma_0^{(1,2)}~,\\\\
\delta_1(J\gamma) &= 
2\, N_0^{-1}\, {\rm Im}(\Delta_0^m\, \partial_m N_0)\wedge \omega_0
 - \textstyle{\frac{49}{60}}\, \tau_0\, N_0\, \gamma_0 - \textstyle{\frac{1}{2}}\, \lambda_0\, (J\gamma)_0
 + \textstyle{\frac{7}{20}}\, \tau_0\,\lambda_0\, \Psi_0\\
&\quad + N_0^{-1}\, \big( \Delta_{0\, m}{}^n\, \bar\Delta_{0\, n}{}^m 
- \textstyle{\frac{1}{4}}\, \Delta_{0\, m}{}^n\, \Psi_0^{mpq}\, (\partial_p N_0)\,\omega_{0\, qn}
\big)\, {\rm Im}\Psi_0~,\\\\
\delta_1(N_t^{-1}\, \tau_{1\, t}) &=  N_0^{-1}\, \left(- \Delta_{0\, m}{}^n\, \bar\Delta_{0\, n}{}^m 
+ \textstyle{\frac{1}{4}}\, \Delta_{0\, m}{}^n\, \Psi_0^{mpq}\, (\partial_p N_0)\,\omega_{0\, qn}\right)\\
&\qquad - \textstyle{\frac{7}{20}}\, i \tau_0\,  \lambda_0
- \textstyle{\frac{7^3}{240}}\, \tau_0^2\, N_0
- \textstyle{\frac{3}{4}}\, \lambda_0^2\, N_0^{-1}~.
\end{align*}

There are several interesting observations to make:
\begin{itemize}
\item $\delta_1 W_1^{\omega}=0 $; this is the only torsion class that cannot, at linear order in $t$, be generated by the flow.
\item If $X_{\delta t}$ is complex, then it is Calabi--Yau. This follows since $\delta_1 W_0 $ and $\delta_1 W_2$ vanish if and only if $N_0$ is constant and $\gamma$ is harmonic. In this case all other torsion classes vanish as well. It follows that the flow cannot connect Calabi--Yau solutions with the conformally balanced non-K\"ahler manifolds of the Strominger system.
\item Suppose that at $t=0$ we set
\be
\tau_0 = 0 \quad , \quad \lambda_0 = 0 \quad , \quad \gamma_0 = 0\; ,
\ee
but keep $N_0$ non-constant. Then $X_{\delta t}$ has a half-flat $SU(3)$ structure: the two Lie forms vanish, and $\delta_1 W_0 $ and $\delta_1 W_2$ are imaginary
\be
\label{eq:CYtoHF}
\begin{split}
\delta_1 W_0 &=  - \textstyle{\frac{i}{3}}\, \dd^{\dagger_0}\dd\, N_0~,\\
\delta_1 W_1^\omega &= 0~,\\
\delta_1 \overline W_1^{\Psi}&=0~,\\
\delta_1 W_2 &=  i\, \left(\textstyle{\frac{1}{3}}\,  (\dd^{\dagger_0}\dd\, N_0)\, \omega_0 -  \dd(J(\dd N_0))\right)~,\\
\delta_1 W_3 &= \textstyle{\frac{1}{2}}\,\left(\bar\partial\partial^\dagger(N_0\, \Psi) 
+ \partial\bar\partial^\dagger(N_0\, \overline\Psi)\right)~.
\end{split}
\ee
To first order in $t$, there is no flux $S_0^X = S_0 =0$, and the dilaton is constant. At linear order, we thus have a $G_2$ holonomy manifold.

\item Suppose that at $t=0$ $\gamma_0$ is non-zero, but $N_0$ is $\dd$-constant. Then 
\be
\begin{split}
\delta_1 W_i &= 0 \; , \; \forall i \neq 2 \; , \\
\delta_1 W_2 &= - 2 N_0 \omega_0 \lrcorner \bar{\partial} \gamma_0^{(2,1)} \; , \\
\partial\gamma_0^{(2,1)} &= 0 \; , \; 
\bar\partial\gamma_0^{(2,1)} = \partial \gamma_0^{(1,2)} \; .
\end{split}
\label{eq:CYtosympHF}
\ee
Note that $\delta_1 W_2$ is real and primitive by construction, but non-zero for generic $\gamma_0$. Thus, $X_{\delta t}$ is a symplectic half-flat $SU(3)$ structure manifold. The flux is non-zero at linear order, and the dilaton is non-constant.  
\end{itemize}

There are two options for the study of the integrability of the infinitesimal flow away from Calabi--Yau derived in the last subsection. First, we could continue the perturbative analysis to higher orders, and complement it with an inductive proof of integrability similar to that of Tian for the integrability of Calabi--Yau preserving deformations \cite{tian86}. Second, we can provide arguments for the integrability of the flow by studying whether the flow of half-flat $SU(3)$ structures allow Calabi--Yau loci, and could thus connect to the flow we have found. After a detour over nearly K\"ahler flows, we will proceed along the second route.
\vspace{0.5cm}


\section{Nearly K\"ahler manifolds}
\label{sec:NK}

Consider the flow of manifolds which are nearly K\"ahler, that is, we set $W_i = 0$ for all $i\ne0$.  As before, we set $\alpha = 1$.
As we will see below, we are able to completely solve for this case.  We show that we necessarily have that $N_t$ is constant, $\tau_0 = 0$, and that the forms $h_t$ and $\gamma$ (and therefore $\chi_t$) vanish.  Moreover, we will prove that the $G_2$-flow of nearly K\"ahler manifolds with ${\rm Im} W_0 = 0$ is not allowed (otherwise we fall back to a Calabi--Yau flow), and that a consistent flow has either ${\rm Re} W_0 = 0$, or the complex phase of $W_0$ needs to vary with the flow parameter $t$.\footnote{The flow of nearly K\"ahler manifolds when $W_0$ has constant phase has recently been discussed in \cite{Klaput:2012vv,Haupt:2014ufa}.}  The two cases do not intersect each other, however the case where the phase of $W_0$ varies with $t$ approaches asymptotically the case where ${\rm Re} W_0 = 0$. While they both flow into a Calabi--Yau manifold at infinity, there are no Calabi--Yau loci at finite $t$ along a nearly K\"ahler flow.

The equations for the $SU(3)$ structure are:
\begin{align}
\dd\omega &=- \textstyle{\frac{3}{2}}\, {\rm Im}(W_0\, \overline\Psi)\label{eq:domegaNK}~,\\
\dd\Psi &= W_0\, \omega\wedge\omega\label{eq:dPsiNK}~.
\end{align}
The only non-zero torsion class $W_0$ is $\dd$-constant
\[\dd W_0 = 0~,\]
as can be seen by taking the exterior derivative of equation \eqref{eq:dPsiNK}.

The variations of the hermitian structure are:
\begin{align*}
\partial_t\omega &= \lambda_t\,\omega + h_t~, \qquad\dd\lambda_t = 0~,\\
\lambda_t &= 2\, \tau_{1\, t} - N_t\, {\rm Im}W_0~,\qquad 2\, \tau_{1\, t} = \partial_t\phi~,\qquad \dd\tau_{1\, t} = 0~,\\
h_t &= -\dd N_t\lrcorner {\rm Re\Psi}~.
\end{align*}
Because $\lambda_t$, $\tau_{1\, t}$ and $W_0$ are $\dd$-constant, the second equation implies that 
\begin{equation}
 {\rm Im}W_0\, \dd N_t= 0~,\label{eq:NKcondition}
 \end{equation}
and therefore ${\rm Im}W_0= 0$ or $N_t$ is a $\dd$-constant. 

Condition \eqref{eq:NKcondition} implies that we do not expect that a Calabi--Yau manifold can flow into a nearly K\"ahler manifold, except perhaps at infinite distances.  To see this, recall that in the first order analysis for the flow from a Calabi--Yau manifold at $t=0$, we found
\[ \delta_1 W_0 = - \frac{i}{3}\, \dd^{\dagger_0}\dd N_0~,\]
which is imaginary.  For ${\rm Im} W_0$ to be non-zero to first order, and hence have a flow into a nearly K\"ahler manifold, it must be the  case that $N_0$ is not constant.  However, when ${\rm Im} W_0\ne0$, condition \eqref{eq:NKcondition} requires that $\dd N_t = 0.$
 It would be interesting to understand whether one can flow from a nearly 
K\"ahler manifold into a half-flat manifold. 

In the Appendix \ref{app:NKimW0}, we prove in fact that ${\rm Im}W_0 \ne 0$, otherwise we fall back into the flow of a Calabi--Yau manifold. Therefore, with $\alpha=1$, the {\it $G_2$ flow of nearly K\"ahler manifolds with ${\rm Im}W_0= 0$ is not allowed.}  From now on we assume that ${\rm Im} W_0\ne 0$,
and hence, by equation \eqref{eq:NKcondition}, we must have
\[ \dd N_t = 0~.\]
We can absorb $N_t$ into the definition of $\dd t$ and choose $N_t$ to be constant. Note that as a consequence $h_t$ vanishes.

The variations of the almost complex structure are 
\begin{align*}
\partial_t\Psi &= K_t\, \Psi + \chi_t~,\qquad \dd K_t = 0~,\\
{\rm Re}K_t &= \textstyle{\frac{3}{2}}\, \lambda_t~, \qquad
{\rm Im}K_t = \textstyle{\frac{1}{4}}\, N_t\, (7\, \tau_0 - 12\, {\rm Re}W_0)~,\\
\chi_t &= - 2\, N_t\, \gamma^{(2,1)}~,\quad \omega\lrcorner\gamma = 0~.
\end{align*}
We begin our analysis by considering the variations of equation \eqref{eq:dPsiNK}, that is
\begin{equation}
\dd\partial_t\Psi = (\partial_t W_0 + 2 \lambda_t\, W_0)\, \omega\wedge\omega ~.
\label{eq:PrevardPsi}
\end{equation}
Taking the wedge product with $\omega$ and recalling that $\gamma$ is a primitive form we find a simple equation for the flow of $W_0$
\begin{equation}
\partial_t W_0 + (2\lambda_t - K_t)\, W_0 = 0~.\label{eq:varW0NK}
\end{equation}
Returning to equation \eqref{eq:PrevardPsi},  and using these results we obtain
\begin{equation}
\dd\partial_t\Psi =  K_t\, W_0\, \omega\wedge\omega~,\label{eq:vardPsiOne}
\end{equation}

Consider now the variation equations for the hermitian form $\omega$. The  flow equations become
\begin{equation}
\partial_t\omega = \lambda_t\,\omega~,\qquad \dd\lambda_t = 0~.\label{eq:varomegaNK}
\end{equation}
Compatibility of the variation of equation \eqref{eq:domegaNK}. 
\begin{align*}
\dd\partial_t\omega 
&= - \textstyle{\frac{3}{2}}\, {\rm Im}\left(
(\partial_t W_0 + \bar K_t\,W_0) \, \overline\Psi + W_0\, \overline \chi_t\right)~,
\end{align*}
with the exterior derivative of  \eqref{eq:varomegaNK}
\[\dd\partial_t\omega = \lambda_t\, \dd\omega = - \textstyle{\frac{3}{2}}\,\lambda_t \,{\rm Im}(W_0\,\overline\Psi)~,\]
gives 
\begin{equation*}
{\rm Im}\left(
\left(\partial_t W_0 
+ \big(\textstyle{\frac{1}{2}}\, \lambda_t 
- \textstyle{\frac{i}{4}}\, N_t\, (7\, \tau_0 - 12\, {\rm Re}W_0)\big)\,W_0
\right) \, \overline\Psi 
-2 W_0\, N_t\, \gamma^{(1,2)}
\right)=0~.
\end{equation*}
Separating by type we obtain an equation for the flow of $W_0$ which is the same as equation \eqref{eq:varW0NK}
and the constraint
\begin{equation}
\gamma = 0~.
\end{equation}
Consequently, the variations of the complex structure are given by
\begin{equation} \partial_t\Psi = K_t\Psi~, \qquad K_t =  \frac{3}{2}\, \lambda_t + \frac{i}{4}\, N_t \, (7\tau_0 - 12{\rm Re} W_0)~.
\label{eq:varPsiNK}
\end{equation}
It is not very hard to check that equation \eqref{eq:varW0NK} is enough to guarantee the compatibility of \eqref{eq:varPsiNK} and 
\eqref{eq:vardPsiOne}.

Next, we consider the Bianchi identities.  For the flux in this case we have 
\begin{align*}
S^X &= 
- \textstyle{\frac{1}{2}}\, \left(\textstyle{\frac{49}{30}} \, \tau_0
- 3 {\rm Re}W_0\right) {\rm Re}\Psi
+ N_t^{-1}\, \tau_{1\, t}\, {\rm Im}\Psi~,\\[5pt]
 S_t &= 2\, N_t\, \left(\textstyle{\frac{7}{15}}\, \tau_0 - {\rm Re}W_0\right)\omega~.
\end{align*}

The first Bianchi identity
\begin{equation*}
\dd S^X = 0 = \left(
- \textstyle{\frac{1}{2}}\, \left(\textstyle{\frac{49}{30}} \, \tau_0
- 3 {\rm Re}W_0\right) {\rm Re}W_0
+ N_t^{-1}\, \tau_{1\, t}\, {\rm Im}W_0
\right) \omega\wedge\omega~,
\end{equation*}
gives a relation between the embedding parameters and the torsion class $W_0$
\begin{equation}
- N_t\, \left(\textstyle{\frac{49}{30}} \, \tau_0
- 3 {\rm Re}W_0\right) {\rm Re}W_0
+ (\lambda_t + N_t\, {\rm Im}W_0)\, {\rm Im}W_0 = 0~.\label{eq:relEmbedd}
\end{equation}

The second Bianchi identity
\[ \dd S_t = \partial_t S^X~,\]
gives two further constraints
\begin{align}
0 &= \tau_0\left(\lambda_t + N_t\, {\rm Im}W_0\right)~,\label{eq:tauornot}\\[3pt]
\partial_t\lambda_t &= - \textstyle{\frac{3}{2}}\,\lambda_t^2  - \lambda_t\, N_t\, {\rm Im}W_0 
+ \textstyle{\frac{56}{5}}\,N_t^2\, \tau_0\, {\rm Re}W_0 - 12 (N_t\, {\rm Re}W_0)^2
- \textstyle{\frac{7^3}{120}}\, (N_t\, \tau_0)^2~.\label{varlambda}
\end{align}
We claim that the first equation implies that $\tau_0 = 0$.  Suppose on the contrary that $\tau_0\ne 0$. Then, equation \eqref{eq:tauornot} requires
\[ \lambda_t = - N_t\, {\rm Im} W_0~.\]
Substituting this into the relation \eqref{eq:relEmbedd} implies 
\begin{equation} 
{\rm Re} W_0 = \frac{49}{90}\, \tau_0~\qquad{\rm or}\qquad
{\rm Re} W_0 = 0~. \label{eq:preconstReW0}
\end{equation}
In both cases this means that
$\partial_t{\rm Re} W_0 = 0$. Taking the real part of the variation of $W_0$ in equation \eqref{eq:varW0NK} and setting this to zero, we find
\[ {\rm Re} W_0 = \frac{1}{2}\, \tau_0~,\]
which is not compatible with \eqref{eq:preconstReW0} unless $\tau_0$ vanishes.

It is worth summarising our results thus far.
The equations of the flow are
\begin{align}
\label{eq:NKptomega}
\partial_t\omega &= \lambda_t\,\omega~,\qquad \dd\lambda_t = 0~,\\
\label{eq:NKptPsi}
\partial_t\Psi &= K_t\Psi~, \qquad K_t =  \frac{3}{2}\, \lambda_t - 3\,i\, N_t \, {\rm Re} W_0~,
\end{align}
and we need to solve
\begin{align}
0&= 3 \,  N_t\, ({\rm Re}W_0 )^2
+ (\lambda_t + N_t\, {\rm Im}W_0)\, {\rm Im}W_0~,\label{eq:quad}\\
\partial_t\lambda_t &= - \textstyle{\frac{3}{2}}\,\lambda_t^2  - \lambda_t\, N_t\, {\rm Im}W_0 
- 12 (N_t\, {\rm Re}W_0)^2~,\label{eq:varlambNK}\\
\partial_t W_0 &= \left(- \textstyle{\frac{1}{2}}\, \lambda_t - 3\, i\, N_t \, {\rm Re} W_0\right)\, W_0~.\label{eq:varW0NKtau0}
\end{align}

We now look for the general solutions of \eqref{eq:quad}-\eqref{eq:varW0NKtau0}. 
Solving equation \eqref{eq:quad} for $\lambda_t$ we find
\begin{equation} \lambda_t = - \frac{N_t}{{\rm Im}W_0} \, \left( 3\,  ({\rm Re}W_0 )^2 + ({\rm Im}W_0)^2\right)~,
\label{eq:thelamb}
\end{equation}
and eliminating $\lambda_t$ from equation \eqref{eq:varW0NKtau0} we have
\begin{equation} 
\partial_t W_0 = \textstyle{\frac{N_t}{2\,{\rm Im}W_0}}\,\left(
    3\,  ({\rm Re}W_0 )^2 + ({\rm Im}W_0)^2
- 6\, i\, {\rm Re} W_0\, {\rm Im}W_0\right)\, W_0~.\label{eq:flowW0NK}
\end{equation}
After a somewhat tedious computation, one can prove that equation \eqref{eq:varlambNK}
is superfluous as it gives an identity when substituting $\lambda_t$ in \eqref{eq:thelamb} and using equation \eqref{eq:flowW0NK}.

We can integrate  equation \eqref{eq:flowW0NK} by writing 
\[ W_0 = r \, e^{i\theta}~.\]
Equation \eqref{eq:flowW0NK} gives two coupled first order differential equations for $(r,\theta)$
\begin{align}
\partial_t\, r^{-1}&= - \textstyle{\frac{N_t}{2\, \sin\theta}}\, (1 + 2\, \cos^2\theta)~,\label{eq:diffr}\\[3pt]
\partial_t\, \theta &= -3\, N_t\, r \cos\theta~.\label{eq:difftheta}
\end{align}

Before continuing with the analysis of flow, we would like to ask whether flows for which $\theta$ is independent of the flow parameter $t$ are allowed.  From equation \eqref{eq:difftheta} we see that this is possible only  when ${\rm Re} W_0 = 0$.
This is rather interesting:  apart from the case where one sets ${\rm Re} W_0 = 0$, {\it the only way a nearly K\"ahler manifold can have a consistent $G_2$ flow is by letting $\theta$ change with the flow}.  Note that this change in the phase of $W_0$ corresponds to a change in the phase of $\Psi$, however these changes leave the almost complex structure invariant.

Consider again the flow equations \eqref{eq:NKptomega} and \eqref{eq:NKptPsi}.  It is not very difficult to prove that one can integrate these equations to find $\omega(t)$ and $\Psi(t)$ in terms of $W_0$, or equivalently, in terms of $r$ and $\theta$.  To see this, one shows first that 
\begin{align}
\lambda_t &= - \partial_t\log r^2~,\label{eq:solnlamb}\\
K_t &= \partial_t\log(r^{-3}\, e^{i\theta})~.\label{eq:solnK}
\end{align}
The first equation follows by a computation of the variation of $r^2 = |W_0|^2$ using \eqref{eq:flowW0NK} and then comparing the result with \eqref{eq:thelamb}.   To find the second relation one only needs the first equation, which gives the real part of $K_t$, and for the imaginary part one uses equation \eqref{eq:difftheta}.  Next, the form of equations \eqref{eq:NKptomega} and \eqref{eq:NKptPsi} means that the general solution has the form
\begin{equation*}
\omega(t) = f(t)\, \omega_0~, \qquad \Psi(t) = \tilde f(t)\, \Psi_0~.
\end{equation*}
where $\omega_0 = \omega(0)$ and $\Psi_0= \Psi(0)$, and $f(0) = \tilde f(0) = 1$.  Putting this together with the flow equations and equations \eqref{eq:solnlamb} and \eqref{eq:solnK} we find
\begin{align}
\omega &= A\, r^{-2}\, \omega_0~,\qquad A = r(0)^2~.\label{eq:omNK}\\
\Psi &= B\, r^{-3}\, e^{i\, \theta}\, \Psi_0~, \qquad B = r(0)^3\, e^{- i\, \theta(0)}\label{eq:PsiNK}~.
\end{align}

Consider the metric on $X$. On an $SU(3)$ manifold, the metric is determined by the complex structure and the hermitian form by
\[ g_{mn} = \omega_{mp}\, J_n{}^p~.\]
The complex structure is an invariant of the flow as $\Psi$ and $\Psi_0$ differ only by a scale factor (see equation \eqref{eq:PsiNK}).
Hence
\begin{equation*}
g_{mn}(t) = A\, r^{-2}\, g_{mn}(0)~,
\end{equation*}
and the metric on the seven dimensional manifold $(Y,\varphi)$ is (see equation \eqref{eq:metric})
\begin{equation*}
\dd^2 s_\varphi = N_t^2\, \dd^2 t + \dd^2 s_X= r^{-2}\, ((r\,N_t)^2\, \dd^2 t + A\, \dd^2 s_{X_0})~, 
\end{equation*}
where $\dd^2 s_{X_0}$ is the metric on $X_0$, that is the metric on $X$ at $t=0$.  In what follows it will be useful to define a new coordinate $T$ such that
\begin{equation}
\dd T = \pm\,N_t\, r \, \dd t~.\label{eq:newT}
\end{equation}
The seven dimensional metric now takes the form
\begin{equation*}
\dd^2 s_\varphi = r^{-2}\, (\dd^2\, T + A\, \dd^2 s_{X_0})~, 
\end{equation*}

We still need to solve equations \eqref{eq:diffr} and \eqref{eq:difftheta} to obtain $W_0$ as a function of $t$ (or $T$).
We begin with the latter.   Changing variables using equation \eqref{eq:newT}, we have 
\begin{equation}
\partial_T\theta = \mp\, 3 \cos\theta~.\label{eq:diffthetaT}
\end{equation}
Integrating we find
\[ \mp\, 3 \,T + c = \log\left(\frac{1 + \sin\theta}{\cos\theta}\right)~.\]
where $c$ is a constant which can be set to zero without loss of generality.  Inverting this relation to find $\theta$ as a function of $T$, we find, after some algebra, the equation 
\begin{equation}
\cos\theta \left(1 - \cos\theta\, \cosh(3\, T)\right) = 0~.\label{eq:twosolnstheta}
\end{equation}
There are two solutions.  The first one, when $\cos\theta = 0$, we have ${\rm Re} W_0 = 0$: this is the case mentioned above in which we  have a flow for a nearly K\"ahler manifold with constant $\theta$.  For clarity, we will study the two cases separately: the case where $\theta$ does not vary with $t$ and the case in which it does.

\subsection{Flow with constant $\theta$}

In this case we have
\[ {\rm Re} W_0 = 0~,\]
and the equation for ${\rm Im} W_0$ is 
\begin{equation}
\partial_t{\rm Im}W_0 =  \textstyle{\frac{1}{2}}\, N_t\, ({\rm Im}W_0)^2~.
\label{eq:flowconsttheta}
\end{equation}
Note that the dilaton remains constant along the flow.

Integrating the equation for the flow of ${\rm Im}W_0$ with respect of $t$ we find a one parameter family of solutions with:
\[ {\rm Im} W_0 = - \frac{1}{\textstyle{\frac{1}{2}}\, N_t\, t + a}~,\qquad
\lambda_t = \frac{N_t}{\textstyle{\frac{1}{2}}\, N_t\, t + a}~,\]
where $a$ is a constant.
The equations for  $\omega$ and $\Psi$ are in this case
\begin{align}
\omega(t) &= \frac{1}{a^2}\, \left(\frac{1}{2}\, N_t\, t + a\right)^2\, \omega_0~,\\[5pt]
\Psi(t) &= \frac{1}{a^3}\, \left(\frac{1}{2}\, N_t\, t + a\right)^3\, \Psi_0~.
\end{align}

Note that there is a singularity in the flow at values of $t = t_s$ for which
\[ t_s = - \frac{2a}{N_t}~.\]
Both forms $\Psi$ and $\omega$ vanish at $t = t_s$. 
The manifolds $X_{t_s}$ have a curvature singularity.  In fact, for nearly K\"ahler manifolds the scalar curvature is \cite{Bedulli2007}
\[ {\cal R} = \frac{15}{2}\, |W_0|^2~.\]

This solution seems to flow to a non-compact Calabi--Yau manifold at $t=\infty$.  In this example we already have ${\rm Re} W_0 = 0$ to begin with, and in the limit $t= \infty$ we also have ${\rm Im} W_0 = 0$.  Moreover,  in this limit  $\lambda_t = 0$, and the scalar curvature also vanishes ${\cal R }= 0$.  Yet, both $\omega$ and $\Psi$ increase monotonically to infinity, and hence the volume of $X$ is infinite in this limit.

Examples of this case have been studied in the literature before, see e.g. \cite{Klaput:2011mz}. This flow has also been studied more recently including $\alpha'$ corrections and the vector bundle $V$ over $X$ which comes with every heterotic string compactification \cite{Klaput:2013nla, Haupt:2014ufa}. Interestingly, it was shown in  \cite{Klaput:2013nla} that $X$ can have a finite volume as $t\rightarrow\infty$  by choosing appropriately the bundle $V$.  It should be noted however that for the
$\alpha'$-corrected flow in \cite{Klaput:2013nla}, the dilaton becomes $t$-dependent. 
Moreover, for solutions where the internal radius tends to a constant, i.e. $\lambda_t\rightarrow 0$, the dilaton blows up as $t\rightarrow\infty$. The decompactification limit we found at zeroth order in $\alpha'$ is therefore traded for a finite volume compactification and a dilaton which blows up.

\subsection{Flow with varying $\theta$}

The second solution to equation \eqref{eq:twosolnstheta} is
\begin{equation*}
\cos\theta = {\rm sech}(3 \, T)~.
\end{equation*}
Consider now the equation for $r$ \eqref{eq:diffr}.  Changing variables using \eqref{eq:newT} we find
\begin{equation*}
- \, \partial_T\, \log r^2 = \frac{4}{\sinh(6T)} + \coth(3T)~,
\end{equation*}
where we have used the relation
\begin{equation}
 \sin\theta = \mp\, \tanh(3T)~.\label{eq:sintheta}
\end{equation}
The sign in this relation is chosen by requiring consistency with equation \eqref{eq:diffthetaT}. 
Integrating we now obtain
\begin{equation}
r^6 = a^6\ \frac{\cosh^2(3T)}{\sinh^3(3T)}~.\label{eq:r6}
\end{equation}

In the expression for $r$, we note that $ 3\, T$ must be positive in order for the right hand side to be positive.  As a function of $T$, $r$ is a monotonically decreasing function and $r \rightarrow 0$ as $T\rightarrow \infty$.  The requirement that $T$ is a positive function means that for $t$ positive, we need to choose the positive sign in equation \eqref{eq:newT} and the negative sign in \eqref{eq:sintheta}. Of course, for $t$ negative, we chose the opposite signs in these equations. 

Finally we need to integrate equation \eqref{eq:newT} to find $T$ as a function of $t$.  Integrating the equation
\[3\, \partial_t T = \pm\, 3\, N_t\, r = \pm\, 3 N_t\, a\, \frac{(\cosh 3T)^{1/3}}{(\sinh 3T)^{1/2}}~,\]
we find 
 \begin{align}
\pm\, \textstyle{\frac{1}{2}}\, N_t\, a\, t &=  b +
  v^{-1/12}\, (1 - v)^{3/4}
+ \frac{8}{11}\,  v^{11/12}\, {}_2\, F_1\left(\frac{1}{4}, \frac{11}{12}, \frac{23}{12}; v\right) \nn\\[5pt]
 &=b + v^{-1/12}\, (1 - v)^{3/4}
 + \frac{2}{3}\, B\left(v; \frac{11}{12}, \frac{3}{4}\right)~,\label{eq:tofv}
  \end{align}
 where $v = \cosh(3T)^{-2}$, $b$ is a constant of integration, and $B(z;p,q)$ is the incomplete Beta function. We choose the constant of integration so that when $v = 1$, that is $3T = 0$, we set $t= 0$.  Hence
\[ b  = - \frac{2}{3} B\left(\frac{11}{12}, \frac{3}{4}\right)~.\]
We can always do this as this choice represents a constant shift in the values of $t$.
  In Figure \ref{fig:3Toft}
 we present a plot of $3T$ as a function of $\tilde t = \textstyle{\frac{1}{2}}\, N_t\, a\, t$.    
\begin{figure}
\centering
\includegraphics[width=0.8\textwidth]{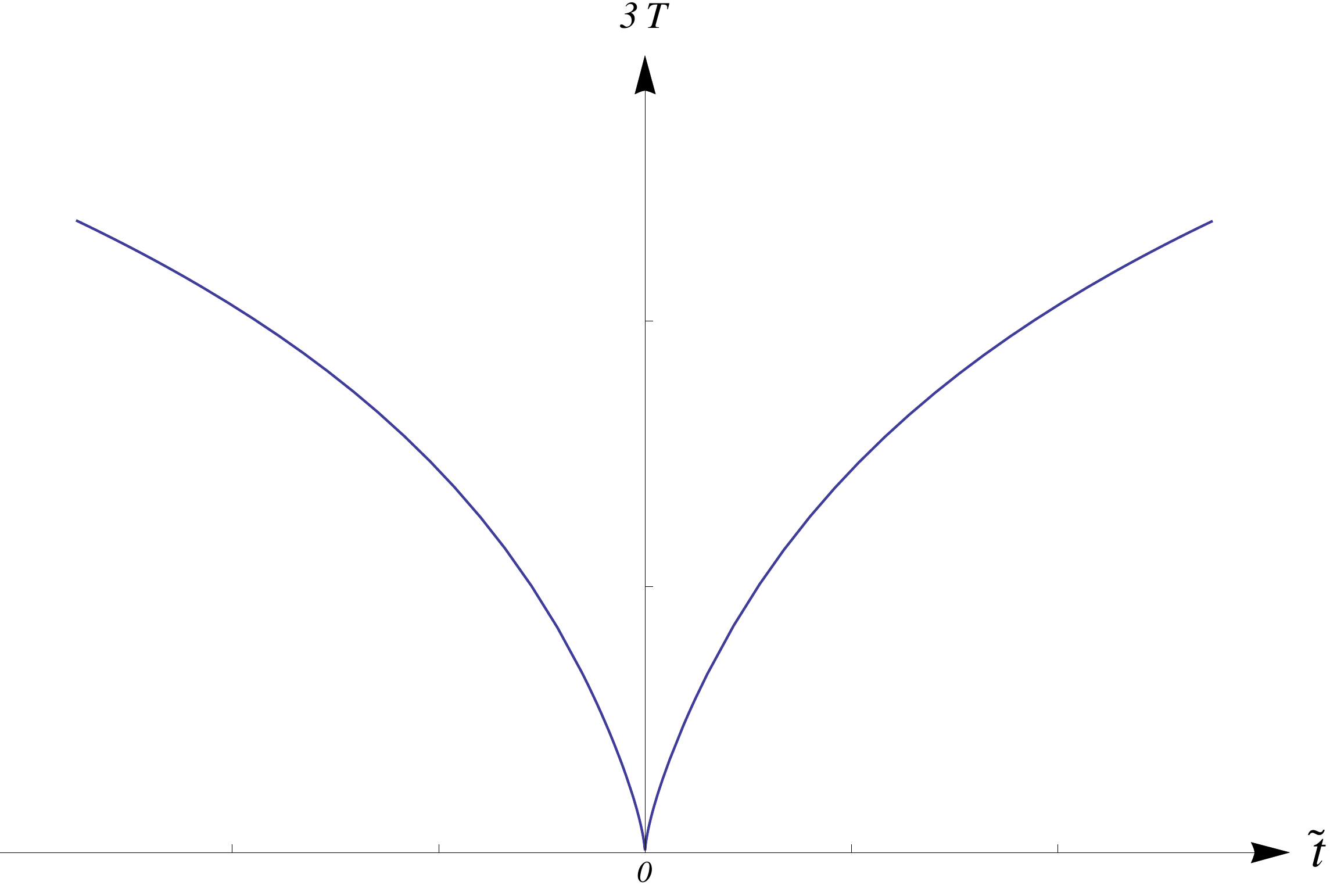}
\caption{Plot of $3T$ as a function of $\tilde t = \textstyle{\frac{1}{2}}\, N_t\, a\, t$.}
\label{fig:3Toft}
\end{figure}

The solution for the torsion class $W_0$ is
\begin{align}
{\rm Re} W_0 &= r \cos\theta = a\, (\sinh(3T))^{-1/2}\, (\cosh(3T))^{-2/3}~,\\
{\rm Im}W_0 &= r \sin\theta = - a\, (\sinh(3T))^{1/2}\, (\cosh(3T))^{-2/3}~.
\end{align}
In Figure \ref{fig:torsion} we show the behaviour of $W_0$ with respect to $\tilde t$.  The torsion class $W_0\rightarrow 0$ as
$\tilde t\rightarrow\infty$, however ${\rm Re}W_0$ falls much faster than ${\rm Im}W_0$.  
In fact, as $t\to\infty$ (or $v\to 0$) we find that  
\[ \cosh(3T)\to \tilde t^{\, 6}~,\]
and hence
\begin{equation}{\rm Re} W_0\to \tilde t^{\, -5}~,\qquad{\rm and} \qquad
{\rm Im} W_0\to - \tilde t^{\, -1}~.\label{eq:asymW0}
\end{equation} 
Finally, note that the point $\tilde t = 0$ corresponds to a manifold with a curvature singularity as $|W_0|^2 = r^2\to a^2\, (3T)^{-1}$ as $\tilde t\to 0$.  

\begin{figure}
\centering
\includegraphics[width=0.8\textwidth]{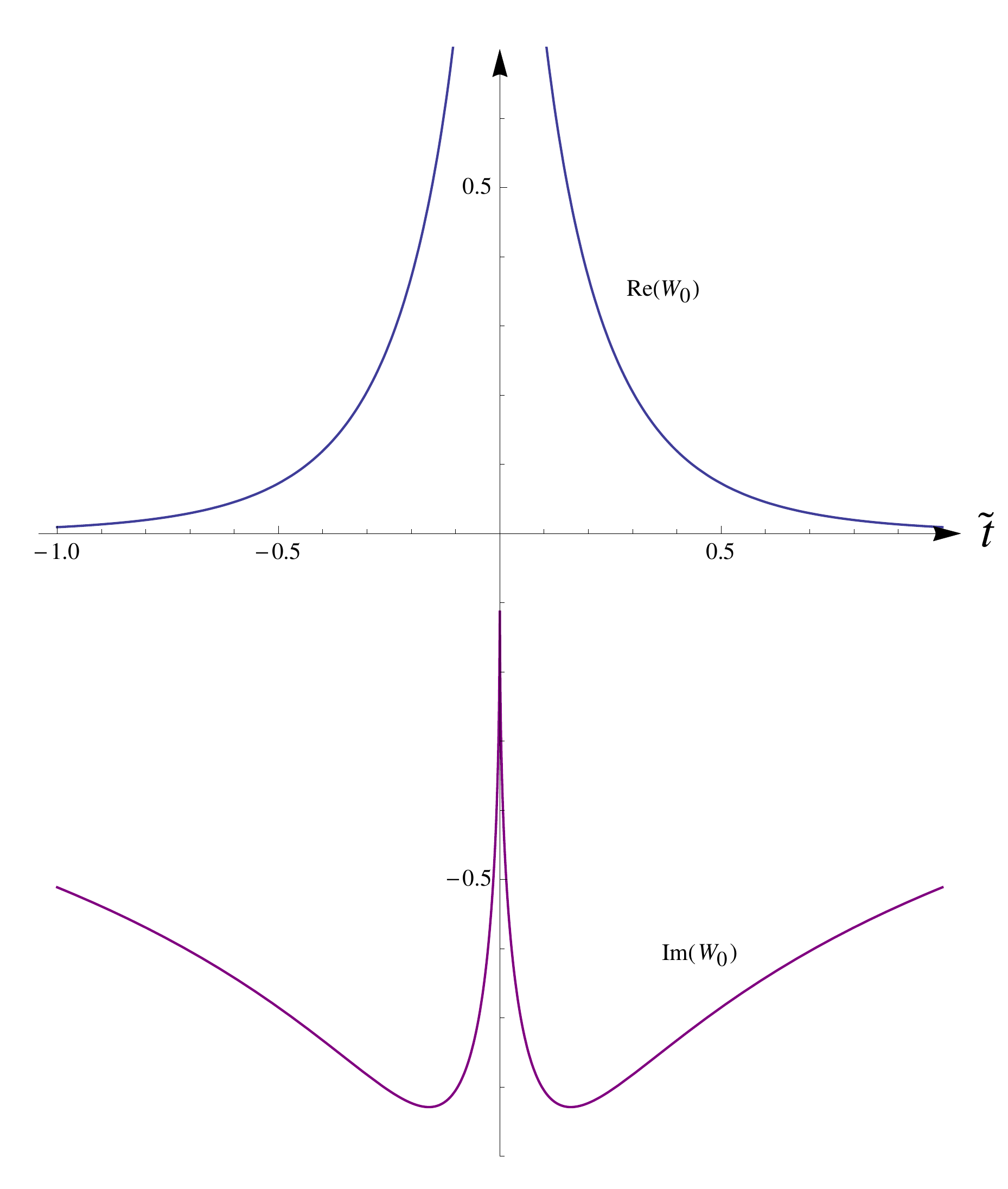}
\caption{Plot of ${\rm Re}W_0$ and ${\rm Im}W_0$ as functions of $\tilde t = \textstyle{\frac{1}{2}}\, N_t\, a\, t$.}
\label{fig:torsion}
\end{figure}

It is interesting also to compute the variation of the dilaton.  Recall that
\[ \partial_t \phi = \lambda_t + N_t\, {\rm Im} W_0~.\]
Using equations \eqref{eq:solnlamb} and \eqref{eq:sintheta}, we find
\[ e^{\phi- \phi_0}\, r^2 =  a^2\, (\cosh(3T))^{-1/3}~,\]
and by equation \eqref{eq:r6}
\[  e^{\phi- \phi_0} = \tanh(3T)~.\]
Thus, the dilaton flows to a constant as $t\to\infty$.

Just as the case in the previous section, this solution flows to a non-compact Calabi--Yau manifold at $\tilde t=\infty$.  In this limit we  have $W_0 = 0$, $\lambda_t = 0$, and the scalar curvature also vanishes ${\cal R }= 0$.  Moreover, both $\omega$ and $\Psi$ increase monotonically to infinity, and hence the volume of $X$ at infinity is infinite.

This flow does not intersect the flow discussed in the previous section.  The two cases only coincide at $t=\infty$ where they both flow into a non-compact Calabi--Yau manifold.  However, the flow with varying $\theta$ approaches asymptotically the flow with constant $\theta$ when $t\to\infty$, as it is easily checked by comparing equation \eqref{eq:asymW0} with \eqref{eq:flowconsttheta}.

\subsection{Appendix: ${\rm Im} W_0 \ne 0$}
\label{app:NKimW0}
In this appendix we prove the claim at the beginning of this section that
${\rm Im} W_0 \ne 0$, or
otherwise we fall back on a Calabi--Yau flow. Assume that 
\[{\rm Im} W_0 = 0~,\]
and consider the first Bianchi identity
\[\dd S^X = 0~,\]
where
\[ S^X = 
- \frac{1}{2}\, \left( \frac{49}{30}\, \tau_0 - 3 {\rm Re} W_0\right)\, {\rm Re}\Psi
+ N_t^{-1}\tau_{1\, t}\, {\rm Im}\Psi + J\gamma~.\]
Taking the wedge product with $\omega$, and noting that $S^X$ is primitive, we find
\begin{align*}
 0 &= \dd S^X\wedge\omega = \dd (S^X\wedge\omega) + S^X\wedge\dd\omega
= \frac{3}{2}\, {\rm Re}W_0\, S^X\wedge{\rm Im}\Psi\\[3pt]
&= - 3\, {\rm Re}W_0\, 
\left( \frac{49}{30}\, \tau_0 - 3 {\rm Re} W_0\right)\, \dd{\rm vol}_X~.
\end{align*} 
For ${\rm Re} W_0 \ne 0$ we must have
\[ {\rm Re} W_0 = \frac{49}{90}\, \tau_0~.\]
Note that this means that 
\begin{equation}
\partial_t W_0 = 0~,\label{eq:noflowW0}
\end{equation}
 as $\tau_0$ is a constant.

Consider now the flow of $\Psi$.  From equation \eqref{eq:eqchit} we have that
\begin{align*}
 \dd\chi_t &= -\dd K_t\wedge\Psi 
+ 2 (\partial_t{\rm Re} W_0 - K_t\, {\rm Re}W_0)\, \rho 
+ 2{\rm Re} W_0\, \partial_t\rho\\
&= -\dd K_t\wedge\Psi + 2\, {\rm Re}W_0\,
((2\, \lambda_t- K_t)\,\rho + h_t\wedge\omega)~.
\end{align*}
where we have used \eqref{eq:noflowW0} and the flow equations for $\omega$. Taking the wedge product with $\omega$ and recalling the $h_t$ and $\gamma$ are primitive, we find 
\[ -  i\, \dd(*\dd N_t) = 3\, {\rm Re}W_0\, (2\lambda_t - K_t)\, \dd{\rm vol}_X~.\]
The real part of this equation gives
\[ \lambda_t = 0~,\]
and, from the imaginary part we have
\[  \dd(*\dd N_t) = 3\, {\rm Re}W_0\, {\rm Im}K_t\, \dd{\rm vol}_X
= \frac{9}{14}\, N_t ({\rm Re} W_0)^2  \dd{\rm vol}_X~,\]
and hence
\[ \dd N_t = 0~,\qquad {\rm Re} W_0 = 0~.\]

\vspace{0.5cm}

\section{Half-flat $SU(3)$ structures}
\label{sec:HF}

In this section we consider the flow of $SU(3)$ structures which are half-flat, that is when the torsion classes
$W_1^\omega$,  $W_1^{\Psi}$ vanish, but $W_0, W_2$ and $W_3$ are non-zero. As we will see below, Hitchin flow \cite{Hitchin:2000jd}, for which the torsion classes of the $G_2$ manifold are all zero, is recovered as a subcase of this flow, with vanishing flux, constant dilaton and constant embedding functions $\alpha, N_t$. We will argue that simplified versions of half-flat flows can allow Calabi--Yau loci, even at finite values of $t$. To conform with the previous sections, we choose $\alpha=1$.

A flow that preserves a half-flat $SU(3)$ structure should for all $t$ obey
 \begin{align}
\dd\omega &= - \textstyle{\frac{3}{2}}\, {\rm Im} (W_0\overline\Psi) + W_3~,\label{eq:IntomHF}\\
\dd\Psi &= 2\, W_0 \,\rho + W_2 \wedge\omega~,\label{eq:InPsiHF}
\end{align}
where 
\[ \rho = *\omega = \frac{1}{2}\, \omega\wedge\omega~,\]
is a closed form. Taking the exterior derivative of equation \eqref{eq:IntomHF} and \eqref{eq:InPsiHF} one finds differential equations for the torsion classes
\be
\begin{split}
\dd W_3 &= \frac{3}{2}\, {\rm Im}\left(W_0\, \overline W_2\wedge\omega
+ \dd W_0 \wedge \overline\Psi\right)~,\\
 \dd^\dagger W_2 &= 2\,J(\dd W_0)~.
 \end{split}
 \ee

The flow equations for $\omega$ in this case are
\begin{align*}
\partial_t\omega &= \lambda_t\, \omega + h_t~,\\
\lambda_t &=2\, \tau_{1\, t} - N_t\, {\rm Im}W_0~,\quad  2\, \tau_{1\, t}= \partial_t\phi ~,\quad \dd\lambda_t = 0~,\\
h_t &= - N_t\, {\rm Im}W_2  -  \dd N_t\lrcorner {\rm Re}\Psi~.
\end{align*}
Since $\lambda_t$ is $\dd$-constant and $W^\omega_1 = \dd\phi = 0$ for all $t$, we find a constraint
\be \dd\big(N_t\, {\rm Im}W_0\big) = 0~.
\label{eq:hfconstr}\ee
This is similar to the nearly K\"ahler case, with the difference that $W_0$ need not be constant. 

The flow equations for $\Psi$ are
\begin{align}
\partial_t\Psi &= K_t\, \Psi + \chi_t~,\nn\\
 {\rm Re} K_t &= \textstyle{\frac{3}{2}}\, \lambda_t~,\qquad 
 {\rm Im} K_t = N_t \left( \textstyle{\frac{7}{4}}\,  \tau_0 - 3\,  {\rm Re} W_0 \right)~,\nn\\[3pt]
{\rm Re}\chi_t &= \dd N_t\wedge\omega + N_t(W_3 - \gamma)~, \quad \omega\lrcorner\gamma = 0~,
\label{eq:ReChi}\\
{\rm Im}\chi_t &= - J(\dd N_t)\wedge\omega - N_t\, J(W_3 - \gamma)~.
\label{eq:ImChi}
\end{align} 

To study the flow of $SU(3)$ structures, it is useful to record the $t$-variations of the torsion classes. Compatibility between the flow equation for $\omega$ and the variation of equation \eqref{eq:IntomHF}, or equivalently,
using equation \eqref{eq:flowW3} with $\alpha=1$ and vanishing Lie forms,  gives
\be
\begin{split}
\partial_t W_3 - \lambda_t\, W_3 &= \dd h_t + {\textstyle\frac{3}{2}}\, {\rm Im}\big((\partial_t W_0 + (\bar K_t - \lambda_t) W_0)\overline\Psi\big)\\[3pt]
&+ {\textstyle\frac{3}{2}}\, N_t \big( ({\rm Re} W_0)\, J(W_3 - \gamma) + ({\rm Im} W_0)\, (W_3 - \gamma)\big)\\[3pt]
&+ {\textstyle\frac{3}{2}}\, \big( ({\rm Re} W_0)\, J(\dd N_t) + ({\rm Im} W_0)\, \dd N_t\big)\wedge\omega \; .
 \label{eq:flowW3HF}
\end{split}
\ee
Note that this equation satisfies
\[ \partial_t(W_3\wedge\omega) = \partial_t W_3\wedge\omega + W_3\wedge h_t = 0~,\]
as required by the primitivity of $W_3$.  It should also satisfy
\[ \partial_t(W_3\wedge\Psi) = \partial_t W_3\wedge \Psi + W_3\wedge\chi_t = 0~,\]
as required by the fact that $W_3$ is a three form of type $(2,1)+(1,2)$.  This constraint gives a flow equation for $W_0$
\begin{equation}
6\, (\partial_t W_0 + (\bar K_t - \lambda_t)\, W_0) = -2\, i\, \Delta_d\, N_t + N_t\, {\rm Im} W_2\lrcorner W_2
 + N_t\, W_3\lrcorner (-J\gamma + i\, (W_3 - \gamma))~.
 \label{eq:flowW0HF}
\end{equation}
Compatibility between the flow equation for $\Psi$ and the variation of equation \eqref{eq:InPsiHF} gives a flow equation for $W_2$ 
\begin{equation}
\begin{split}
(\partial_t W_2 + (-K_t + \lambda_t)\, W_2 )\wedge \omega & =
\dd \chi_t + \dd K_t\wedge\Psi - 2\, (\partial_t W_0 + (-K_t +2\, \lambda_t) W_0)\, \rho\\
&- ( 2\, W_0\, \omega + W_2 )\wedge h_t~.
\end{split}
\label{eq:chiandW2}
\end{equation}
By assumption, the variation of the Lie forms is zero along a half-flat flow.

The first Bianchi identity $\dd S^X = 0$ is a constraint on the exterior derivative of $J\gamma$.  With
\[ S^X = 
- \textstyle{\frac{1}{2}}\, \left(\textstyle{\frac{49}{30}} \, \tau_0
- 3 {\rm Re}W_0\right) {\rm Re}\Psi
+ N_t^{-1}\, \tau_{1\, t}\, {\rm Im}\Psi + J\gamma~.\]
we have
\be
\begin{split}
- \dd J(\gamma) = \dd (N_t^{-1} \tau_{1t}) \w {\rm Im} \Psi 
&+
\left(
N_t^{-1}\, \tau_{1\, t}\, {\rm Im} W_0 
- \textstyle{\frac{1}{2}}\, \left(\textstyle{\frac{49}{30}} \, \tau_0
- 3 {\rm Re}W_0\right) {\rm Re} W_0
\right) \omega \w \omega\\
+ \textstyle{\frac{3}{2}} \dd {\rm Re}W_0 \w {\rm Re}\Psi
&+
\left(
N_t^{-1}\, \tau_{1\, t}\, {\rm Im} W_2 
- \textstyle{\frac{1}{2}}\, \left(\textstyle{\frac{49}{30}} \, \tau_0
- 3 {\rm Re}W_0\right) {\rm Re} W_2
\right) \w \omega \; .
\end{split}
\label{eq:hfbi1}
\ee
Taking the wedge product of this equation with $\omega$, and recalling that $\gamma$ is primitive, we find 
\begin{equation} \frac{1}{3}\, W_3\lrcorner\gamma = 
2\, N_t^{-1}\, \tau_{1\, t}\, {\rm Im} W_0 
- \left(\frac{49}{30}\, \tau_0 - 3\, {\rm Re} W_0\right)\, {\rm Re}W_0~,
\label{eq:W3gamma}
\end{equation}
which will be of use below.

The second Bianchi identity, $\dd S_t = \partial_t S^X$, where
\[ S_t = N_t\left( {\rm Re} W_2 + 2 \left(\textstyle{\frac{7}{15}}\, \tau_0 - {\rm Re} W_0\right)\, \omega\right) ~,\]
determines, among other things, the flow of $J\gamma$.  For our discussion below, we will be particularly interested in the necessary constraints obtained by wedging the second Bianchi identity with ${\rm Re} \Psi$ and $\omega$. From the former, we obtain
\be
\label{eq:hfbi2Psi}
\begin{split}
N_t\, \partial_t(N_t^{-1}\, &\tau_{1\, t}) 
+  3 \, \tau_{1\, t}\, (\tau_{1\, t} - N_t\, {\rm Im}W_0)\\
&=
 - \frac{1}{4}\, N_t^2\, \left(||{\rm Re} W_2||^2 + ||\gamma||^2
 + 21\, \Big( {\rm Re}W_0 - \frac{1}{2}\tau_0\Big)^2
 + \frac{7}{15}\, \tau_0^2 \right)
~,
 \end{split}
\ee
where we have used equation \eqref{eq:W3gamma}. From the latter we get
\be
\label{eq:hfbi2omega}
\begin{split}
N_t\, ({\rm Im} W_2\lrcorner \gamma + 6\, J( \dd{\rm Re} W_0) &- {\rm Re}W_2\lrcorner JW_3)
=
\dd N_t\lrcorner{\rm Re}W_2 - 2 \, N_t^{-1}\, \tau_{1\, t}\, \dd N_t\\[3pt]
 & + \textstyle{\frac{7}{2}}\,\left( \tau_0 - 2\, {\rm Re} W_0\right) J\dd N_t
 - \left((\dd N_t)\lrcorner{\rm Re}\Psi\right)\lrcorner\gamma~.
\end{split}
\ee
Note that \eqref{eq:hfbi2Psi} can be viewed as a differential equation for the dilaton $\phi(t)$, whereas \eqref{eq:hfbi2omega} determines the one-form $\dd N_t$.  Moreover, the right hand side of equation \eqref{eq:hfbi2Psi} is negative definite, hence it gives an inequality
\begin{equation}
N_t\, \partial_t(N_t^{-1}\, \tau_{1\, t}) 
+  3 \, \tau_{1\, t}\, (\tau_{1\, t} - N_t\, {\rm Im}W_0)\le 0~,\label{eq:ineq}
\end{equation}
which is saturated only when $X$ is a Calabi--Yau manifold.

Given the complexity of the flow equations \eqref{eq:flowW3HF}-\eqref{eq:chiandW2} and the constraints from the Bianchi identities \eqref{eq:hfbi1}-\eqref{eq:hfbi2omega}, we will not attempt to solve this system of equations in full generality. Instead, we proceed to study simplified cases, where we make assumptions on the flux and the embedding of the $SU(3)$ structure. 

We begin by an inspection of the first Bianchi identity \eqref{eq:hfbi1}. This is a strong constraint on $\gamma$ that completely determines its non-coclosed components. If we moreover assume that $\gamma$ is coclosed, \eqref{eq:hfbi1} immediately tells us that 
\be \label{eq:hfcoclosed}
\dd J(\gamma) = 0 \implies \dd N_t = 0~, \quad \dd W_0 = 0~, \quad \dd^\dagger W_2 =0~, \quad
\dd W_3 = \frac{3}{2} {\rm Im} (W_0 \overline{W}_2) \w \omega~,
\ee
and furthermore gives relations that determine the flow of the complex phase of $W_0$ and $W_2$. Using this in equation \eqref{eq:ineq} and \eqref{eq:hfbi2omega}, we derive the necessary constraints
\be
\nn
\partial_t \tau_{1t}  \le \frac{3}{2} \tau_{1t} N_t {\rm Im} W_0 \; ,\qquad
0 = {\rm Re} \left(W_2\lrcorner [W_3 -i J(\gamma)]\right) \; ,
\ee
where the first inequality can only be saturated when $X$ is a Calabi--Yau manifold as mentioned above. 

It is interesting that the requirement that $\gamma$ is coclosed has so far-reaching consequences for the $SU(3)$ structure and its  $G_2$ embedding. Recall that in our first order analysis of $G_2$ flows away from a Calabi--Yau manifold $X(t=0)$, we showed that a half-flat
$SU(3)$ structure was obtained when $\gamma(t=0)=0$, but  $\dd N_0 \neq 0$ (cf.~equation \eqref{eq:CYtoHF}). Although the flow of a half-flat $SU(3)$ structure with vanishing $\gamma$ would be a  natural guess for the completion of this first order flow, the constraint we just derived shows that this is not the case. A non-zero $\gamma$ seems necessary in order for a half-flat flow to contain Calabi--Yau loci at finite values of $t$.

Half-flat flows with coclosed $\gamma$ have been studied before \cite{Hitchin:2000jd,chiossi,Gray:2012md,Lukas:2010mf,Klaput:2011mz,Klaput:2012vv,Klaput:2013nla}. A particular case is Hitchin flow, where not only $\gamma$, but also $S^X$ and $S_t$ vanish, and the dilaton and embedding parameter $\alpha$ are both constant. In this case, we can embed the half-flat $SU(3)$ structure in  a $G_2$ holonomy manifold, and the flow equations can be summarised as
\be \label{eq:hitchflow}
\begin{split}
\partial_t \rho = - \dd \mbox{Im} \Psi  \\
\partial_t \mbox{Re} \Psi = \dd \omega \; .
\end{split}
\ee
in the gauge where $\mbox{Re} \Psi$ is taken to be constant.\footnote{Recall that when the two Lie forms vanish, we are free to choose the phase of $\Psi$ without changing the properties of the $SU(3)$ structure. Thus, we can always choose a gauge where the real (imaginary) part of $\Psi$ is closed, and hence $W_0, W_2$ are imaginary (real). Once we have fixed $\alpha=1$, these gauge choices correspond to different ways of embedding the $SU(3)$ structure in a $G_2$ structure.} From our discussion above, it is clear that this flow cannot contain  Calabi--Yau loci. This can also be seen as follows: $\Psi$ and $\omega$ are closed at a Calabi--Yau locus, which means that the right-hand sides of the equations \eqref{eq:hitchflow} vanish at that point of the flow. It is easy to see that the first order derivatives of the torsion classes \eqref{eq:flowW3HF}-\eqref{eq:chiandW2}  then also vanish. Moreover, a straightforward inductive analysis shows that all higher order $t$-derivatives of the torsion classes vanish, and hence $\omega$ and $\Psi$ have to remain closed along the flow. Consequently, Calabi--Yau manifolds are fix points of the Hitchin flow.

In order to connect first order analysis of the flow from a Calabi--Yau threefold in section \ref{sec:o1sum}, we therefore need flows where either the embedding parameter $N_t$ is non-constant and the flux is non-vanishing. In the next subsection, we analyse a simple type of such a flow.

\subsection{Symplectic half-flat}

There is a way to simplify the analysis of half-flat flows, that still allows to connect with the perturbative analysis of the Calabi--Yau flow. In section \ref{sec:o1sum} we found that, at linear order, a Calabi--Yau manifold with constant $N_0$ and non-harmonic $\gamma$ resulted in a symplectic half-flat $SU(3)$ structure with real $W_2$, see equation \eqref{eq:CYtosympHF}. Therefore, let us now consider the flow of such $SU(3)$ structures, that is where
 \begin{align}
\dd\omega &= 0~.\label{eq:omegaSHF}\\
\dd\Psi &= W_2 \wedge\omega~,\label{eq:PsiSHF}
\end{align}
Since $\omega$ is closed, it provides the six-manifold with a symplectic structure. As above, we will look for points along the flow where $W_2$ vanishes.

We begin by proving that 
\be
\dd N_t = 0 ~, \qquad {\rm Im} W_2 = 0~,
\label{eq:sympcostr}
\ee
along this flow. To see this consider the flow of the torsion class $W_0$, \eqref{eq:flowW0HF}, for symplectic half-flat $SU(3)$ structures. The real and imaginary parts of this equation must vanish separately, so with  $W_0 = 0 = W_3$ for all $t$ we have
\be
0 = -2\,  \Delta_d\, N_t 
+ N_t\, ({\rm Im} W_2\lrcorner{\rm Im} W_2 )
~,
\quad \quad \quad 
0 =  N_t\, {\rm Im} W_2\lrcorner{\rm Re} W_2 ~.
\ee
Taking the Hodge dual of the first equation we find
\[ - 2\, \dd(*\dd N_t) = N_t\, {\rm Im} W_2 \wedge {\rm Im} W_2 \wedge \omega 
=  {\textstyle\frac{1}{6}}\, N_t\, ||{\rm Im} W_2||^2 \,\omega\wedge\omega\wedge\omega~.\]
As the left hand side is $\dd$-exact, and the right hand side is either strictly positive or strictly negative, this identity can only be true if both sides vanish.  Hence, equation \eqref{eq:sympcostr} follows. Recall that any $t$-dependence of a $\dd$-constant $N_t$ can be absorbed by a coordinate change. We will therefore take $N_t$ constant below. 

The flow equations for $\omega$ become (note that $h_t = 0$)
\[
\partial_t\omega = \lambda_t\, \omega~,\quad
\lambda_t =2\, \tau_{1\, t} = \partial_t\phi ~,\quad \dd\lambda_t = 0~,
\] 
with solution
\be
\omega(t,x) = e^{\phi(t)} \omega_{0} (x) \; .
\ee
The closure of $\lambda_t$ is automatic, since we assume that $W_1^{\omega} = \dd \phi$ is zero along the flow. Thus, requiring that $\omega$ is closed for all $t$ leads to no further constraints. This is consistent with the fact that the flow equation for $W_3$, \eqref{eq:flowW3HF}, is trivially satisfied along the flow.

The flow equations for $\Psi$ are (note that $\dd N_t = 0$)
\begin{align}
\partial_t\Psi &= (\textstyle{\frac{3}{2}}\, \lambda_t~+ i \textstyle{\frac{7}{4}}\, N_t\, \tau_0) \Psi +  \chi_t~,\nn\\
{\rm Re}\chi_t &=  - N_t\, \gamma~, \quad 
{\rm Im}\chi_t =  N_t\, J\gamma~, \quad \omega\lrcorner\gamma = 0~.
\label{eq:ImChi2}
\end{align} 
For future reference, we record the compatibility conditions coming from $\partial_t \dd \Psi = \dd \partial_t \Psi$:
\be \label{eq:sympcst3}
\begin{split}
N_t \dd (J\gamma) = \dd{\rm Im} \chi_t &= - \textstyle{\frac{7 }{4}}\, \tau_0 N_t \,{\rm Re} W_2 \w \omega \; ,\\
-N_t \dd \gamma = \dd{\rm Re} \chi_t &= (\partial_t {\rm Re}W_2 - \textstyle{\frac{1}{2}} \lambda_t {\rm Re}W_2)\w \omega \; .
\end{split}
\ee

Finally, the supersymmetry equations fix the flux to
\be
\begin{split}
S^X &= 
- \textstyle{\frac{49}{60}} \, \tau_0 \, {\rm Re}\Psi
+ N_t^{-1}\, \tau_{1\, t}\, {\rm Im}\Psi + J\gamma ~,\\
S_t &= \frac{14}{15}\, N_t\, \tau_0\, \omega  + N_t\, {\rm Re} W_2 \; .
\end{split}
\ee
As above, the first Bianchi identity is a constraint on the exterior derivative of $J\gamma$, which can be read off from \eqref{eq:hfbi1} after setting $W_0$ to zero. The second Bianchi identity determines the flow of $J\gamma$ with $t$. The compatibility of these equations  with \eqref{eq:sympcst3} must be checked. First, we have
\be
0 = \dd S^X 
\Leftrightarrow
\dd(J \gamma) = \textstyle{\frac{49}{60}} \, \tau_0
 {\rm Re} W_2 \w \omega \; ,
 \ee
which is compatible with \eqref{eq:sympcst3} if and only if
\be \label{eq:sympcst4}
\tau_0 = 0 \; .
\ee

The second Bianchi identity gives, upon using \eqref{eq:sympcst4}, 
\be \label{eq:sympcst5}
\begin{split}
\partial_t S^X &= \dd S_t~,
\\ \iff 
N_t \dd {\rm Re}W_2 &=  N_t^{-1} \left( \partial_t \tau_{1t} + 3  \tau_{1t}^2 \right){\rm Im}\Psi +\partial_t (J \gamma) + \tau_{1t} \, J \gamma 
\\ \iff 
N_t^2 \dd {\rm Re}W_2
&=
\partial_t^2\, {\rm Im} \Psi 
- 2 \tau_{1t} \partial_t \, {\rm Im} \Psi 
-2 \partial_t \tau_{1t} \, {\rm Im}\Psi   \; .
\end{split}
\ee
Thus, given $W_2$, the second Bianchi identity further constrains the flow of ${\rm Im} \Psi$. This equation is difficult to solve in general, but we note that all terms in the last equation are $\dd$-closed. Thus, the equation is compatible with the constraints on the flow.

Let us summarise the conditions for a symplectic half-flat flow. We have
\be
N_t \mbox{ constant}, \quad \tau_0 = 0, \quad {\rm Im} W_2 = 0 \; ,
\ee
and can rewrite \eqref{eq:sympcst3} and \eqref{eq:sympcst5} as constraints on $\gamma$:
\be \label{eq:sympcst6}
\begin{split}
\dd J \gamma &= 0~,\\
-N_t \dd \gamma &= (\partial_t {\rm Re}W_2 - \tau_{1t} {\rm Re}W_2)\w \omega \; , \\
\partial_t (J \gamma) + \tau_{1t} \, J \gamma &= N_t \dd {\rm Re}W_2 - N_t^{-1} \left( \partial_t \tau_{1t} + 3  \tau_{1t}^2 \right){\rm Im}\Psi 
 \; .
\end{split}
\ee
Consider the wedge product with ${\rm Re}\Psi$ of the last equation (see equation \eqref{eq:hfbi2Psi})
\[ - N_t^2\, (\gamma\lrcorner\gamma  + {\rm Re} W_2\lrcorner {\rm Re} W_2) =
4 \left( \partial_t \tau_{1\,t} + 3  \tau_{1\,t}^2 \right)~.\]
As the left hand side is  negative definite on a symplectic half-flat manifold, it must be the case that
\begin{equation}
 \partial_t \tau_{1\,t} + 3  \tau_{1\,t}^2  < 0~.\label{eq:negfunct}
 \end{equation}
Consequently, flows with constant dilaton $\phi$ are thus excluded for symplectic half-flat SU(3) structures: the inequality in \eqref{eq:negfunct} can only be saturated if $\gamma$ and $W_2$ both vanish. This takes us back to a Calabi--Yau flow as discussed earlier. 

The flow specified by the equations \eqref{eq:sympcst6} depends on whether $\gamma$ vanishes or not. We therefore have two different classes of solutions.

\subsubsection{$\gamma=0$}

In the case that $\gamma=0$ for all $t$, which is in particular true if ${\rm dim}(H^{(2,1)}(X)\oplus H^{(1,2)}(X))^{\rm prim}$ is zero, we can solve \eqref{eq:sympcst6} explicitly. First, note that $\gamma=0$ implies that $\partial_t \Psi = 3 \tau_{1t} \Psi$ and $\partial_t W_2 = \tau_{1t} W_2$. The flow equations for $\omega, \Psi$ and $W_2$ then integrate to
\be
\label{eq:scaleshf}
\begin{split}
\omega(t,x) &= e^{\phi(t)} \omega_0(x)~, \\
\Psi(t,x) &= e^{\frac{3}{2}\phi(t)} \Psi_0(x)~, \\
W_2(t,x) &= e^{\frac{1}{2}\phi(t)} W_{2,0}(x) 
 \; ,
\end{split}
\ee
where a zero denotes the value of the form at some point $t=0$ along the flow, and $\phi(0)=0$. Clearly, there is no non-singular Calabi--Yau locus in this flow: if $W_2$ goes to zero, so do $\omega$ and $\Psi$, leading to a manifold of vanishing volume. As we discuss further below, this is consistent with the first order analysis of flows from Calabi--Yau manifolds.

The third equation of \eqref{eq:sympcst6} can be decomposed into
\be \label{eq:sympcst7}
\begin{split}
\dd ({\rm Re}W_{2,0}) &=  -\frac{1}{2} N_t^{-2} C_0 \, {\rm Im}\Psi_0~, \\
-C_0 &= e^{\phi} \left( \partial_t^2 \phi + \frac{3}{2}  (\partial_t \phi)^2 \right) 
 \; ,
\end{split}
\ee
where $C_0$ must be positive by \eqref{eq:negfunct}. The first equation is a constraint on the $SU(3)$ structure of $X(t=0)$. The second equation is an non-linear differential equation for the flow of $\phi$:
\begin{equation*}
\phi'' + \frac{3}{2}\, (\phi')^2 + a^2\, e^{-\phi} = 0~,
\end{equation*}
where $a^2 =  2\, C_0$ is a positive constant. To solve this, we proceed as in the analysis of nearly K\"ahler flows (cf.~section \ref{sec:NK}): let
\begin{equation}
 \dd T = \pm N_t\, e^{-\phi/2}\, \dd t~.\label{eq:Toft}
 \end{equation}
In terms of $T$ the differential equation becomes
\[ \frac{\dd^2\phi}{\dd T^2} + \left( \frac{\dd\phi}{\dd T}\right)^2 + a^2 = 0~.\]
Substituting $v = \frac{\dd\phi}{\dd T}$ and $w = v^2$ we find a first order differential equation for $w$. Substituting back, we can solve for $\phi(T)$.
We have
\begin{equation*}
\pm \, (T-T_0) = \int\frac{1}{\sqrt{- a^2 + C_1\, e^{-2\phi}}}\, \dd\phi
= - \frac{1}{a}\, {\rm Arctan}\left(\frac{1}{a}\, \sqrt{- a^2 + C_1\, e^{-2\phi}}\right)~.
\end{equation*}
Note that the right hand side of always negative irrespective of the sign one chooses for $a = \pm\, \sqrt{2\, C_0}$~.  Inverting this relation we obtain
\begin{equation}  e^{2\phi} =\frac{C_1}{ a^2}\, \cos^2\left(a\, (T-T_0)\right)~,
\label{eq:phiofT}
\end{equation}
where $C_1$ must be positive. Moreover, when taking the square root of this relation, one must take the absolute value of the cosine so that $e^\phi$ is always positive. 

At infinitely many points $T_n$, where  $a\, (T-T_0)= (n+1)\pi/2$ and $n$ is an integer, we see that 
\be
e^{2\phi}|_{T_n} = 0 \; .
\ee
Recalling the relations \eqref{eq:scaleshf}, we thus see that the volume of the $SU(3)$ structure manifold $X(t)$ vanishes at regular intervals. However, this is no curvature singularity, since the scalar curvature is proportional to $||W_2||^2$ and hence goes to zero at these points. To follow the flow through these possibly singular points\footnote{It is possible that the $G_2$ manifold $Y$ compactifies and stays regular when the $SU(3)$ structure fibre shrinks to zero size, in analogy with how a one-circle fibered over an interval forms a two-sphere by shrinking to zero size at the boundaries of the interval.}
requires an analysis that goes beyond the supergravity approximation used in this paper. Although $g_s$ corrections should still be small at $T_n$, $\alpha'$ corrections might still be important. We hope to come back to this discussion in the future. For the present analysis, we will circumvent this discussion by focusing on the flow of $\phi$ in the domain  $-\pi/2< a(T-T_0)< \pi/2$.

Finally, to find $\phi$ as a function of $t$, we solve equation \eqref{eq:Toft}.  Substituting \eqref{eq:phiofT} into \eqref{eq:Toft}, we have
\begin{equation*}
\pm\, N_t\, (t- t_0) = \int e^{\phi/2}\, \dd T =
\left(\frac{C_1}{ a^2}\right)^{1/4}\,  \int \left(\cos^2\left(a\, (T-T_0)\right)\right)^{1/4}\, \dd T~.
\end{equation*}
For values of $T$ in the domain $-\pi/2< a(T-T_0)< \pi/2$, the integral gives
\begin{equation*}
\pm\, N_t\, (t- t_0) = 2\,  \left(\frac{C_1}{ a^2}\right)^{1/4}\, E\left(\frac{a(T-T_0)}{2}, 2\right)~,
\end{equation*}
where $E(x,m)$ is the elliptic function of the second kind. The range of $t$ is finite, and bounded by $t_{\pm} = t_0 \pm 2 N_t^{-1} \left(\frac{C_1}{ a^2}\right)^{1/4}\, E\left(\frac{\pi}{4}, 2\right)$. By inverting this relation to get $T(t)$, we can plot the evolution of $\phi$ in the interval $t_- \le t \le t_+$. The result is given in figure \ref{fig:shfu}, where we have chosen $t_0 = 0$ and normalised $\phi$ so that $\phi(0) = 0$.

\begin{figure}
\centering
\includegraphics[width=0.6\textwidth]{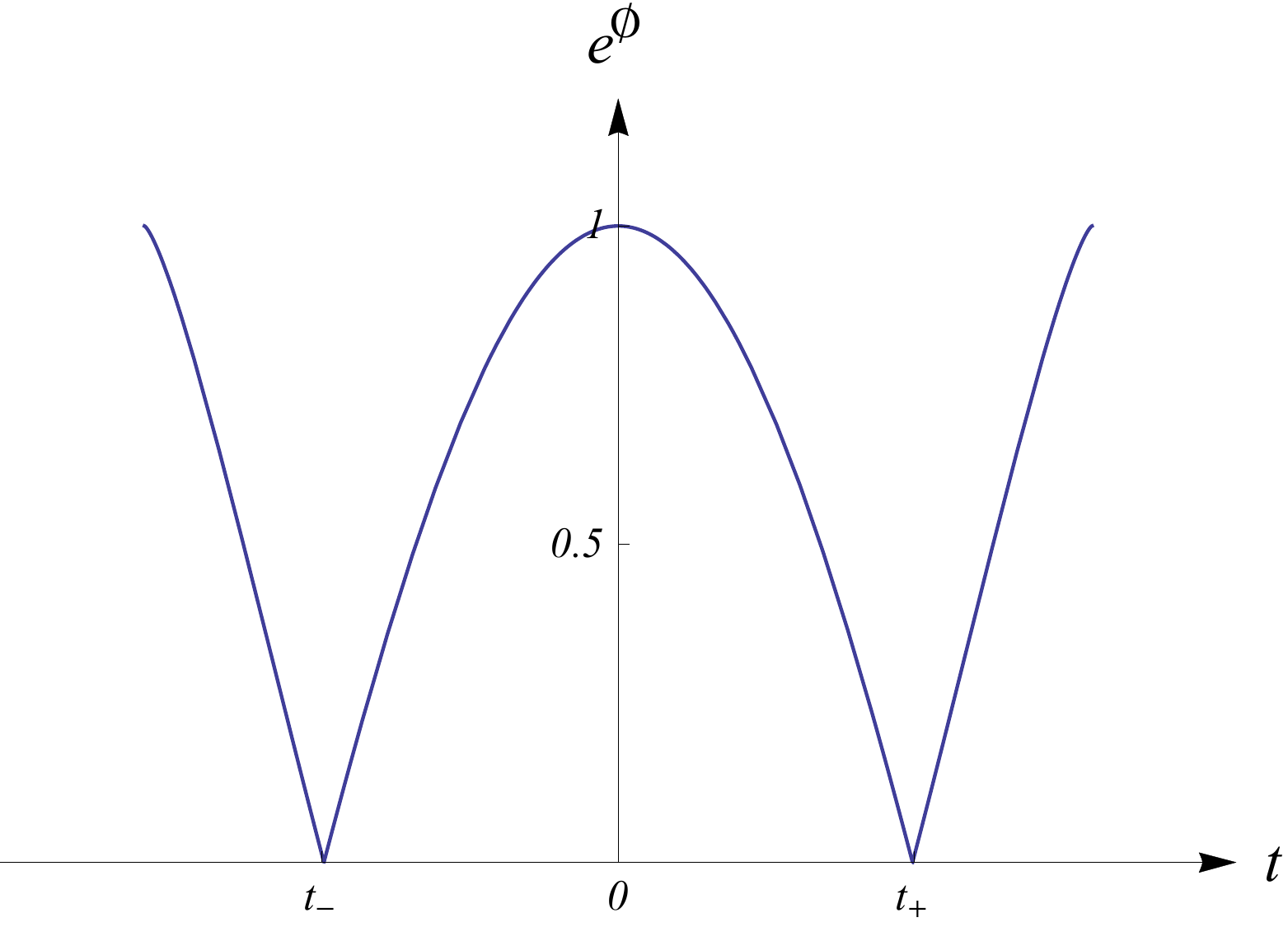}
\caption{$e^{\phi}$ as a function of $t$ for a symplectic half-flat $SU(3)$ structure manifold $X(t)$ with vanishing $\gamma$. At the points $t_{\pm}$, $X(t)$ has vanishing volume and curvature.}
\label{fig:shfu}
\end{figure}

\subsubsection{$\gamma \neq 0$}

A non-zero $\gamma$ makes the equations \eqref{eq:sympcst6} more intricate. Indeed, without recourse to the analysis of explicit examples, which goes beyond the scope of this paper, only a few qualitative observations can be made. First, the second equation in \eqref{eq:sympcst6} clarifies that it is only when $\gamma$ is non-closed, that ${\rm Re}W_2$ can vanish at some point $t=t_*$ along the flow, without $\omega$ also going to zero. This is in agreement with the first order analysis of the flow away from Calabi--Yau, where we found that a non-closed $\gamma$ was required to generate $\delta_1 W_2$. In contrast to the nearly K\"ahler case discussed in section \ref{sec:NK}, there are no direct obstructions to $t_*$ being finite.\footnote{ However, to answer this question in ernest requires the analysis of example manifolds with  ${\rm dim}(H^{(2,1)}\oplus H^{(1,2)})^{\rm prim} \neq 0$. On such spaces, $W_2$ can be explicitely computed, and various ans\"atze for $\gamma$ can be tested. We hope to come back to such an analysis in the future.}

Second, we recall \eqref{eq:negfunct} which shows that symplectic half-flat flows require a non-constant $\phi(t)$ at least for some range of $t$, in order not to fall back into a Calabi--Yau flow for all $t$. This constraint requires that $e^{\frac{3}{2} \phi}$ is a concave function of $t$. As such, it will cross the $t$-axis at at least one point, and at these point the $SU(3)$ structure of $X(t)$ is singular, since its volume goes to zero.

Finally, it should be mentioned that it is still allowed that $e^{\frac{3}{2} \phi}$ asymptotes to a linear, or even a constant function for some $t$. Concretely, suppose that the strict inequality \eqref{eq:negfunct} holds for $t < t_*$, but is saturated after $t_*$: the symplectic half-flat manifold would then flow to a Calabi--Yau manifold without flux, with volume proportional to $e^{\phi}$. However, to keep the dilaton evolution smooth, it is likely that the transition will be asymptotic, and the point $t_*$ will be at infinity. We thus conclude that non-singular Calabi--Yau loci at finite $t$ require the presence of flux, and in particular a non-vanishing $\gamma$.
\vspace{0.5cm}

\section{Conclusions}
\label{sec:conclusion}
In this paper, we have been concerned with heterotic string compactifications to a four-dimensional non-compact space that is is not maximally symmetric. To be more precise, we have focused on four-dimensional half-BPS domain wall compactifications. We saw that by allowing for a non-maximally symmetric spacetime, the internal six-dimensional space is in general torsional, even at zeroth order in $\alpha'$. The compact geometry is, however, always a manifold with $SU(3)$ structure.

This paper is a continuation of an ongoing program that strives to get a better understanding of heterotic string compactifications to non-maximally symmetric spacetimes \cite{Lukas:2010mf, Klaput:2011mz, Gray:2012md, Klaput:2012vv, Klaput:2013nla, Larfors:2013zva, Maxfield:2014wea, Gurrieri:2004dt,Micu:2004tz,deCarlos:2005kh,Gurrieri:2007jg,Micu:2009ci,Haupt:2014ufa,Held:2010az,Held:2011uz,Chatzistavrakidis:2012qb}. The physical motivation for these studies is that such compactifications might, after inclusion of non-perturbative effects, complete to phenomenologically relevant four-dimensional models. Mathematically, one goal of this program is to get a better understanding of the moduli space of the $SU(3)$ structures that arise in flux compactifications. This is a very hard question in general. For instance, the usual cohomological tools available for the analysis of the moduli space of Calabi--Yau compactifications no longer apply, or at least require modifications, due to the torsional nature of the compact space.

By focusing on domain wall compactifications, the  $SU(3)$ structure can be embedded into a non-compact seven-dimensional $G_2$ structure. Such an embedding corresponds to a flow along a path in the $SU(3)$ structure moduli space, thus facilitating the analysis of some of the features of this space. In this paper, we investigated how restrictive the embedding is when selecting these paths. For example, is it possible to flow from one $SU(3)$ structure with a certain type of torsion classes, to another with different torsion classes turned on? By considering both $SU(3)$ structure preserving and changing flows, we showed by examples that the answer to this question is yes. In particular, we derived the constraints for a $G_2$ embedding to preserve $SU(3)$ structures of Calabi--Yau and nearly K\"ahler type. In both cases, we found that restrictions only apply to a function $N_t$, which encodes warping in the $G_2$ metric, and a primitive (2,1)-form $\gamma$, that is part of the $H$-flux in the configuration. Through a perturbative analysis of the Calabi--Yau flow to linear order, we showed  that if the constraints on $N_t$ and $\gamma$ are violated, torsion classes are switched on by the flow. 

Additionally, we have analysed the flows of nearly K\"ahler and half-flat $SU(3)$ structures, and provided evidence that both contain Calabi--Yau loci. For nearly K\"ahler manifolds, where we solved the flow equations explicitly, these loci are necessarily in the limit where $t$ goes to infinity. For half-flat manifolds, we found support for the existence of Calabi--Yau loci also at finite distance, but due to the complexity of the equations in this case, we cannot present a definite answer to this question. Nevertheless, this does not render the analysis void, as we have discovered attributes to the flow, necessary for consistent non-Calabi--Yau solutions. For example, it seems that at least one singularity or ``domain wall'' is required along the flow. It would be interesting to find a physical interpretation of these walls, possibly in terms of NS5-branes and Kaluza--Klein monopoles.

From a physics perspective, the relaxation of constraints on the compact geometry can be understood since there is less supersymmetry, and hence more freedom, in the compactification compared with the maximally symmetric case. Still, the $G_2$ structure embedding is under enough control to provide a fertile ground for the study of string compactifications and the moduli space of $SU(3)$ structures in general. Our study has uncovered some features of these spaces, but many intriguing questions remain for future studies.

These include further investigations of the integrability of flows that connect Calabi--Yau and non-Calabi--Yau manifolds, for example by extending our perturbative analysis to higher orders. It would also be interesting to study the flow of specific examples of $SU(3)$ structure manifolds with known torsion classes. Our results indicate that symplectic half-flat manifolds are particularly relevant, as they are simple but might still flow to Calabi--Yau manifolds at finite distance. A study of the flow of smooth compact toric varieties with $SU(3)$ structure, that have been constructed in \cite{Larfors:2010wb,Larfors:2011zz,Larfors:2013zva}, might also be rewarding.

Another area to explore is whether the torsion-changing flows we find also change the topology of the six-dimensional manifold. From our analysis, there are no indications of this: what seems to happen is that different $SU(3)$ structures on the same topological space flow into each other. However, many flows contain singularities where the volume of the compact space goes to zero, and where the scalar curvature either goes to zero or infinity. In this paper, we have refrained from studying the flow through such points, but we cannot avoid noting the similarities with conifold points in Calabi--Yau moduli spaces, and how they smoothly connect different moduli spaces \cite{Candelas:1989js,Strominger:1995cz,Chen:2010ssa}. Topology-changing Calabi--Yau flows have recently been shown to be relevant in the study of two-dimensional supersymmetric theories  \cite{Gadde:2013sca}. The flows studied here might find similar applications.

In the present paper, we have analysed heterotic compactifications at zeroth order in $\alpha'$. It is known that effects at linear order in $\alpha'$ can change the flow quite drastically for certain $SU(3)$ structure manifolds, see e.g. \cite{Klaput:2012vv,Haupt:2014ufa}. Indeed, such effects couple the gauge and gravity sectors of heterotic supergravity, leading to a more intricate mathematical problem. It would be interesting to see how these effects impact on the findings of the present paper. A relatively straightforward study would be to extend the analysis of \cite{Klaput:2012vv,Haupt:2014ufa} to the phase-changing flow of nearly K\"ahler manifolds that we discussed in section \ref{sec:NK}.

Furthermore, in order to properly include the gauge sector, a modified  description of the heterotic configuration might be required. In \cite{delaOssa:2014cia, Anderson:2014xha,delaOssa:2014msa} it was shown that maximally symmetric heterotic compactifications to six dimensions can be encoded in a holomorphic structure on a certain extension bundle. This structure also relates to the generalised geometric description of the heterotic string \cite{Garcia-Fernandez:2013gja, Baraglia:2013wua, Bedoya:2014pma, Coimbra:2014qaa}. Moreover, by rephrasing the whole system in terms of a holomorphic structure, it was possible to compute the first order moduli space in the maximally symmetric case. This is done by means of cohomologies of the holomorphic structures involved. 
It may be that similar structures are required when considering the $\alpha'$ corrected domain wall compactifications. However, since in this case the base space of the bundle is non-complex in general, these structures will no longer be holomorphic. The framework would therefore have to be generalised to apply to the domain wall case. Hints for such a generalisation might be found in recent studies of cohomologies for flux compactifications of type II supergravities \cite{Tseng:2011gv}, and it would be interesting to study this in more detail in the future.

\section*{Acknowledgements}
The authors thank Thomas Bruun Madsen, Philip Candelas, Diego Conti, Tudor Dimofte, I\~naki Garcia-Extebarria, James Gray, Spiro Karigiannis and Simon Salamon for interesting discussions over the course of this project. XD would like to acknowledge the hospitality of Perimeter Institute where part of this project was undertaken. ML would like to thank the University of Oxford for hospitality during the completion of this work. XD's research was supported in part by the EPSRC grant BKRWDM00. The research of ML was supported by the Swedish Research Council (VR) under the contract 623-2011-7205.  ES was supported by the Clarendon Scholarship of OUP, and a Balliol College Dervorguilla Scholarship. Part of ES's work was also made in the ILP LABEX (under reference ANR-10-LABX-63), and was supported by French state funds managed by the ANR within the Investissements d'Avenir programme under reference ANR-11-IDEX-0004-02.

\newpage
\begin{appendix}

\section{Conventions and identities}
\label{ap:conv}

In this appendix we collect our conventions and some useful identities. Our conventions conform with those of \cite{huybrechts}.

\vskip10pt
\noindent\underline{Index conventions}: we use latin indices $m,n,...$ for real indices, and greek indices $\mu, \nu,...$ for holomorphic indices. 

\noindent\underline{Torsion}: 
\be 
T(\nabla) \cdot (X\wedge Y) = \nabla_X Y - \nabla_Y X - [X,Y]\quad\hbox{or, equivalently,}\quad
T(\nabla)_{mn}{}^p = 2\, \Gamma_{[mn]}{}^p~. \nn
\ee

\subsection{Differential forms and wedge products}
A $p$-form $\alpha$ is defined by
\be \alpha = \frac{1}{p!} \alpha_{m_1...m_p} \dd x^{m_1} \w ...\w \dd x^{m_p} \; . \ee
The wedge product between a $p$-form $\alpha$ and a $q$-form $\beta$ is 
\be \alpha \w \beta = \frac{1}{p!\, q!} \alpha_{m_1...m_p} \beta_{n_1...n_q}\dd x^{m_1} \w ...\w \dd x^{m_p} \w \dd x^{n_1} \w ...\w \dd x^{n_q}\; . \ee

\subsection{Hodge operator and the adjoint operator \texorpdfstring{$\dd^\dagger$}\ }

Let $\alpha_p$ be a $p$-form.  The Hodge $*$ operator on $\alpha_p$ for a $d$-dimensional manifold is defined by
\be
**\alpha_p = (-1)^{p(d-p)}\, \alpha_p~.\ee
Then if $d$ is even,
\[
**\alpha_p = (-1)^p\, \alpha_p~,\]
and if $d$ is odd
\[
**\alpha_p =  \alpha_p~.\]
In general, 
\be *\alpha_p = \frac{|\det g|^{1/2}}{(d - p)!\, p!}\ 
\alpha_{p \ m_1\ldots m_p} \epsilon^{m_1\ldots m_p}{}_{m_{p+1}\ldots m_{d}}\
{\rm d}x^{m_{p+1}}\wedge\cdots\wedge{\rm d}x^{m_{d}}.\ee 
The adjoint operator $\dd^\dagger$ on $\alpha_p$ is defined by
\be \dd^\dagger\alpha_p = (-1)^{dp + d + 1}\, *\dd *\alpha_p~.\ee
If $d$ is even, then
\[ \dd^\dagger\alpha_p = - *\dd *\alpha_p~,\]
and if $d$ is odd
\[\dd^\dagger\alpha_p = (-1)^p\, *\dd *\alpha_p~.\]

\subsection{Contraction operator \texorpdfstring{$\lrcorner$}\ }

The operator $\lrcorner$ denotes the contraction of forms, and is defined by
\be \alpha\lrcorner\beta = 
\frac{1}{k!\, p!}\ \alpha^{m_1\cdots m_k}\ \beta_{m_1\cdots m_k n_1\cdots n_p}
\dd x^{n_1}\cdots \dd x^{n_p}~,  \ee
where $\alpha$ is any $k$-form and $\beta$ is any $p+k$-form. It is easy to deduce the identity
\be \alpha\lrcorner \beta = (-1)^{p(d-p-k)}\, *(\alpha\wedge*\beta)~.\ee
\vskip10pt

\subsection{Conventions on almost complex manifolds}

\underline{Real and complex coordinates}: We define complex coordinates 
\[\zeta^\mu = x^{2\mu-1} + i\, x^{2\mu} \; ,
\mbox{where} \; \mu=1,...n , \]
and $2n =d$ is the real dimension of the manifold.

\vskip10pt
\noindent\underline{Hermitian form}:  
$$\omega(X,Y) = g(JX, Y)\quad\hbox{ or, equivalently,}\quad \omega_{mn} = J_m{}^p \ g_{pn}.$$
In complex coordinates $\omega_{\mu\bar\nu} = i\, g_{\mu\bar\nu}$. We have the identity
\be
{\rm d}^{\dagger}\omega = - |\det g|^{-1/2}\, \partial_p(|\det g|^{1/2} \omega^{pq})\, g _{qm}\, {\rm d} x^m= - \nabla^{LC}_n J_m{}^n \, {\rm d}x^m \; .\ee

\vskip10pt
\noindent\underline{Almost complex structure and projection operators}
\vskip10pt

\noindent Let $\alpha$ be a $(p,q)$-form.  The almost complex structure $J$ acts on a $(p,q)$ form $\alpha$ as follows
\be J(\alpha) = J_{m_1}{}^{n_1} ... J_{m_{p+q}}{}^{n_{p+q}} \alpha_{n_1 ... n_{p+q}} =  i^{p-q}\, \alpha~.\ee
The projection operators
\be P = \frac{1}{2}\, (1 - i J)~,\qquad Q = \frac{1}{2}\, (1 + i J)~.\ee
project a 1-form onto its holomorphic and antiholomorphic components, respectively.

\vskip10pt
\noindent\underline{Volume form}
\vskip -5pt
\begin{align}
\det g &= \det g_{mn} = - 2^{2n}\, (\det g_{\mu\bar\nu})^2~,\nn\\[3pt]
\dd{\rm vol}_X 
&= \sqrt{|\det g|}\ \dd x^1\wedge\cdots\wedge\dd x^d\nn\\
&= \frac{1}{d!}\  \sqrt{|\det g|}\ \epsilon_{m_1\cdots m_d}\, 
\dd x^{m_1}\wedge \cdots\wedge \dd x^{m_d}\nn\\[5pt]
& = \frac{i}{2^n}\, \sqrt{|\det g|}\ \dd \zeta^1\wedge\cdots\wedge\dd\zeta^n\wedge
\dd\zeta^{\bar 1}\wedge\cdots\wedge\dd\zeta^{\bar n}~, \nn \\
\dd{\rm vol}_X&= \frac{1}{n!}\, \underbrace{\omega\wedge\omega\wedge ... \w\omega}_\text{$n$ times} ~.\nn
\end{align}

\vskip10pt
\noindent\underline{Hodge $*$ operator}

Let $\alpha_k$ be a primitive $k$-form.  Then, with $j \le n+k$,
\be *(\omega^j\wedge\alpha_k) = (-1)^{k(k+1)/2}\, \frac{j!}{(n-k-j)!}\,  
\omega^{n-k-j}\wedge J(\alpha_k)~.\ee

Let $\alpha$ be a one form.  Then
\begin{equation}
\omega\lrcorner \dd\alpha = \dd^\dagger(\alpha\lrcorner\omega) + \alpha\lrcorner \dd^\dagger\omega
= - \dd^\dagger(J(\alpha)) + \alpha\lrcorner \dd^\dagger\omega~.\label{idone}
\end{equation}
This identity is a ``non-K\"ahler identity'' and reduces to one of the  K\"ahler identities when $\omega$ is $\dd$-closed.

\subsubsection*{Identitites in 6 dimensions}

\vskip10pt
\noindent\underline{Holomorphic volume form} 

\noindent Any almost complex 6-manifold has a nowhere-vanishing holomorphic top form, defined by
\bea \nn
\Psi &= \frac{1}{3!}\, \Psi_{\mu\nu\rho}\, \dd\zeta^\mu\wedge\dd\zeta^\nu\wedge\dd\zeta^\rho = 
\frac{1}{3!}\, f\, \epsilon_{\mu\nu\rho}\, \dd\zeta^\mu\wedge\dd\zeta^\nu\wedge\dd\zeta^\rho\\
&=
f \,\dd\zeta^1\wedge\dd\zeta^2\wedge\dd\zeta^3~.\nn
\eea
The following identities hold
\begin{align}
*\Psi &= -i\, \Psi~. \\
i\, \Psi\wedge\overline\Psi &= ||\Psi||^2\, \dd{\rm vol}_6 ~, \qquad
||\Psi||^2 = \frac{1}{3!}\, \overline\Psi^{\mu\nu\rho}\, \Psi_{\mu\nu\rho}
=  8 |f|^2 \, (\det g)^{-1/2}~.\nn\\[5pt]
\Psi_{mnp}\, \overline\Psi^{mnq} &= 2\, ||\Psi||^2\, P_p{}^q~, \qquad
\Psi_{mpq}\, \overline\Psi^{mrs} = 2\, ||\Psi||^2\, P_{[p}{}^r\, P_{q]}{}^s~.\nn
\end{align}
Note that $\omega \w \Psi = 0$ ; together $\omega$ and $\Psi$ determine an $SU(3)$ structure. We also have
 \[ *{\rm Re}(\alpha\Psi) = {\rm Im}(\alpha\Psi)~,\qquad
*{\rm Im}(\alpha\Psi) = - {\rm Re}(\alpha\Psi)~.\]

\vskip10pt
\noindent\underline{Useful identies} 

Let $\alpha$ be a $k$-form
\begin{equation}
\begin{split}
\omega\lrcorner(\omega\wedge\alpha) &= (3-k)\, \alpha + \omega\wedge(\omega\lrcorner\alpha)~.\\
(\omega\wedge\omega)\lrcorner (\alpha\wedge\omega) &= 
2(4-k)\, \omega\lrcorner\alpha +  \omega\wedge ((\omega\wedge\omega)\lrcorner\alpha)~,\qquad
k\ge 2~.
\end{split}
\end{equation}

\vskip10pt
Let $\beta$ be a 1-form
\begin{equation}
\begin{split}
\omega\lrcorner (\beta\wedge\Psi) &= - J(\beta)\lrcorner\Psi~.\\
\overline{\Psi}\lrcorner (\beta\wedge\Psi) &= - (\beta\lrcorner\Psi)\lrcorner\overline\Psi = - ||\Psi||^2\, \beta_m\, Q_n{}^m\, \dd x^n
= - ||\Psi||^2\, \beta^{(0,1)}~.
\end{split}
\end{equation}

\vskip10pt
Let $\beta$ be a 2-form
\begin{equation}
\begin{split}
\omega\lrcorner(\beta\wedge\Psi) &=
- (\omega\lrcorner\beta)\, \Psi
-\frac{1}{2}\, \omega^{mn}\, \beta_{mp}\, \Psi_{qrn}\, \dd x^p\wedge\dd x^q\wedge\dd x^r~.\\
\Psi\lrcorner(\beta\wedge\omega) &= -i\, \beta\lrcorner\Psi~. \\
\overline\Psi\lrcorner(\beta\wedge\Psi) 
&= \frac{1}{2}\, ||\Psi||^2\, \beta_{mn}\, Q_p{}^m\, Q_q{}^n\, \dd x^p\wedge\dd x^q
= ||\Psi||^2\, \beta^{(0,2)}~.
\end{split}
\end{equation}

\vskip10pt
Let $\beta$ be a 3-form
\begin{equation}
\begin{split}
\overline\Psi\lrcorner(\beta\wedge\Psi) &=   
- \frac{1}{3!}\, ||\Psi||^2\, \beta_{mnp}\, Q_r{}^m\, Q_s{}^n\,Q_q{}^
p\, \dd x^r\dd x^s\wedge\dd x^q = - ||\Psi||^2\, \beta^{(0,3)}~,\\
\overline\Psi\lrcorner (\beta\wedge\omega) &= (\overline\Psi\lrcorner\beta)\, \omega 
+ \textstyle{\frac{i}{2}}\, g_{rp}\, \overline\Psi^{\, mnr}\, \beta_{mnq}\, \dd x^p\wedge\dd x^q~.
\end{split}
\end{equation}

\vskip10pt
\noindent\underline{Identities involving a closed holomorphic $(3,0)$-form $\Psi$}

\vskip10pt
Let $f$ be a function
\begin{equation}
\dd f \lrcorner\Psi = *(\dd f\wedge *\Psi) = *\dd(f*\Psi) = - \dd^\dagger(f\Psi)~,
\label{eq:idtwo}
\end{equation}

\vskip10pt
Let $\beta$ be a two-form
\begin{equation}
\Psi\lrcorner\dd\beta = - \dd^\dagger(\beta\lrcorner\Psi)~,
\label{eq:idthree}
\end{equation}

\vskip10pt
\noindent\underline{Calabi--Yau identities:} 
On a Calabi--Yau manifold, we have
\be
\dd \omega = 0~, \qquad \dd \Psi = 0~.
\ee

Let $f$ be a function. Then
\be
\dd (\dd f \lrcorner \Psi) \w \omega =
\dd (\bar{\partial} f \lrcorner \Psi \w \omega) \sim
\dd (\Psi \w \bar{\partial} f) = - \Psi \w \dd \bar{\partial} f = 0~.
\ee
\vskip10pt

\noindent\underline{The almost complex structure $J$ in terms of $\Psi$:}
\vskip10pt
The (real part of the) complex 3-form $\Psi$ determines a unique almost complex structure $J$ such that $\Psi$ is a $(3,0)$-form 
with respect to $J$,
\begin{equation}
\label{eq:complexstr}
{J_m}^n=\frac{{I_m}^n}{\sqrt{-\frac{1}{6}{\rm tr} I^2}}~,\qquad
{I_m}^n=(\textrm{Re}\Psi)_{mpq}(\textrm{Re}\Psi)_{rst}\,\epsilon^{npqrst}~,
\qquad J^2=-{\bf 1}~.
\end{equation}

The variations of $J$ in terms of the variations of $\Psi$ are given by
\begin{equation}
\partial_t J = 2\, i\, \Delta_t = 2\, i\, \Delta_t{}_{\, n}{}^m\, \dd x^n\otimes \partial_m~,\qquad
\Delta_t{}_n{}^m = \frac{1}{2 \, ||\Psi||^2} \, \chi_{t\, npq}\, \overline\Psi^{mpq}~.
\label{eq:DeltafromPsi}
\end{equation}

\vskip10pt
\subsection{Conventions on 7-dimensional $G_2$ manifolds}

Let $\varphi$ be the 3-form that determines the $G_2$ structure. Then
\begin{equation}
||\varphi ||^2 = \varphi\lrcorner\,\varphi = \frac{1}{3!}\, \varphi^{a_1a_2a_3}\,\varphi_{a_1a_2a_3} = 7~.
\end{equation}

Let $\beta$ be a $k$-form with
$$\beta = \dd t \wedge \beta_t + \beta^X~,$$
then
\be *_7\,\beta = (-1)^k \, N\wedge *\beta^X + N_t^{-1}\, *\beta_t~.\ee

\end{appendix}


\newpage
\bibliographystyle{JHEP}

\begin{thebibliography}{10}

\bibitem{Candelas:1985en}
P.~Candelas, G.~T. Horowitz, A.~Strominger, and E.~Witten, {\it {Vacuum
  Configurations for Superstrings}},  {\em Nucl.Phys.} {\bf B258} (1985)
  46--74.

\bibitem{Strominger:1986uh}
A.~Strominger, {\it {Superstrings with Torsion}},  {\em Nucl.Phys.} {\bf B274}
  (1986) 253.

\bibitem{Hull:1986kz}
C.~Hull, {\it {Compactifications of the Heterotic Superstring}},  {\em
  Phys.Lett.} {\bf B178} (1986) 357.

\bibitem{Lust:1986ix}
D.~L\"ust, {\it {Compactification of Ten-dimensional Superstring Theories Over
  Ricci Flat Coset Spaces}},  {\em Nucl.Phys.} {\bf B276} (1986) 220.

\bibitem{Dasgupta:1999ss}
K.~Dasgupta, G.~Rajesh, and S.~Sethi, {\it {M theory, orientifolds and G -
  flux}},  {\em JHEP} {\bf 9908} (1999) 023,
  [\href{http://xxx.lanl.gov/abs/hep-th/9908088}{{\tt hep-th/9908088}}].

\bibitem{Ivanov:2000fg}
S.~Ivanov and G.~Papadopoulos, {\it {A No go theorem for string warped
  compactifications}},  {\em Phys.Lett.} {\bf B497} (2001) 309--316,
  [\href{http://xxx.lanl.gov/abs/hep-th/0008232}{{\tt hep-th/0008232}}].

\bibitem{Becker:2002sx}
K.~Becker and K.~Dasgupta, {\it {Heterotic strings with torsion}},  {\em JHEP}
  {\bf 0211} (2002) 006, [\href{http://xxx.lanl.gov/abs/hep-th/0209077}{{\tt
  hep-th/0209077}}].

\bibitem{Becker:2003sh}
K.~Becker, M.~Becker, P.~S. Green, K.~Dasgupta, and E.~Sharpe, {\it
  {Compactifications of heterotic strings on nonKahler complex manifolds. 2.}},
   {\em Nucl.Phys.} {\bf B678} (2004) 19--100,
  [\href{http://xxx.lanl.gov/abs/hep-th/0310058}{{\tt hep-th/0310058}}].

\bibitem{Becker:2003yv}
K.~Becker, M.~Becker, K.~Dasgupta, and P.~S. Green, {\it {Compactifications of
  heterotic theory on nonKahler complex manifolds. 1.}},  {\em JHEP} {\bf 0304}
  (2003) 007, [\href{http://xxx.lanl.gov/abs/hep-th/0301161}{{\tt
  hep-th/0301161}}].

\bibitem{Gauntlett:2002sc}
J.~P. Gauntlett, D.~Martelli, S.~Pakis, and D.~Waldram, {\it {G structures and
  wrapped NS5-branes}},  {\em Commun.Math.Phys.} {\bf 247} (2004) 421--445,
  [\href{http://xxx.lanl.gov/abs/hep-th/0205050}{{\tt hep-th/0205050}}].

\bibitem{Cardoso:2002hd}
G.~Lopes~Cardoso, G.~Curio, G.~Dall'Agata, D.~L\"ust, P.~Manousselis, {\em
  et.~al.}, {\it {NonKahler string backgrounds and their five torsion
  classes}},  {\em Nucl.Phys.} {\bf B652} (2003) 5--34,
  [\href{http://xxx.lanl.gov/abs/hep-th/0211118}{{\tt hep-th/0211118}}].

\bibitem{Gauntlett:2003cy}
J.~P. Gauntlett, D.~Martelli, and D.~Waldram, {\it {Superstrings with intrinsic
  torsion}},  {\em Phys.Rev.} {\bf D69} (2004) 086002,
  [\href{http://xxx.lanl.gov/abs/hep-th/0302158}{{\tt hep-th/0302158}}].

\bibitem{Gran:2005wf}
U.~Gran, P.~Lohrmann, and G.~Papadopoulos, {\it {The Spinorial geometry of
  supersymmetric heterotic string backgrounds}},  {\em JHEP} {\bf 0602} (2006)
  063, [\href{http://xxx.lanl.gov/abs/hep-th/0510176}{{\tt hep-th/0510176}}].

\bibitem{Becker:2006xp}
M.~Becker, L.-S. Tseng, and S.-T. Yau, {\it {Moduli Space of Torsional
  Manifolds}},  {\em Nucl.Phys.} {\bf B786} (2007) 119--134,
  [\href{http://xxx.lanl.gov/abs/hep-th/0612290}{{\tt hep-th/0612290}}].

\bibitem{Gran:2007kh}
U.~Gran, G.~Papadopoulos, and D.~Roest, {\it {Supersymmetric heterotic string
  backgrounds}},  {\em Phys.Lett.} {\bf B656} (2007) 119--126,
  [\href{http://xxx.lanl.gov/abs/0706.4407}{{\tt arXiv:0706.4407}}].

\bibitem{Ivanov:2009rh}
S.~Ivanov, {\it {Heterotic supersymmetry, anomaly cancellation and equations of
  motion}},  {\em Phys.Lett.} {\bf B685} (2010) 190--196,
  [\href{http://xxx.lanl.gov/abs/0908.2927}{{\tt arXiv:0908.2927}}].

\bibitem{Hitchin:2000jd}
N.~J. Hitchin, {\it {The geometry of three-forms in six and seven dimensions}},
   \href{http://xxx.lanl.gov/abs/math/0010054}{{\tt math/0010054}}.

\bibitem{chiossi}
S.~{Chiossi} and S.~{Salamon}, {\it {The intrinsic torsion of SU(3) and G\_2
  structures}},  {\em ArXiv Mathematics e-prints} (Feb., 2002)
  [\href{http://xxx.lanl.gov/abs/math/0202}{{\tt math/0202}}].

\bibitem{Grana:2005jc}
M.~Gra\~na, {\it {Flux compactifications in string theory: A Comprehensive
  review}},  {\em Phys.Rept.} {\bf 423} (2006) 91--158,
  [\href{http://xxx.lanl.gov/abs/hep-th/0509003}{{\tt hep-th/0509003}}].

\bibitem{Blumenhagen:2006ci}
R.~Blumenhagen, B.~Kors, D.~L\"ust, and S.~Stieberger, {\it {Four-dimensional
  String Compactifications with D-Branes, Orientifolds and Fluxes}},  {\em
  Phys.Rept.} {\bf 445} (2007) 1--193,
  [\href{http://xxx.lanl.gov/abs/hep-th/0610327}{{\tt hep-th/0610327}}].

\bibitem{Koerber:2010bx}
P.~Koerber, {\it {Lectures on Generalized Complex Geometry for Physicists}},
  {\em Fortsch.Phys.} {\bf 59} (2011) 169--242,
  [\href{http://xxx.lanl.gov/abs/1006.1536}{{\tt arXiv:1006.1536}}].

\bibitem{Larfors:2013zva}
M.~Larfors, {\it {Revisiting toric SU(3) structures}},  {\em Fortsch.Phys.}
  {\bf 61} (2013) 1031--1055, [\href{http://xxx.lanl.gov/abs/1309.2953}{{\tt
  arXiv:1309.2953}}].

\bibitem{Gray:2012md}
J.~Gray, M.~Larfors, and D.~L\"ust, {\it {Heterotic domain wall solutions and
  SU(3) structure manifolds}},  {\em JHEP} {\bf 1208} (2012) 099,
  [\href{http://xxx.lanl.gov/abs/1205.6208}{{\tt arXiv:1205.6208}}].

\bibitem{Lukas:2010mf}
A.~Lukas and C.~Matti, {\it {G-structures and Domain Walls in Heterotic
  Theories}},  {\em JHEP} {\bf 1101} (2011) 151,
  [\href{http://xxx.lanl.gov/abs/1005.5302}{{\tt arXiv:1005.5302}}].

\bibitem{Klaput:2011mz}
M.~Klaput, A.~Lukas, and C.~Matti, {\it {Bundles over Nearly-Kahler Homogeneous
  Spaces in Heterotic String Theory}},  {\em JHEP} {\bf 1109} (2011) 100,
  [\href{http://xxx.lanl.gov/abs/1107.3573}{{\tt arXiv:1107.3573}}].

\bibitem{Klaput:2012vv}
M.~Klaput, A.~Lukas, C.~Matti, and E.~E. Svanes, {\it {Moduli Stabilising in
  Heterotic Nearly K\"ahler Compactifications}},  {\em JHEP} {\bf 1301} (2013)
  015, [\href{http://xxx.lanl.gov/abs/1210.5933}{{\tt arXiv:1210.5933}}].

\bibitem{Klaput:2013nla}
M.~Klaput, A.~Lukas, and E.~E. Svanes, {\it {Heterotic Calabi-Yau
  Compactifications with Flux}},  {\em JHEP} {\bf 1309} (2013) 034,
  [\href{http://xxx.lanl.gov/abs/1305.0594}{{\tt arXiv:1305.0594}}].

\bibitem{Gurrieri:2004dt}
S.~Gurrieri, A.~Lukas, and A.~Micu, {\it {Heterotic on half-flat}},  {\em
  Phys.Rev.} {\bf D70} (2004) 126009,
  [\href{http://xxx.lanl.gov/abs/hep-th/0408121}{{\tt hep-th/0408121}}].

\bibitem{Micu:2004tz}
A.~Micu, {\it {Heterotic compactifications and nearly-Kahler manifolds}},  {\em
  Phys.Rev.} {\bf D70} (2004) 126002,
  [\href{http://xxx.lanl.gov/abs/hep-th/0409008}{{\tt hep-th/0409008}}].

\bibitem{deCarlos:2005kh}
B.~de~Carlos, S.~Gurrieri, A.~Lukas, and A.~Micu, {\it {Moduli stabilisation in
  heterotic string compactifications}},  {\em JHEP} {\bf 0603} (2006) 005,
  [\href{http://xxx.lanl.gov/abs/hep-th/0507173}{{\tt hep-th/0507173}}].

\bibitem{Gurrieri:2007jg}
S.~Gurrieri, A.~Lukas, and A.~Micu, {\it {Heterotic String Compactifications on
  Half-flat Manifolds. II.}},  {\em JHEP} {\bf 0712} (2007) 081,
  [\href{http://xxx.lanl.gov/abs/0709.1932}{{\tt arXiv:0709.1932}}].

\bibitem{Micu:2009ci}
A.~Micu, {\it {Moduli Stabilisation in Heterotic Models with Standard
  Embedding}},  {\em JHEP} {\bf 1001} (2010) 011,
  [\href{http://xxx.lanl.gov/abs/0911.2311}{{\tt arXiv:0911.2311}}].

\bibitem{Candelas:1990pi}
P.~Candelas and X.~de~la Ossa, {\it {Moduli Space of {Calabi-Yau} Manifolds}},
  {\em Nucl.Phys.} {\bf B355} (1991) 455--481.

\bibitem{Tseng:2009gr}
L.-S. Tseng and S.-T. Yau, {\it {Cohomology and Hodge Theory on Symplectic
  Manifolds. I.}},  {\em J.Diff.Geom.} {\bf 91} (2012) 383--416,
  [\href{http://xxx.lanl.gov/abs/0909.5418}{{\tt arXiv:0909.5418}}].

\bibitem{Tseng:2010kt}
L.-S. Tseng and S.-T. Yau, {\it {Cohomology and Hodge Theory on Symplectic
  Manifolds. II}},  {\em J.Diff.Geom.} {\bf 91} (2012) 417--444,
  [\href{http://xxx.lanl.gov/abs/1011.1250}{{\tt arXiv:1011.1250}}].

\bibitem{Tseng:2011gv}
L.-S. Tseng and S.-T. Yau, {\it {Generalized Cohomologies and Supersymmetry}},
  \href{http://xxx.lanl.gov/abs/1111.6968}{{\tt arXiv:1111.6968}}.

\bibitem{Anderson:2014xha}
L.~B. Anderson, J.~Gray, and E.~Sharpe, {\it {Algebroids, Heterotic Moduli
  Spaces and the Strominger System}},
  \href{http://xxx.lanl.gov/abs/1402.1532}{{\tt arXiv:1402.1532}}.

\bibitem{delaOssa:2014cia}
X.~de~la Ossa and E.~E. Svanes, {\it {Holomorphic Bundles and the Moduli Space
  of N=1 Supersymmetric Heterotic Compactifications}},
  \href{http://xxx.lanl.gov/abs/1402.1725}{{\tt arXiv:1402.1725}}.

\bibitem{OKS}
X.~de~la Ossa, S.~Karigiannis, and E.~E. Svanes To appear.

\bibitem{FerGray82}
M.~Fernandez and A.~Gray, {\it Riemannian manifolds with structure group g2},
  {\em Ann. Mat. Pura Appl.} {\bf 32} (1982) 19--45.

\bibitem{Firedrich:2003}
T.~{Friedrich} and S.~{Ivanov}, {\it {Killing spinor equations in dimension 7
  and geometry of integrable G$_{2}$-manifolds}},  {\em Journal of Geometry and
  Physics} {\bf 48} (Oct., 2003) 1--11,
  [\href{http://xxx.lanl.gov/abs/math/0112201}{{\tt math/0112201}}].

\bibitem{Martelli:2010jx}
D.~Martelli and J.~Sparks, {\it {Non-Kahler heterotic rotations}},  {\em
  Adv.Theor.Math.Phys.} {\bf 15} (2011) 131--174,
  [\href{http://xxx.lanl.gov/abs/1010.4031}{{\tt arXiv:1010.4031}}].

\bibitem{Gray:1980fk}
A.~Gray and L.~Hervella, {\it The sixteen classes of almost hermitian manifolds
  and their linear invariants},  {\em Annali di Matematica Pura ed Applicata}
  {\bf 123} (1980), no.~1 35--58.

\bibitem{Hitchin01stableforms}
N.~Hitchin, {\it Stable forms and special metrics},  in {\em Global
  differential geometry: the mathematical legacy of Alfred Gray (Bilbao, 2000),
  number 288 in Contemp. Math}, pp.~70--89, 2001.

\bibitem{MR0112154}
K.~Kodaira and D.~C. Spencer, {\it On deformations of complex analytic
  structures. {I}, {II}},  {\em Ann. of Math. (2)} {\bf 67} (1958) 328--466.

\bibitem{0128.16902}
K.~Kodaira and D.~Spencer, {\it {On deformations of complex analytic
  structures. III: Stability theorems for complex structures.}},  {\em Ann.
  Math. (2)} {\bf 71} (1960) 43--76.

\bibitem{tian86}
G.~Tian, {\it Smoothness of the universal deformation space of compact
  calabi-yau manifolds and its petersson-weil metric},  in {\em Mathematical
  aspects of string theory (San Diego, Calif., 1986)}, Adv. Ser. Math. Phys.,
  pp.~629--646, World Sci. Publishing, Singapore, 1986.

\bibitem{Haupt:2014ufa}
A.~S. Haupt, O.~Lechtenfeld, and E.~T. Musaev, {\it {Order $\alpha'$ heterotic
  domain walls with warped nearly K\"ahler geometry}},
  \href{http://xxx.lanl.gov/abs/1409.0548}{{\tt arXiv:1409.0548}}.

\bibitem{Bedulli2007}
L.~{Bedulli} and L.~{Vezzoni}, {\it {The Ricci tensor of SU(3)-manifolds}},
  {\em Journal of Geometry and Physics} {\bf 57} (Mar., 2007) 1125--1146,
  [\href{http://xxx.lanl.gov/abs/math/0606}{{\tt math/0606}}].

\bibitem{Maxfield:2014wea}
T.~Maxfield and S.~Sethi, {\it {Domain Walls, Triples and Acceleration}},  {\em
  JHEP} {\bf 1408} (2014) 066, [\href{http://xxx.lanl.gov/abs/1404.2564}{{\tt
  arXiv:1404.2564}}].

\bibitem{Held:2010az}
J.~Held, D.~L\"ust, F.~Marchesano, and L.~Martucci, {\it {DWSB in heterotic flux
  compactifications}},  {\em JHEP} {\bf 1006} (2010) 090,
  [\href{http://xxx.lanl.gov/abs/1004.0867}{{\tt arXiv:1004.0867}}].

\bibitem{Held:2011uz}
J.~Held, {\it {BPS-like potential for compactifications of heterotic
  M-theory?}},  {\em JHEP} {\bf 1110} (2011) 136,
  [\href{http://xxx.lanl.gov/abs/1109.1974}{{\tt arXiv:1109.1974}}].
  
\bibitem{Chatzistavrakidis:2012qb}
  A.~Chatzistavrakidis, O.~Lechtenfeld and A.~D.~Popov,
  {\it {Nearly K\"ahler heterotic compactifications with fermion condensates}},  {\em JHEP} {\bf 1204} (2012) 114
  [\href{http://xxx.lanl.gov/abs/1202.1278}{{\tt arXiv:1202.1278}}].  

\bibitem{Larfors:2010wb}
M.~Larfors, D.~L\"ust, and D.~Tsimpis, {\it {Flux compactification on smooth,
  compact three-dimensional toric varieties}},  {\em JHEP} {\bf 1007} (2010)
  073, [\href{http://xxx.lanl.gov/abs/1005.2194}{{\tt arXiv:1005.2194}}].

\bibitem{Larfors:2011zz}
M.~Larfors, {\it {Flux compactifications on toric varieties}},  {\em
  Fortsch.Phys.} {\bf 59} (2011) 730--733.

\bibitem{Candelas:1989js}
P.~Candelas and X.~C. de~la Ossa, {\it {Comments on Conifolds}},  {\em
  Nucl.Phys.} {\bf B342} (1990) 246--268.

\bibitem{Strominger:1995cz}
A.~Strominger, {\it {Massless black holes and conifolds in string theory}},
  {\em Nucl.Phys.} {\bf B451} (1995) 96--108,
  [\href{http://xxx.lanl.gov/abs/hep-th/9504090}{{\tt hep-th/9504090}}].

\bibitem{Chen:2010ssa}
F.~Chen, K.~Dasgupta, P.~Franche, and R.~Tatar, {\it {Toward the Gravity Dual
  of Heterotic Small Instantons}},  {\em Phys.Rev.} {\bf D83} (2011) 046006,
  [\href{http://xxx.lanl.gov/abs/1010.5509}{{\tt arXiv:1010.5509}}].

\bibitem{Gadde:2013sca}
A.~Gadde, S.~Gukov, and P.~Putrov, {\it {Fivebranes and 4-manifolds}},
  \href{http://xxx.lanl.gov/abs/1306.4320}{{\tt arXiv:1306.4320}}.

\bibitem{delaOssa:2014msa}
X.~de~la Ossa and E.~E. Svanes, {\it {Connections, Field Redefinitions and
  Heterotic Supergravity}},  \href{http://xxx.lanl.gov/abs/1409.3347}{{\tt
  arXiv:1409.3347}}.

\bibitem{Garcia-Fernandez:2013gja}
M.~Garcia-Fernandez, {\it {Torsion-free generalized connections and Heterotic
  Supergravity}},  \href{http://xxx.lanl.gov/abs/1304.4294}{{\tt
  arXiv:1304.4294}}.

\bibitem{Baraglia:2013wua}
D.~Baraglia and P.~Hekmati, {\it {Transitive Courant Algebroids, String
  Structures and T-duality}},  \href{http://xxx.lanl.gov/abs/1308.5159}{{\tt
  arXiv:1308.5159}}.

\bibitem{Bedoya:2014pma}
O.~A. Bedoya, D.~Marques, and C.~Nunez, {\it {Heterotic $\alpha$'-corrections
  in Double Field Theory}},  \href{http://xxx.lanl.gov/abs/1407.0365}{{\tt
  arXiv:1407.0365}}.

\bibitem{Coimbra:2014qaa}
A.~Coimbra, R.~Minasian, H.~Triendl, and D.~Waldram, {\it {Generalised geometry
  for string corrections}},  \href{http://xxx.lanl.gov/abs/1407.7542}{{\tt
  arXiv:1407.7542}}.

\bibitem{huybrechts}
D.~Huybrechts, {\em Complex Geometry An Introduction}.
\newblock Springer, 2005.

\end{thebibliography}

\providecommand{\href}[2]{#2}\begingroup\raggedright\endgroup

\end{document}